\documentclass{theoretics}

\title{Upper and Lower Bounds on the Smoothed Complexity of the Simplex Method}

\ThCSauthor[CNRS,UCA]{Sophie Huiberts}{sophie@huiberts.me}[0000-0003-2633-014X]
\ThCSauthor[Washington,Microsoft]{Yin Tat Lee}{yintat@uw.edu}[]
\ThCSauthor[Washington]{Xinzhi Zhang}{xinzhi20@cs.washington.edu}[]

\ThCSaffil[CNRS]{CNRS / LIMOS, Clermont-Ferrand, France}
\ThCSaffil[UCA]{Université Clermont Auvergne, Clermont-Ferrand, France}
\ThCSaffil[Washington]{University of Washington, Seattle, USA}
\ThCSaffil[Microsoft]{Microsoft Research, Redmond, USA}

\ThCSthanks{A preliminary version of this article appeared in STOC 2023 \cite{HuibertsLZ23}. \newline
  Sophie Huiberts' work was supported by grants from the Simons Foundation (966674) and ANR (ANR-24-CE48-2762). Xinzhi Zhang's research supported by  NSF grant CCF-1813135, NSF CAREER Award and Packard Fellowships.}
\ThCSshortnames{S.\ Huiberts, Y.T.\ Lee, X.\ Zhang}  
\ThCSshorttitle{Upper and Lower Bounds on the Smoothed Complexity of the Simplex Method}
\ThCSyear{2025}
\ThCSarticlenum{23}
\ThCSreceived{May 16, 2024}
\ThCSrevised{Apr 30, 2025}
\ThCSaccepted{Sep 1, 2025}
\ThCSpublished{Oct 15, 2025}
\ThCSdoicreatedtrue
\ThCSkeywords{simplex method, linear programming, smoothed analysis, beyond worst case analysis, extended formulation}

\usepackage{mathtools}
\usepackage{hyperref,cleveref}

\usepackage{nicefrac}
\usepackage{makecell}
\usepackage{tikz}

\newcommand{\R}{\mathbb R}
\newcommand{\E}{\mathbb E}
\newcommand{\N}{\mathbb N}
\newcommand{\Z}{\mathbb Z}

\newcommand{\T}{\top}
\newcommand{\sfe}{\mathbb{S}^{d-1}}

\newcommand{\aff}{\operatorname{affhull}}
\newcommand{\eps}{\varepsilon} % bestsilon
\renewcommand{\epsilon}{\varepsilon} % bestsilon

\newcommand{\conv}{\operatorname{conv}}
\newcommand{\edges}{\operatorname{edges}}
\newcommand{\verts}{\operatorname{vertices}}

\newcommand{\poly}{\operatorname{poly}}

\newcommand{\norm}[1]{\left\|#1 \right\|}
\newcommand{\dist}{\operatorname{dist}}

\newcommand{\vol}{\operatorname{vol}}
\newcommand{\ball}{\mathbb{B}_2}
\DeclarePairedDelimiter\abs{\lvert}{\rvert}

\definecolor{b2}{RGB}{51,153,255}
\definecolor{red}{RGB}{255,153,51}

\renewcommand{\subset}{\subseteq}

\allowdisplaybreaks

\crefname{fact}{Fact}{Facts}
\crefname{claim}{Claim}{Claims}
\crefname{sidefigure}{Figure}{Figures}

\begin{document}

\maketitle
% \author{Sophie Huiberts}
% \affiliation{
% \institution{Columbia University}
% \country{USA}
% }
% \email{sophie@huiberts.me}
% \author{Yin Tat Lee}
% \affiliation{
% \institution{Microsoft Research}
% \country{USA}
% }
% \email{yintat@uw.edu}
% \author{Xinzhi Zhang}
% \affiliation{
% \institution{University of Washington}
% \country{USA}
% }
% \email{xinzhi20@cs.washington.edu}

\begin{abstract}
The simplex method for linear programming is known to be highly efficient in practice,
and understanding its performance from a theoretical perspective is an active research topic.
The framework of smoothed analysis, first introduced by Spielman and Teng (JACM '04)
for this purpose, defines the smoothed complexity of solving a linear program
with $d$ variables and $n$ constraints as the expected running time when
Gaussian noise of variance $\sigma^2$ is added to the LP data.
We prove that the smoothed complexity of the simplex method is
$O(\sigma^{-3/2} d^{13/4}\log^{5/4} n)$,
improving the dependence on $1/\sigma$ compared to the previous bound of
$O(\sigma^{-2} d^2\sqrt{\log n})$.
We accomplish this through a new analysis of the \emph{shadow bound},
key to earlier analyses as well.
Illustrating the power of our new approach, we moreover prove a
nearly tight upper bound on the smoothed complexity of two-dimensional
polygons.

We also establish the first non-trivial lower bound on the smoothed complexity
of the simplex method, proving that the \emph{shadow vertex simplex method}
requires, with a given auxiliary objective, at least
$\Omega \Big(\min \big(\sigma^{-1/2} d^{-1/2}\log^{-1/4} d,2^d \big) \Big)$
pivot steps with high probability. A key part of our analysis
is a new variation on the extended formulation for the regular $2^k$-gon.
We end with a numerical experiment that suggests our lower bound could
be further improved.
\end{abstract}

% \begin{CCSXML}
% <ccs2012>
%    <concept>
%        <concept_id>10003752.10003809.10003716.10011138.10010041</concept_id>
%        <concept_desc>Theory of computation~Linear programming</concept_desc>
%        <concept_significance>500</concept_significance>
%        </concept>
%  </ccs2012>
% \end{CCSXML}

% \ccsdesc[500]{Theory of computation~Linear programming}

% \keywords{Simplex Algorithm, Linear Programming, Smoothed Complexity}

\section{Introduction}

Introduced by Dantzig \cite{dan51},
the simplex method is one of the primary methods
for solving linear programs (LP's) in practice and is an essential component in many software packages
for combinatorial optimization.
It is a family of local search algorithms which begin by finding a vertex of the set of feasible solutions
and iteratively move to a better neighboring vertex along
the edges of the feasible polyhedron until an optimal solution is reached. These moves are known as \emph{pivot steps}.
Variants of the simplex method can be differentiated by the choice of \emph{pivot rule},
which determines which neighboring vertex is chosen in each iteration,
as well as by the method for obtaining the initial vertex.
Some well-known pivot rules are the most negative reduced cost rule,
the steepest edge rule, and an approximate steepest edge rule known as the
Devex rule.  In theoretical work, the parametric objective rule, also known as
the shadow vertex rule, plays an important role.

Empirical evidence suggests that the simplex algorithm typically takes $O(d + n)$ pivot steps, see
\cite{sha87, andrei2004complexity, Goldfarb1994}.
However, obtaining a rigorous explanation for this excellent
performance has proven challenging.
In contrast to the practical success of the simplex method, all studied variants
are known to have super-polynomial or even exponential worst-case running times.
For deterministic variants, many published bad inputs are based on \emph{deformed cubes},
see \cite{km72, jer73, AC78, GS79, Mur80, gol76} and a unified construction
in \cite{jour/cm/AZ98}.
For randomized and history-dependent variants, bad inputs have been constructed
based on \emph{Markov Decision Processes}
\cite{k92, msw96, conf/stoc/FHZ11, conf/ipco/Friedmann11, hz15, disser2020exponential}.
The fastest provable (randomized) simplex algorithm takes
$O(2^{\sqrt{d\log n}})$ pivot steps in expectation \cite{k92, msw96, hz15}.

Average-case analyses of the simplex method have been performed for a variety of random distributions over linear programs
\cite{b82,b87,b99, jour/mapr/Smale83, report/Haimovich83, jour/mapr/Megiddo86, jour/jc/AKS87, jour/mapr/Todd86, jour/jacm/AM85}.
While insightful, the results from average-case analyses might not be fully realistic due to the fact that ``random''
linear programs tend to have certain properties that ``typical'' linear programs do not.

To better explain why the simplex algorithm performs well in practice, while avoiding some of the pitfalls
of average-case analysis, Spielman and Teng \cite{ST04} introduced the smoothed complexity framework. 
For any base LP
data $\bar{A} \in \R^{n\times d}, \bar{b}\in \R^n, c\in \R^d \backslash \{0\}$
where the rows of $(\bar A,\bar b)$ are normalized to have $\ell_2$ norm at most $1$, they consider
the \emph{smoothed LP} obtained by adding independent Gaussian perturbations to the
constraints:
\begin{align*}
    \max_{x\in \R^d} c^\top x \quad \text{subject to} \quad (\bar{A} + \hat{A}) x \leq (\bar{b} + \hat{b}).
\end{align*}
The entries of $\hat{A}$ and $\hat{b}$ are i.i.d. Gaussian random variables
with mean $0$ and variance $\sigma^2$. The \emph{smoothed complexity} of a simplex
algorithm $\mathcal{A}$ is defined to be the maximum (over $\bar A, \bar b, c$) expected
number of pivot steps the algorithm takes to solve the smoothed LP, i.e.,
\begin{align*}
    \mathcal{SC}_{\mathcal{A},n,d,\sigma} := \max_{\substack{\bar{A} \in \R^{n\times d}, \bar{b}\in \R^n, c\in \R^d \\ \|[\Bar{A},\Bar{b}]\|_{1, 2} \leq 1}} \left(\E_{\hat{A},\hat{b}}\left[ T_{\mathcal{A}}(\bar{A} + \hat{A}, \bar{b} + \hat{b}, c) \right] \right).
\end{align*}
Here $T_{\mathcal{A}}(A,b,c)$ is the number of pivot steps that the algorithm $\mathcal{A}$ takes to solve the linear program $\max_{x\in \R^d} \{c^\top x:  ~ A x \leq b\}$.
We may note that if $\sigma \to \infty$ then $\mathcal{SC}_{\mathcal A, n, d, \sigma}$
approaches the average-case complexity of $\mathcal{A}$ on independent Gaussian distributed input data.
In contrast, if $\sigma \to 0$ then $\mathcal{SC}_{\mathcal A, n, d, \sigma}$ will approach
the worst-case complexity of $\mathcal A$.
As a result, most interest has been directed at understanding the dependence on $\sigma$
in the regime where $\sigma \geq 2^{-\Omega(d)}$ but $\sigma \leq 1/\poly(d)$.

The motivation for smoothed analysis lies in the observation that the above-mentioned
worst-case instances are very ``brittle'' to perturbations, and computer implementations
require great care in handling numerical inaccuracies to obtain the theorized running times
even on problems with a small number of variables. When implemented with a larger number of
variables, the limited accuracy of floating-point numbers make it impossible to reach
the theorized running times.

An algorithm is said to have polynomial smoothed complexity if under the perturbation of constraints, it has expected running
time $\poly(n,d,\sigma^{-1})$, and \cite{ST04} proved that the smoothed complexity of a specific simplex method based on the shadow vertex simplex method (which we will describe next) is at most
$O(d^{55}n^{86}\sigma^{-30} + d^{70}n^{86})$.
The best bound available in the literature is
$O(\sigma^{-2}d^2\sqrt{\log n})$ pivot steps due to \cite{DH18}, assuming $\sigma \leq 1/\sqrt{d\log n}$.
We note that assuming an upper bound on $\sigma$ can be done without loss of generality;
its influence can be captured as an additive term in the upper bound that does not depend on $\sigma$.

This work improves the dependence on $\sigma$ of the smoothed complexity, obtaining an upper bound of $O(\sigma^{-3/2} d^{13/4} \log^{5/4} n)$
for $\sigma \leq 1/d\sqrt{\log n}$.
As a second contribution, we prove the first non-trivial lower bound on the smoothed complexity of a simplex method,
finding that the shadow vertex simplex method requires $\Omega(\min(\frac{1}{\sqrt{\sigma d \sqrt{\log n}}}, 2^d))$
pivot steps.

\paragraph{Shadow Vertex Simplex Algorithm}
One of the most extensively studied simplex algorithms in theory is the shadow vertex simplex algorithm \cite{gas55,b82}.
Given an LP
\begin{align*}
    \max_{x\in \R^d} c^\top x, ~~ A x \leq b,
\end{align*}
for $A \in \R^{n\times d}, b\in \R^n, c\in \R^d$,
let $P = \{x\in \R^d : A x \leq b\}$ denote the feasible polyhedron of the linear program.
The algorithm starts from an initial vertex $x_0 \in P$ that optimizes an initial objective $c_0$\footnote{There are many standard methods of finding such initialization with at most multiplicative $O(d)$ overhead in running time, so we can assume that both $x_0$ and $c_0$ are already given. See the discussion in \cite{DH18}.}.
During the execution, it maintains an intermediate objective $c_\lambda = \lambda c + (1-\lambda) c_0$ and a vertex that optimizes $c_{\lambda}$. 
Thus by slowly increasing $\lambda$ from $0$ to $1$ during different pivot steps, the temporary objective gradually changes from $c_0$ to $c$, revealing the desired solution at the end.
Since each pivot step requires $\poly(d,n)$ computational work, theoretical analysis
has focused on analysing the number of pivot steps.

The algorithm is called the shadow vertex simplex method because, when performing orthogonal
projection of the feasible set onto the two-dimensional linear subspace $W = \mathrm{span}(c_0, c)$, the
vertices visited by the algorithm project onto the boundary of the projection (``shadow'')
$\pi_W(P)$. Assuming certain non-degeneracy conditions, which will hold with probability $1$ for the distributions
we consider, this projection gives an injective map from iterations of the
method to vertices of the shadow, meaning that we can
upper bound the number of pivot steps in the algorithm by the number of
vertices of the shadow polygon. 
This characterization makes the shadow vertex simplex method ideally suited for
probabilistic analysis.

To analyse the ``shadow size'', the number of vertices of the shadow polygon,
we follow earlier work of \cite{ver09} and reduce to the case that $b=\mathbf 1_n$.
In this case, well-established principles of polyhedral duality show that
we can bound the number of vertices of a convex polygon by the number of edges of its dual polygon:
\[
    \text{vertices}(\pi_W(P)) \leq \edges(W \cap \conv(0,a_1,\dots,a_n)) \leq \edges(W \cap \conv(a_1,\dots,a_n))+1.
\]
Here, $a_1,\dots,a_n$ denote the rows of the matrix $A$ used to define
$P = \{x \in \R^d : Ax \leq \mathbf 1_n\}$.
Note that the second inequality holds because $0 \in W$.

The smoothed complexity of shadow vertex simplex algorithm can thus be reduced to the
smoothed complexity of a two-dimensional slice of a convex hull.
For this reason, let us define the maximum smoothed shadow size as
\begin{align}\label{eq:def-shadow-bound}
    \mathcal{S}(n, d,\sigma) = \max_{\substack{\bar{a}_1, \ldots, \bar{a}_n \in \R^d \\ \max_{i \in [n]} \|\bar{a}_i\|_2 \leq 1 \\ W \subset \R^d, \dim(W) = 2}} \E_{\tilde{a}_1, \ldots, \tilde{a}_n \sim \mathcal{N}(0,\sigma^2)} \left[ \edges\big( \conv(\bar{a}_1 + \tilde{a}_1, \ldots, \bar{a}_n + \tilde{a}_n) \cap W\big)\right]
\end{align}

The following upper bound we take from \cite{DH18}, which states that the analysis of \cite{ver09}
can be strengthened to obtain the claimed bound. This upper bound should be understood
as stating that there exists a shadow vertex rule based simplex algorithm which
satisfies that smoothed complexity bound.
The lower bound is due to \cite{b87} and shows that the shadow vertex simplex rule, with
two adversarially specified objectives $c_0, c$,
can be made to follow paths of this length.

\begin{theorem}[Smoothed Complexity of Shadow Vertex Simplex Algorithm]
    Given any $n \geq d \geq 2, \sigma > 0$, the smoothed complexity of the shadow vertex simplex algorithm satisfies
    \begin{align*}
        \mathcal{S}(n, d, \sigma)/4 \leq \mathcal{SC}_{\textsc{ShadowSimplex}, n, d, \sigma} \leq 2\cdot\mathcal{S}\left(n + d, d, \min(\sigma, \frac{1}{\sqrt{d} \log d}, \frac{1}{\sqrt{d \log n}}) \right) + 4.
    \end{align*}
\end{theorem}

With this reduction, analysing the smoothed complexity of the simplex method
comes down to bounding the smoothed shadow size $\mathcal S(n,d,\sigma)$.
As such, that will be the focus of the remainder of this paper.

\subsection{Our Results}
The previous best shadow bound is due to \cite{DH18}, who prove that $\mathcal S(n,d,\sigma) \leq O(d^2 \sqrt{\log n}\sigma^{-2})$.
Our first main result strengthens this result for small values of $\sigma$.
\begin{theorem}
    For $n \geq d \geq 3$ and $\sigma \leq \frac{1}{8d\sqrt{\log n}}$,
    the smoothed shadow size satisfies
    \[
        \mathcal S(n,d,\sigma) = O\left( \sigma^{-3/2} d^{13/4} \log^{5/4} n \right).
    \]
\end{theorem}
%It is worth noting that by a more thorough analysis, one can attain a slightly improved bound of $O(\sigma^{-3/2} d^{13/4} \log^{5/4} n + d^{19/4} \log^{13/4} n)$.
A full overview of bounds on the smoothed shadow size, including previous results in the literature, can be found in
\Cref{tab:shadow_history}.

Second, we prove the first non-trivial lower bound on the smoothed shadow size, establishing that
$\mathcal S(4d-13,d,\sigma) \geq \Omega(\min(\frac{1}{\sqrt{\sigma d \sqrt{\log d}}}, 2^d))$
for $d > 5$. This lower bound is proven by constructing a polyhedron
$P = \{x \in \R^d : Ax \leq \mathbf 1_n\}$ and a two-dimensional subspace $W$ such that
for any small perturbation of $A$, the new polyhedron $P$ projected onto $W$ will have many vertices.
The construction is based on an extended formulation similar to those first constructed by
\cite{BN01, glineur}.
\begin{theorem}
    For any $d > 5$ and $\sigma \leq \frac{1}{360d\sqrt{\log(4d)}}$,
    the smoothed shadow size satisfies
    \[
        \mathcal S(4d-13,d,\sigma) = \Omega \left(\min \Big(\frac{1}{\sqrt{d\sigma\sqrt{\log d}}},2^d \Big) \right).
    \]
\end{theorem}

It is possible that the exponent of $\sigma$ in our bound can be further optimized.
In \Cref{sub:experiments}, we describe numerical experiments which suggest that
the actual shadow size for random perturbations of our constructed polytope might be as high as $\Omega(\min(\sigma^{-3/4}, 2^d))$.
We leave open the question whether having $n/d > 4$ can lead to stronger
lower bounds in the regime of $\sigma < 2^{-d}$.

\begin{table}[h!]
  \begin{center}
  \renewcommand{\arraystretch}{1.6}
    \begin{tabular}{|l|l|l|l|} % <-- Alignments: 1st column left, 2nd middle and 3rd right, with vertical lines in between
        \hline
        {\bf Reference} & {\bf Shadow size} & {\bf Model} \\
        \hline
        \cite{b87} & $\Theta(d^{3/2}\sqrt{\log n})$ & \makecell[l]{Average-\\ case,\\ Gaussian } \\ 
        \hline
        \hline
        \cite{ST04} & $O(\sigma^{-6}d^3n + d^6n\log^3 n)$ & Smooth \\
        \hline
        \cite{ds05} & $O(\sigma^{-2}dn^2\log n + d^2n^2\log^2 n)$ & Smooth \\
        \hline
        \cite{ver09} & $O(\sigma^{-4}d^3 + d^5\log^2n)$ & Smooth \\
        \hline
        \cite{DH18} & $O(\sigma^{-2}d^2\sqrt{\log n}+ d^{3}\log ^{1.5}n)$ & Smooth \\
        \hline
        This paper & $O(\sigma^{-3/2} d^{13/4} \log^{5/4} n + d^{19/4} \log^{2} n)$ & Smooth \\
        \hline
        \hline
        This paper & $\Omega(\min(\frac{1}{\sqrt{\sigma d \sqrt{\log d}}},2^d))$ & Smooth \\
        \hline
        \cite{zad73,Mur80, report/Goldfarb83, } & $2^d$ & Worst \\
        \hline
    \end{tabular}
    \caption{Bounds of expected number of pivots in previous literature, assuming $d \geq 3$. Logarithmic factors are simplified. The lower bound of \cite{b87} holds in the smoothed model as well.}\label{tab:shadow_history}
  \end{center}
\end{table}

\paragraph{Two-dimensional polygons}
To better understand the smoothed complexity of the intersection polygon $\conv(a_1,\dots,a_n)\cap W$,
we also analyse its two-dimensional analogue introduced by \cite{damerow2004extreme}.
Taking $\bar a_1,\dots,\bar a_n \in \R^2$, each satisfying $\|\bar{a}_i\|_2\leq 1$,
we are interested in the number of edges of the smoothed polygon
$\conv(\bar a_1 + \hat a_1,\dots,\bar a_n + \hat a_n)$,
where $\hat a_1,\dots,\hat a_n \sim N(0,\sigma^2)$ are independent.
The previous best upper bound on the smoothed complexity of this polygon
is $O(\sigma^{-1} + \sqrt{\log n})$, due to \cite{bwcachapter}.
Their analysis is based on an adaptation of the shadow bound by \cite{DH18}.
In \Cref{sec:ub-2d} we improve this upper bound, obtaining the following theorem.
\begin{theorem}[Two-Dimensional Upper Bound]\label{thm:twod-intro}
    Let $\bar{a}_1, \ldots, \bar{a}_n \in \R^2$ be $n > 2$ vectors with norm at
    most $1$. For each $i \in [n]$, let $a_i$ be independently distributed as
    $\mathcal{N}(\bar{a}_i, \sigma^2 I_{2})$. Then
    \begin{align*}
        \E\left[ \edges\left(\conv(a_1, \ldots, a_n) \right) \right] \leq O\left(\sqrt{\log n} + \frac{\sqrt[4]{\log n}}{\sqrt{\sigma}}\right).
    \end{align*}
\end{theorem}
To confirm that the above upper bound is stronger than that of \cite{bwcachapter}, one may
use the AM-GM inequality to verify that $2\frac{\sqrt[4]{\log n}}{\sqrt{\sigma}} \leq \sqrt{\log n} + \sigma^{-1}$.
Combined with the trivial upper bound of $n$ vertices, our bound nearly matches the lower bound of
$\Omega(\min(\sqrt{\log n} + \frac{\sqrt[4]{\log (n\sqrt\sigma)}}{\sqrt\sigma}, n))$ proven by \cite{DGGT16}.  
A full overview of previous results on the smoothed complexity of the two-dimensional convex hull
can be found in \Cref{tab:twod_history}.

\begin{table}[h!]
  \begin{center}
  \renewcommand{\arraystretch}{1.6}
    \begin{tabular}{|l|l|}
        \hline
        {\bf Reference} & {\bf Smoothed polygon complexity} \\
        \hline
        \cite{damerow2004extreme} & $O(\log(n)^2 + \sigma^{-2}\log n)$ \\
        % it would be fair to mention these people really study maximal points
        % in $d$ dimensions
        \hline
        \cite{thesis/Schnalzger14} & $O(\log n + \sigma^{-2})$ \\
        \hline
        \cite{DGGT16} & $O(\sqrt{\log n} + \sigma^{-1}\sqrt{\log n})$ \\
        \hline
        \cite{bwcachapter} & $O(\sqrt{\log n} + \sigma^{-1})$ \\
        \hline
        This paper & $O(\sqrt{\log n} + \frac{\sqrt[4]{\log (n)}}{\sqrt{\sigma}})$ \\
        \hline
        \hline
        \cite{DGGT16} & $\Omega(\min(\sqrt{\log n} + \frac{\sqrt[4]{\log (n\sqrt\sigma)}}{\sqrt{\sigma}},n))$ \\
        \hline
    \end{tabular}
    \caption{Bounds on the smoothed complexity of a two-dimensional polygon.}\label{tab:twod_history}
  \end{center}
\end{table}

\subsection{Related work}
\paragraph{Shadow Vertex Simplex Method}

The shadow vertex simplex algorithm has played a key role in many analyses of simplex and simplex-like algorithms.
On well-conditioned polytopes, such as those of the form $\{x \in \R^d : Ax \leq b\}$ where $A$
is integral with subdeterminants bounded by $\Delta$, the shadow vertex method has been studied
by in \cite{conf/icalp/BR13, jour/dcg/DH16}.
The shadow vertex method on polytopes all whose vertices are integral was studied in \cite{simplexzeroone,black22}

On random polytopes of the form $\{x \in \R^d : Ax \leq \mathbf 1_n\}$, assuming the
constraint vectors are independently drawn from any rotationally symmetric distribution, the
expected iteration complexity of the shadow vertex simplex method was studied by \cite{b82, b87, b99}.
In the case when the rows of $A$ arise from a Poisson distribution on the unit sphere,
concentration results for the shadow size and diameter bounds were proven in \cite{bdghl21}.
The diameter of smoothed polyhedra was studied by \cite{narayanan2021spectral}, who used the shadow bound of \cite{DH18}
to show that most vertices, according to some measure, are connected by short paths.

A randomized algorithm for solving linear programs in weakly polynomial time, using
the shadow vertex simplex method as a subroutine, was proposed in \cite{conf/stoc/KS06}.
The shadow vertex algorithm was recently used as part of the analysis of an interior point method
by \cite{adlnv}.

\paragraph{Extended Formulations}
For a polyhedron $P \subset \R^d$, an \emph{extended formulation}
is any polyhedron $Q \subset \R^{d'}$, $d' \geq d$,
such that $P$ can be obtained as the orthogonal projection of $Q$ to
some $d$-dimensional subspace.
Importantly, $Q$ can have much fewer facets than $P$.
While there is a wider literature on extended formulations,
here we describe only what is most relevant to the construction in \Cref{sec:lb}.

The construction in our lower bound is based on an adaptation
of the extended formulation by \cite{BN01} of the regular $2^k$-gon.
They used this construction to obtain a polyhedral approximation
of the second order cone
$\{x \in \R^{n+1} : \sum_{i=1}^n x_i^2 \leq x_{n+1}^2\}$.
A variant on their construction using fewer variables and inequalities was
given by \cite{glineur}. A more general construction based on \emph{reflection relations}
is used to construct extended formulation for the regular $2^k$-gon, as well as
other polyhedra, in \cite{Kaibel2013}.
Extended formulations for regular $n$-gons, when $n$ is not a power of $2$,
can be found in \cite{VANDAELE2017217}.

Approximations of the second order cone based on the work of \cite{BN01,
glineur} have been used to solve second order conic programs, see, e.g.,
\cite{Brmann2015}. These approximations were included in the open-source MIP solver
\texttt{SCIP} until version 7.0 \cite{GamrathEtal2020OO}.

\subsection{Proof Overview}

\subsubsection{Smoothed Complexity Upper Bound}

We write the random polytope from \eqref{eq:def-shadow-bound} as $Q = \conv(a_1, \ldots, a_n)$ where each $a_i$ is sampled independently from $\mathcal{N}_d(\bar{a}_i, \sigma^2 I)$ such that $\|\bar{a}_i\| \leq 1$. 
Our goal is to upper bound the expected number of edges of the polygon $Q \cap W$ for
fixed two-dimensional plane $W \subset \R^d$ and $\bar{a}_1, \ldots, \bar{a}_n$. 
This will immediately give us an upper bound of $\mathcal{S}(n, d, \sigma)$.

A fact used since the early days of smoothed analysis \cite{ST04} states that the intersection polygon $Q \cap W$ is \emph{non-degenerate} with probability measure $1$:
every edge on $Q \cap W$ is uniquely given by the intersection between $W$ and a facet of $Q$,
and every facet of $Q$ is spanned by exactly $d$ vertices.
For any index set $I \in \binom{[n]}{d}$, write $E_I$ as the event that $\conv(a_i : i\in I) \cap W$ is an edge of $Q \cap W$.
Non-degeneracy implies that every edge of $Q \cap W$ uniquely corresponds to an index set $I \in \binom{[n]}{d}$ such that $E_I$ holds.

Before sketching our proof, we first review the approach of \cite{DH18}, then discuss the main technical challenges in achieving an upper bound with better dependence on $\sigma$.

As a first step in \cite{DH18}, the authors replace the Gaussian distribution with a distribution they call the
Laplace-Gaussian distribution\footnote{The Laplace-Gaussian distribution follows the Gaussian probability density function near its mean but exhibits exponentially decaying tails.}.
 The latter distribution approximates the probability density of the former, in
 particular having nearly equivalent smoothed shadow size, while being
$O(\sigma^{-1}\sqrt{d\log n})$-log-Lipschitz for any point on its domain.
A probability distribution with probability density function $\mu$ is
$L$-log-Lipschitz for some $L > 0$ if, for any $x, y\in \R^d$, the condition
$|\log(\mu(x)) - \log(\mu(y))| \leq L\|x - y\|$ holds.

Next, define $\ell_I$ as the length of the edge on $Q \cap W$ that corresponds to $I$, i.e., when $\conv(a_i : i\in I) \cap W$ is an edge of $Q \cap W$ then $\ell_I$ gives the length of this line segment, and otherwise $\ell_I = 0$.
\cite{DH18} showed that, for any family $S \subseteq \binom{[n]}{d}$,
the expected number of edges of $Q \cap W$ coming from $S$ is at most
\begin{align}
    \E\left[\sum_{I \in S} 1[E_I]\right] \leq \frac{\E[\mathrm{perimeter}(Q \cap W)]}{\min_{I \in S} \E[\ell_I \mid E_I]}.
    \label{eq:counting-strategy-length}
\end{align}
Therefore, by taking $S = \{I \in\binom{[n]}{d}: \Pr[E_I] \geq \binom n d ^{-1}\}$, the expected number of edges of $Q \cap W$ is at most
\begin{align}
\E[\edges(Q \cap W)] = \E\Bigg[\sum_{I \in \binom{[n]}{d}} 1[E_I]\Bigg] \leq 1 + \E\left[\sum_{I \in S} 1[E_I]\right] \leq 1 + \frac{\E[\mathrm{perimeter}(Q \cap W)]}{\min_{I \in S} \E[\ell_I \mid E_I]}.
    \label{eq:counting-strategy-old}
\end{align}

To upper bound the numerator of \eqref{eq:counting-strategy-old}, notice that $Q \cap W$ is a convex polygon contained in the two-dimensional disk centered at $\mathbf{0}$ with radius $\max_{i \in [n]}\|\pi_W(a_i)\|$. 
It follows that the perimeter of $Q \cap W$ is at most the perimeter of this disk, namely,
\begin{align}
    \E[\mathrm{perimeter}(Q \cap W)] \leq 2\pi \cdot \E[\max_{i \in [n]}\|\pi_W(a_i)\|] \leq 2\pi \cdot (1 + 4\sigma \sqrt{\log n}), \label{eq:perimeter}
\end{align}
where the last step comes from a Gaussian tail bound.
% Thus the expected perimeter of $Q \cap W$ is at most that of the disk, i.e. $2\pi \cdot \E[\max_{i \in [n]}\|\pi_W(a_i)\|]$. 
% It then follows from Gaussian tail bound and $\max_{i \in [n]} \|\bar{a}_i\| \leq 1$ that $2\pi \cdot \E[\max_{i \in [n]}\|\pi_W(a_i)\|] \leq 2\pi (1 + O(\sigma \sqrt{\log n}))$. 
For the denominator of \eqref{eq:counting-strategy-old}, \cite{DH18} showed that for
any $I \in \binom{[n]}{d}$ with $\Pr[E_I] \geq \binom n d ^{-1}$ that, conditional on $E_I$, the expected edge length is at least
\begin{align}
    \E[\ell_I \mid E_I] \geq \Omega(\frac{\sigma^2}{d^2 \sqrt{\log n}} \cdot \frac{1}{1 + \sigma \sqrt{d \log n}}) \label{eq:shortest-edge}
\end{align}
By combining \eqref{eq:counting-strategy-old}, \eqref{eq:perimeter} and \eqref{eq:shortest-edge}, they proved an upper bound of $O(\sigma^{-2} d^2 \sqrt{\log n} + d^3 \log^{1.5}n)$.

\paragraph{New Strategy for Counting Edges}

While \cite{DH18} made the best possible analysis based of their edge-counting strategy as outlined in \eqref{eq:counting-strategy-old}, the strategy itself is sub-optimal.
The main drawback is that using the minimum expected length of edge $\min_{I \in \binom{[n]}{d}} \E[\ell_I \mid E_I]$ at the denominator of \eqref{eq:counting-strategy-old} is too pessimistic when the edges of $Q \cap W$ are long.
For instance, consider the case where an edge on $Q \cap W$ has length $\Omega(1)$ without perturbation. 
After the perturbation, is very likely that the length of this edge is still $\Omega(1)$, but \cite{DH18} uses a lower-bound of $\Omega(\frac{\sigma^2}{d^2 \sqrt{\log n}})$.

To improve this, we developed a new edge-counting strategy that can handle the long and short edges separately with two different ways of counting the edges. 
Take any index set $I \in \binom{[n]}{d}$; conditional on $E_I$, we write $e_I$ for its edge $\conv(a_i : i \in I) \cap W$.
We call the next edge in clockwise direction as $e_{I^+}$ and say it has length $\ell_{I^+}$.
We say $e_{I^+}$ is \emph{likely to be long} if $\Pr[\ell_{I^+} \geq t \mid E_I] \geq 0.05$ for some parameter $t > 0$ to be specified later. Here $0.05$ can be replaced by some arbitrary constant in $(0, 0.5)$.
Let $S_0 \subseteq \binom{[n]}{d}$ denote the set of $I \in \binom{[n]}{d}$ such that $e_{I^+}$ is likely to be long.
Following \eqref{eq:counting-strategy-length}, the expected number of edges in $Q \cap W$ that are likely to be long is at most
\begin{align}
    \E[\sum_{I \in S_0} 1[E_I]] \leq \frac{\E[\mathrm{perimeter}(Q \cap W)]}{0.05 t} \leq \frac{(2\pi + O(\sigma \sqrt{\log n}))}{0.05 t}.
    \label{eq:ub-long-edges}
\end{align}
Here the second step uses the exact same upper bound of $\E[\mathrm{perimeter}(Q \cap W)]$ as in \eqref{eq:perimeter}.

In the other case, $e_{I^+}$ is \emph{unlikely to be long}, i.e., $\Pr[\ell_{I^+} \geq t \mid E_I] < 0.05$. 
Now we will upper bound the number of such edges by claiming that their exterior angles each are large in expectation.
Let $\theta_I$ be the exterior angle at the endpoint of $\conv(a_i : i\in I)\cap W$ that comes last in clockwise order, i.e.,
the vertex where $e_I$ and $e_{I^+}$ meet.
Our key observation is that $\sin(\theta_I) \cdot \ell_{I^+}$ equals the distance from the affine hull of $e_I$ to the second vertex of $e_{I^+}$ in clockwise order. Therefore, by establishing a universal lower bound for such line-to-vertex distances, we can derive a lower bound for $E[\theta_I \mid E_I]$ for any edge $e_{I^+}$ that is unlikely to be long.
See \Cref{fig:edgecounts} for an illustration.

More formally,
let $p_{I^+}$ denote the next vertex of $Q \cap W$ after $e_I$ in clockwise order.

% let $S_1 \subseteq \binom{[n]}{d}$ be the collection of index sets $I \in \binom{[n]}{d}$ for which 
Let $S_1 \subseteq \binom{[n]}{d}$ be the collection of index sets $I \in \binom{[n]}{d}$ for which $\Pr[E_I] \geq 10 \binom{n}{d}^{-1}$. For a specific value of $\gamma > 0$, suppose that for each $I \in S_1$ we have
\begin{align}
    \Pr[\dist(p_{I^+}, \aff(e_I)) \geq \gamma \mid E_I] \geq 0.1 . \label{eq:gamma}
\end{align}

Then, for $I \in S_1 \setminus S_0$, conditional on $E_I$, we have
$\theta_I \geq \sin(\theta_I) \geq \frac{\gamma}{t}$ with probability at least $0.05$,
and the expectation of the exterior angle at the shared endpoint of $e_I$ and $e_{I^+}$ is at least $\frac{\gamma}{20t}$.

\begin{figure}[hbtp]
    \centering 
        \includegraphics[width=0.5\linewidth]{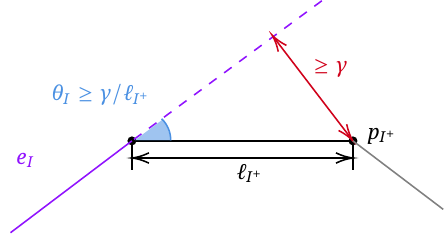}
        \caption{Illustration of the case when $e_{I^+}$ is short. In purple is the edge $e_I$ with its extension line dashed.
        The next edge in clockwise direction, $e_I^+$, has length $\ell_{I^+}$ and is drawn in black. In red is the line-to-vertex distance $\dist(p_{I^+}, \aff(e_I))$, and in blue is the angle $\theta_I$. If $\dist(p_{I^+}, \aff(e_I)) \geq \gamma$ then $\theta_I \geq \sin(\theta_I) \geq \gamma / \ell_{I^+}$.}
    \label{fig:edgecounts}
\end{figure}

On the other hand, the sum of exterior angles of a polygon equals to $2\pi$.
Therefore we can upper bound the expected number of edges that are not likely to be long of at most
\begin{align}
    \E[\sum_{I \in S_1 \setminus S_0} 1[E_I]] \leq \frac{2\pi}{\min_{I \in S_1 \setminus S_0} \mathbb{E}[\theta_I \mid E_I]} \leq \frac{2\pi\cdot20t}{\gamma}.
    \label{eq:ub-short-edges}
\end{align}
We will select $\gamma > 0$ as large as possible subject to
the fact that every $I \in \binom{[n]}{d}$ with $\Pr[E_I] \geq 10 \binom{n}{d}^{-1}$
satisfies $I \in S_1$.

Summing up the 
number of edges induced by sets in $S_0, S_1\setminus S_0$ and $\binom{[n]}{d} \setminus S_1$,
we get an upper bound on the expected edge-count of $Q \cap W$ of at most
\begin{align}
    \E[\edges(Q \cap W)] \leq \frac{2\pi + O(\sigma \sqrt{\log n})}{t} + \frac{40\pi t}{\gamma} + 10
    = O\left(\sqrt{\frac{1 + \sigma \sqrt{\log n}}{\gamma} } \right) \label{eq:counting-strategy-new}
\end{align}
where the final step follows from optimizing $t > 0$ to get the strongest possible bound.
We summarize our result in \Cref{thm:ub-general}.
For details of the edge-counting strategy, see \Cref{sec:ub-general}.

\paragraph{Two-dimensional Upper Bound}

In the second part of our proof, we need to show a lower bound of the expected
distance from the affine hull of an edge of $Q \cap W$ to the next vertex in clockwise order,
which is the quantity $\gamma$ mentioned in \eqref{eq:gamma}. 

As a warm-up, we first discuss the two-dimensional case $d=2$, which will be explained in detail in \Cref{sec:ub-2d}.
In this case, $W$ will become the entire two-dimensional space and will disappear from consideration. 
Therefore, we can focus on lower-bounding the distance from any edge $e$ of the polygon $Q = \conv(a_1, \ldots, a_n)$ to any other vertices, i.e., it suffices to find $\gamma > 0$ such that for any $I \in \binom{[n]}{2}$,
\begin{align*}
    \Pr[ \dist(\{a_j : j \notin I\}, \aff(e_I)) \geq \gamma \mid E_I] \geq 0.1.
\end{align*}
We can obtain a lower bound on this quantity for any $L$-log-Lipschitz distribution.
Through an appropriate coordinate transformation we prove that,
irrespective of the values of $a_j, j \notin I$,
the distance $\dist(\conv(a_j, j \notin I), \aff(e_I))$, conditional on being non-zero,
follows from a $2L$-log-Lipschitz distribution.
We will calculate that we may choose $\gamma = \Omega(1/L)$.
This result can be applied with small changes to our Gaussian random variables $a_1, \ldots, a_n$
by substituting for the Laplace-Gaussian distribution following \cite{DH18}.
Plugging into \eqref{eq:counting-strategy-new}, we get that in the
two-dimensional case,
$\E[\edges(Q)] \leq O(\sqrt{(1+\sigma\sqrt{\log n})/(\sigma/\sqrt{\log n})}) = O(\sqrt[4]{\log n} / \sqrt{\sigma} + \sqrt{ \log n})$
as stated in \Cref{thm:twod-intro}.

\paragraph{Multi-Dimensional Upper Bound}

As in the two-dimensional case, we must lower-bound of the line-to-vertex distance $\dist(p_{I^+}, \aff(e_I))$ (see \eqref{eq:gamma}) of $Q \cap W$.
Analysing this, however, becomes more challenging.
In two-dimensional case, each edge is the convex hull of two vertices among $a_1, \ldots, a_n$ and is independent of the other potential vertices $a_i$. 
In contrast, if $d \geq 3$ then each edge on $Q \cap W$ will be the intersection between $W$ and a $(d-1)$-dimensional facet of $Q$ (which is the convex hull of $d$ vertices), and each vertex will be the intersection between $W$ and a $(d-2)$-dimensional ridge of $Q$ (which is the convex hull of $d-1$ vertices). 
So the distributions of $e_I$ and $p_{I^+}$ are correlated.

To overcome these difficulties, we first factor $\dist(p_{I^+}, \aff(e_I))$ into the product
of separate parts which are easier to analyse, and then use log-Lipschitzness of $a_1, \ldots, a_n$ to lower-bound each part with good probability.
Fix without loss of generality $e = e_{[d]} = \conv(a_1, \ldots, a_d) \cap W$, as the potential edge of interest. 
Consider the second endpoint $p$ on $e$ in
clockwise direction and let $J \in \binom{[d]}{d-1}$ be the index set
such that $\{p\} = \conv(a_j : j \in J) \cap W$.
Let $p' = \conv(a_i : i \in J') \cap
W$ (with $J' \in \binom{[n]}{d-1}$) be the vertex next to the edge $e$ in clockwise direction.
From the non-degeneracy conditions, we know that $J'$ only differs from $J$ with
two vertices almost surely, so we can assume without loss of generality that $J
= \{2, \ldots, {d}\}$ and $J' = \{3, \ldots, d\} \cup \{k\}$ for some
$k \in \{d+1, \ldots, n\}$.

The main idea of our analysis is the observation that if the Euclidean diameter of $Q$ is
bounded above by $O(1)$ (which happens with overwhelming probability due to
Gaussian tail bound), then we can lower bound the two-dimensional line-to-vertex
distance $\dist(p', \aff(e))$ by the product of two distances $\Omega(\delta \cdot r)$, where
\begin{itemize}
    \item $\delta$ is the $d$-dimensional distance from the facet-defining hyperplane $\aff(a_1, \ldots, a_d)$ containing $e$, to the vertices of $Q$ that are not in the facet, i.e., 
    $$\delta = \dist(\conv(a_{d+1}, \ldots, a_n), \aff(a_1, \ldots, a_d));$$
    \item $r$ is the distance from the boundary of the ridge $\partial \conv(a_2, \ldots, a_d)$ to the one-dimensional line $\aff(e)$, i.e., $r = \dist(\aff(e), \partial \conv(a_2, \ldots, a_d))$.
\end{itemize}
We will give the formal statement of the distance splitting lemma in \Cref{lem:det}.

%We also remark that after conditioning on a fixed $(d-1)$-dimensional plane $\aff(a_1, \ldots, a_d)$ such that it is at the exterior of $Q$,
%$\delta$ will solely depend on the position of $a_{d+1},\dots,a_n$, and $r$ will solely depend on the location of $a_1, \ldots, a_{d}$ on their affine plane. 
%Therefore we can treat $\delta$ and $r$ as independent variables after specifying $\aff(a_1, \ldots, a_d)$.
%\sophie{I guess this part is not necessary, we don't use independence of $r$ and $\delta$. We use a union bound.}

It then remains to show that $r$ and $\delta$ are both unlikely to be too small. 
Similar to the two-dimensional case, we will also use log-Lipschitzness of $a_1, \ldots, a_d$ as our main tool.
\begin{itemize}
    \item After specifying $\aff(a_1,\dots,a_d)$, the lower bound on $\delta$ is derived from
        the remaining randomness in $a_{d+1},\dots,a_n$. Here we use both the $L$-log-Lipschitzness
        of the distributions of $a_{d+1},\dots,a_n$, as well as the knowledge that
        we only need to consider hyperplanes $\aff(a_1,\dots,a_d)$ which are
        likely to have all points $a_{d+1},\dots,a_n$ on its one side.
        This is made precise in \Cref{sub:rand-delta-lb}.
    \item The lower bound of $r$ resembles a more technical version of the proof of the ``distance lemma'' of \cite{ST04}.

    Write $\pi : \aff(a_1,\dots,a_d) \to \aff(e)^\perp$
    for the orthogonal projection sending $e$ to a single point $p = \pi(e)$.
    With this notation we have $r = \dist(p, \partial\conv(\pi(a_2),\dots,\pi(a_d)))$.
    
        First we show that each vertex of the projected ridge $\conv(\pi(a_2), \ldots, \pi(a_d))$ is a distance $\Omega(1/(d^2L))$ away from the hyperplane spanned by its other vertices.
        That means that the projected $(d-2)$-dimensional ridge $\conv(\pi(a_2),\dots,\pi(a_d))$ is likely to be wide in every direction.

        In the second step, we show that $W$ intersects $\conv(a_2,\ldots,a_d)$ ``through the center''.
        Specifically, we show that if we write $p = \sum_{i \in [d]} \lambda_i \pi(a_i)$ as the convex combination of $a_2, \cdots, a_d$, 
        then with constant probability $\min_{i \in [d]}\lambda_i \geq \Omega(1/(d^2L))$.

        The product of the lower bounds $
        \Omega(1/(d^2L))$
        and $\Omega(1/(d^2L))$ for the individual quantities
        will yield a lower bound for $r$
        with good probability.
        This is included in \Cref{sub:rand-r-lb}.
\end{itemize}

We conclude our main result of the line-to-vertex distance lower bound in \Cref{lem:normalvector}. Readers are referred to \Cref{sec:ub-multi-new} for detailed discussions.

\subsubsection{Smoothed Complexity Lower Bound}

Our smoothed complexity lower bound (\Cref{thm:lb-main}) is based
on two geometric observations using the inner and outer radius of the perturbed polytope.
For a polytope $P$ and a unit norm ball $\mathbb{B}$, its outer radius with center $x$ is the smallest $R$ such that there exists  $P \subset R \cdot \mathbb{B} + x$.
Its inner radius with center $x$ is the largest $r$ such that $r \cdot \mathbb{B} + x \subset P$.

The first observation is that, if a two-dimensional polygon $T$ has
inner $\ell_2$-radius of $r$ and outer $\ell_2$-radius of $(1 + \epsilon) \cdot r$ with respect to the same center, then $T$ has
at least $\Omega(\epsilon^{-1/2})$ edges (\Cref{lem:edgecount}). This comes
from the fact that every edge of $T$ has length at most $O(r\sqrt{\epsilon})$,
whereas the perimeter of $T$ is at least $2\pi r$.

Second, if two polytopes $Q, \tilde{Q} \subseteq \R^d$, each with inner radius $t$,
have Hausdorff distance $\epsilon < t/2$ to each other, 
% both defined using the same norm,
then $Q$ will approximate $\tilde{Q}$ in the way that (\Cref{lem:hdist})
$$(1 - 2\epsilon/t)\cdot Q \subseteq \tilde{Q} \subseteq (1 + \epsilon/t)\cdot Q.$$
In particular, for any two-dimensional linear subspace $W$ we have
\begin{equation}\label{eq:introlowerbound}
    (1 - 2\eps/t)\cdot Q \cap W \subseteq \tilde{Q} \cap W \subseteq (1 + \eps/t)\cdot Q \cap W.
\end{equation}

To prove our lower bound, we construct a polytope 
$Q = \conv(\bar{a}_1, \ldots, \bar{a}_n) \subset \R^d$
and a two-dimensional linear subspace $W$ such that
$\Omega(1) \cdot  \mathbb{B}^d_1 \subset Q \subset \mathbb{B}^d_1$,
and $Q \cap W$ has outer $\ell_2$-radius $r > 0$ and inner $\ell_2$-radius $\frac{r}{(1+4^{-d})}$.
Perturbing the vertices of $Q$, we obtain
$\tilde{Q} = \conv(a_1,\ldots, a_n)$, where $a_i \sim N(\bar a_i, \sigma^2 I_{d\times d})$ for each $i \in [n]$.
Note that $Q \subset \mathbb B_1^d$ implies that $\bar a_1,\dots,\bar a_n$ satisfy
the normalization requirement in \eqref{eq:def-shadow-bound}.
With high probability the Hausdorff distance in $\ell_1$ between $Q$ and $\tilde Q$ is bounded by
$\max_{i \in [n]} \|a_i - \tilde{a}_i\|_1 \leq O(\sigma d\sqrt{\log n})$.
Using \eqref{eq:introlowerbound}, we bound the inner and outer radius of $\tilde Q \cap W$.
A lower bound on the number of edges of $\tilde Q \cap W$ then follows from \Cref{lem:edgecount} as described above.

We remark that the polytopes $Q = \conv(a_1,\dots,a_n) \subset \R^d$ with $n = O(d)$ and two-dimensional subspaces $W$
such that $Q \cap W$ has inner $\ell_2$-radius $\frac{r}{1+4^{-d}}$ and outer $\ell_2$-radius $r > 0$
were first obtained by \cite{BN01} as an \emph{extended formulation} for
a regular $2^k$-gon with $O(k)$ variables and $O(k)$ inequalities.
Their polytope, however, has an outer and inner radius that differ by a factor $2^{\Omega(k)}$,
meaning that we cannot apply \Cref{lem:hdist} for $\sigma > 2^{-k}$.
We construct an alternative extended formulation where the ratio between
inner and outer $\ell_1$-radius is only $O(1)$.
With an appropriate scaling to get $Q \subset B_1^d$,
we find that the perturbed polytope $\tilde Q$ will have intersection
$\tilde Q \cap W$ with inner radius $\frac{r}{1+4^{-d}}(1-2\epsilon/t)$ and outer radius $(1+\epsilon/t)r$,
where $\epsilon = O(\sigma d \sqrt{\log n})$, and thus has
$\Omega(\min(\frac{1}{\sqrt \epsilon},2^d))$ edges,
with high probability.

\section{Preliminaries}
We write $\mathbf 1_n$ for the all-ones vector in $\R^n$, $\mathbf 0_n$ for the all-zeroes vector in $\R^n$,
and $I_{n \times n}$ for the $n$ by $n$ identity matrix. The standard basis vectors are denoted by $e_1,\dots,e_n \in \R^n$.
For a linear subspace $W \subset \R^n$ we denote the orthogonal projection onto $W$ by $\pi_W$.
The subspace of vectors orthogonal to a given vector $\omega \in \R^n$ is denoted $\omega^\perp$.

For a vector $x \in \R^n$, the $\ell_1$ norm is $\|x\|_1 = \sum_{i\in[n]} |x_i|$, the $\ell_2$-norm is $\|x\|_2 = \sqrt{\sum_{i\in[n]} x_i^2}$
and the $\ell_\infty$-norm is $\|x\|_\infty = \max_{i\in[n]} |x_i|$.
A norm without a subscript is always the $\ell_2$-norm.
Given $p \geq 1, d \in \Z_+$, define $\mathbb{B}_p^d = \{x \in \R^d : \|x\|_p \leq 1\}$ as the $d$-dimensional unit ball of $\ell_p$ norm.

We write $[n] := \{1,\dots,n\}$. By $\conv(a_1,\dots,a_n) = \conv(a_i : i \in [n])$ we denote the convex hull of vectors $a_1,\dots,a_n$,
and similarly we write for the affine hull as $\aff(a_i : i \in [n])$.
For sets $A, B \subset \R^d$, the distance between the two is
$\dist(A, B) = \inf_{a \in A, b \in B} \|a-b\|$.
For a point $x \in \R^d$ we write $\dist(x, A) = \dist(A, x) = \dist(A, \{x\})$.

We say a random event happens \emph{almost surely} if it occurs with probability $1$.

For a convex body $K \in \R^d$,  we define $\partial K \subset \mathrm{span}(K)$ as the boundary of $K$ in the linear subspace spanned by the vectors in $K$.

\subsection{Polytopes}

\begin{definition}[Polytope]
    A convex set $P \subset \mathbb{R}^d$ is a polyhedron if it can be expressed as $P = \{x \in \mathbb{R}^d : A x \leq b\}$ for some $A \in \mathbb{R}^{n \times d}$, $b \in \mathbb{R}^n$ where $n \in \mathbb{Z}_+$.
    A bounded polyhedron is also called a polytope.
\end{definition}

\begin{definition}[Valid Condition and Facet]
    Given a polytope $P \subset \R^d$, vector $c\in \R^d$ and $d\in \R$, we say the linear condition $x^\top c \leq d$ is valid for $P$ if  the condition holds for all $x\in P$.
    
    A subset $F \subset P$ is called a face of $P$ if $F = P \cap \{x\in \R^d : x^\top c = d\} \neq \emptyset$ for some valid condition $x^\top c \leq d$.
    A facet is a $d-1$-dimensional face, a ridge is a $d-2$-dimensional face, an edge is a $1$-dimensional face and a vertex is a $0$-dimensional face.
\end{definition}

\begin{definition}[Polar dual of a convex body]\label{def:smooth_polar}
    Let $P \subset \R^d$ be a convex set. Define the polar dual of $P$ as 
    \begin{align*}
        P^\circ = \{y \in \R^d : y^\top x \leq 1, \forall x \in P\}.
    \end{align*}
    % Specifically, when $P=\{x\in \mathbb{R}^d, A x \le b\} \subset \R^d$ is a polytope given by a linear system, the polar dual of $P$ equals to
    % \begin{align*}
    %     P^\circ := \{y\in \R^d ~|~ y^\top x\le 1, \forall x \in P\} ~.
    % \end{align*}
\end{definition}

We state some basic facts from duality theory:
\begin{fact}[Polar dual of polytope]
    Let $P \subset \R^d$ be a polytope given by the linear system $P=\{x\in \mathbb{R}^d, A x \le \mathbf{1}_n\} \subset \R^d$ for some $A\in \R^{n\times d}$. Then the polar dual of $P$ equals to
    \begin{align*}
        P^\circ := \mathrm{conv}(\mathbf{0}_d,a_1,a_2,\ldots,a_n).
    \end{align*}
    where $a_1, \ldots, a_n \in \R^d$ are the row vectors of $A$.
    Moreover, $P$ is bounded if and only if $\mathbf{0}_d \in \operatorname{int}(\conv(a_1, \ldots, a_n))$.
\end{fact}

\begin{fact}\label{fact:dual-containing}
    Let $P, Q \subset \R^d$ be two convex sets such that $P \subseteq Q$. Then $Q^\circ \subseteq P^\circ$.
\end{fact}

\begin{fact}\label{fact:dual-intersection}
    Let $P \subset \R^d$ be a polytope, and let $W \subset \R^d$ be any $k \leq d$-dimensional linear subspace. Then the polar dual of $\pi_W(P)$, considered as a subset of the linear space $W$, is equal to $P^\circ \cap W$.
\end{fact}

\subsection{Probability Distributions}\label{sec:laplace-gaussian}
All probability distributions considered in this paper will admit a probability density
function with respect to the Lebesgue measure.

\begin{definition}[Gaussian distribution]
    The $d$-dimensional Gaussian distribution $\mathcal{N}_d(\Bar{a}, \sigma^2I)$ with support $\R^d$, mean $\Bar{a}\in \R^d$, and standard deviation $\sigma$ is defined by the probability density function
    \begin{align*}
        (2\pi)^{-d/2} \cdot \exp\left(-\|s - \Bar{a}\|^2 / (2\sigma^2) \right).
    \end{align*}
    at every $s \in \R^d$.
\end{definition}

A basic property of Gaussian distribution is the following strong tail bound:
\begin{lemma}[Gaussian tail bound]\label{lem:gaussian-tail}
    Let $x \in \R^d$ be a random vector sampled from $\mathcal{N}_d(\mathbf{0}, \sigma^2 I)$.
    For any $t\geq 1$ and any $\mathbf{\theta}\in \mathbb{S}^{d-1}$ where $\mathbb{S}^{d-1}$ is the unit sphere in the $d$-dimensional space, we have
        \begin{align*}
            \Pr[\|x\| \geq t\sigma \sqrt{d}] & ~ \leq \exp(-(d/2)(t-1)^2).
            % \Pr[|x^\top \mathbf{\theta}| \geq t\sigma] & ~ \leq \exp(-t^2/2).
        \end{align*}
\end{lemma}

From this, one can upper-bound the maximum norm over $n$ Gaussian random vectors with mean $\mathbf{0}_d$ and variance $\sigma^2$ by $4\sigma \sqrt{d \log n}$ with dominating probability.
\begin{corollary}[Global diameter of Gaussian random variables]\label{cor:gaussian-globaldiam}
    For any $n \geq 2$,
    let $x_1, \ldots, x_n \in \R^d$ be random variables where each $x_i \sim \mathcal{N}_d(\mathbf{0}_d, \sigma^2 I)$. Then with probability at least $1 - \binom{n}{d}^{-1}$, $\max_{i \in [n]}\|x_i\| \leq 4 \sigma \sqrt{d\log n}$.
\end{corollary}
\begin{proof}
    From \Cref{lem:gaussian-tail}, we have for each $i \in [n]$ that
    \begin{align*}
        \Pr[\|x_i\| > 4\sigma \sqrt{d\log n}] \leq \exp(-\frac{d(4\sqrt{\log n} - 1)^2}{2}) \leq \exp(-2d\log n) \leq n^{-1} \cdot \binom{n}{d}^{-1}.
    \end{align*}    
    Then the statement follows from the union bound over all choices of $i \in [n]$.
\end{proof}

A helpful technical substitute for the Gaussian distribution was introduced by \cite{DH18}:
\begin{definition}[$(\sigma, r)$-Laplace-Gaussian distribution]\label{def:laplace-gaussian-distribution}
    For any $\sigma, r > 0, \Bar{a} \in \R^d$, define the $d$-dimensional $(\sigma, r)$-Laplace-Gaussian distribution with mean $\Bar{a}$, or $LG_d(\Bar{a}, \sigma, r)$, if its density function is proportional to
    \begin{align}
        f(x) = \begin{cases}
        \exp\left( -\frac{\|x - \Bar{a}\|^2 }{2\sigma^2} \right) & \text{, if } \|x - \Bar{a}\| \leq r\sigma \\
        \exp\left( -\frac{\|x - \Bar{a}\| r }{\sigma} + \frac{r^2}{2} \right) & \text{, if } \|x - \Bar{a}\| > r\sigma.
        \end{cases} \label{eq:laplace-gaussian}
    \end{align}
\end{definition}

The Laplace-Gaussian random variables satisfies many desirable properties. Like the
Gaussian distribution, the distance to its mean is bounded above with high
probability. Moreover, its probability density is log-Lipschitz throughout its
domain (as a contrast, the probability density of Gaussian distribution is only
log-Lipschitz close to the expectation).  The definition of $L$-log-Lipschitz
is as follows:

\begin{definition}[$L$-log-Lipschitz random variable]\label{def:log-lipschitz}
    Given $L > 0$, we say a random variable $x \in \R^d$ with probability density $\mu$ is $L$-log-Lipschitz (or $\mu$ is $L$-log-Lipschitz), if for all $x, y\in \R^d$, we have
    \begin{align*}
        |\log(\mu(x)) - \log(\mu(y))| \leq L\|x - y\|,
    \end{align*}
    or equivalently, $\mu(x) / \mu(y) \leq \exp(L\|x - y\|)$.
\end{definition}

\begin{lemma}[Properties of Laplace-Gaussian random variables, Lemmas 3.7 and 3.33 of \cite{DH18}]\label{lem:properties-laplace-gaussian}
    Given any $n \geq d$, $\sigma > 0$. Let $a_1, \ldots, a_n \in \R^d$ be independent random variables each sampled form $LG_d(\Bar{a}, \sigma, 4\sigma \sqrt{d \log n})$ (see Definition \ref{def:laplace-gaussian-distribution}).
    Then $a_1, \ldots, a_n$ satisfy the follows:
    \begin{enumerate}
        \item (Log-Lipschitzness) For each $i \in [n]$, the probability density of $a_i$ is $(4\sigma^{-1} \sqrt{d\log n})$-log-Lipschitz.
        \item (Bounded maximum norm) 
        % For any $k \leq d$, any fixed $k$-dimensional linear subspace $H \subset \R^d$.
        With probability at least $1 - \binom{n}{d}^{-1}$, $\max_{i \in n} \|a_i - \bar{a}_i\| \leq 4\sigma \sqrt{d \log n}$.%$ \cdot \max_{i \in [n]} \|\bar{a}_i\|$.
        \item (Bounded expected radius of projection)  For any $k \leq d$, any fixed $k$-dimensional linear subspace $H \subset \R^d$, we have $\E[\max_{i \in [n]}\|\pi_H(a_i - \bar{a}_i)\|] \leq 8\sigma\sqrt{k\log n}$.%\cdot \max_{i \in [n]} \|\bar{a}_i\|$.
    \end{enumerate}
\end{lemma}

Most importantly, Laplace-Gaussian perturbations lead to nearly the same shadow size as Gaussian perturbations.
\begin{lemma}[Lemma 3.34 of \cite{DH18}]\label{lem:from-LG-to-Gaussian}
    Given any $n \geq d \geq 2, \sigma > 0$, any two-dimensional linear subspace $W \subset \R^d$, and any $\bar{a}_1,\ldots,\bar{a}_n \in \R^d$.
    For every $i \in [n]$, let $a_i \sim \mathcal{N}_d(\bar{a}_i, \sigma)$
    and $\widehat a_i \sim LG_d(\bar{a}_i, \sigma, 4\sigma \sqrt{d\log n})$ be independently sampled.
    Then the following holds
    \[
         \E\left[ \edges(\conv(a_1 , \ldots, a_n) \cap W) \right]
        \leq 1 + \E \left[ \edges(\conv(\widehat a_1 , \ldots, \widehat a_n) \cap W) \right].
    \]
\end{lemma}
Although \cite{DH18} state the above lemma only for $d \geq 3$, their proof applies without
change to the case $d = 2$.

\iffalse
\begin{corollary}[Bounded Global Diameter of $n$ Laplace-Gaussian Vertices]\label{cor:globaldiam}
    For any $\sigma \leq \frac{1}{8\sqrt{d \log n}}$, any $\bar{a}_1, \ldots, \bar{a}_n \in \R^d$ where $\max_{i\in [d]} \|\bar{a}_i\| \leq 1$.
    Let $a_1, \ldots, a_n \in \R^d$ be independent random vectors where each $a_i \sim LG_d(\bar{a}_i, \sigma, 4\sigma \sqrt{d \log n})$ (see \eqref{eq:laplace-gaussian}).
    Then we have
    \[
        \Pr[\exists i,j \in [n] ; \|a_i - a_j\| \geq 3] \leq \binom{n}{d}^{-1}.
    \]
\end{corollary}
\begin{proof}
    By property 1 of \Cref{lem:properties-laplace-gaussian} and a union bound over all $i\in [n]$, we have that 
    $$\Pr[\forall i\in [n], \|a_i - \bar{a}_i\| \leq 4\sigma \sqrt{d \log n} \leq \frac{1}{2}] \geq 1 - \binom{n}{d}^{-1}.$$
    Therefore, with probability at least $(1 - \binom{n}{d}^{-1})$, for any $i, j \in n$ we have $\|a_i - a_j\| \leq \|\bar{a}_i\| + \|a_i - \bar{a}_i\| + \|\bar{a}_j\| + \|a_j - \bar{a}_j\| \leq 3$, and the lemma follows.
\end{proof}
\fi

\subsection{Change of Variables}\label{sec:change-of-variables}

We will make use of a specific change of variables, which is a standard tool in stochastic
and integral geometry.
It will allow us to investigate, for points $a_1,\dots,a_d \in \R^d$,
the distribution of these points when conditioning on a specific affine hyperplane
$\aff(a_1,\dots,a_d)$.
This will enable us to make conclusions about the shapes of the faces of
the convex hull $\conv(a_1,\dots,a_n)$ of $n$ vectors.

\begin{definition}[Change of variables]\label{def:change-of-variables}
    %Let $a_1,\ldots, a_d$ be $d$ affine independent vectors in $\R^d$.
    Let $\theta\in \mathbb{S}^{d-1}, t\in \R$
    be such that $\forall i\in [d], \theta^\top a_i = t$ and suppose $\theta^\T e_1 > 1$ without loss of generality where $e_1 = (1, 0, \ldots, 0)^\top \in \R^d$ is the unit vector that has nonzero element on the first coordinate.
    
    Fix $h$ as any isometric embedding from $\R^{d-1} \to e_1^\perp$.
    Let $\tilde{R}_{\theta}: \R^d \to \R^d$ denote the rotation
    that rotates $e_1$ to $\theta$ in the two-dimensional subspace
    $\mathrm{span}(e_1, \theta)$, and is the identity transformation on $\mathrm{span}(e_1, \theta)^\perp$.
    Define $R_{\theta} = \tilde{R}_{\theta} \circ h$ to be the resulting
    isometric embedding from $\R^{d-1}$, identified with $e_1^\perp$, to $\theta^\perp$.
    Now define the transformation $\phi$ from $\theta\in \mathbb{S}^{d-1}, t \in \R, b_1,\ldots,b_d\in \R^{d-1}$ to $a_1,\ldots, a_d \in \R^d$ as follows:
    \begin{align}\label{eq:change-of-variables}
        \phi(\theta, t, b_1,\ldots,b_d) = (R_{\theta}(b_1) + t\theta, \ldots, R_{\theta}(b_d) + t\theta).% = (a_1, \ldots, a_d).
    \end{align}
\end{definition}

The choice of the embeddings $\tilde R_\theta$ is largely arbitrary.
What matters is that the transformation $\phi$ and its inverse are continuous almost everywhere.

\begin{lemma}[Jacobian of the transformation, see Theorem 7.2.7 in \cite{Schneider2008}]\label{lem:smooth-jacobian}
    Let $\phi: \sfe \times \R \times \R^{(d-1) \times d} \to \R^ {d \times d}$ be the transformation defined in Definition \ref{def:change-of-variables}.
    The transformation $\phi$ is defined almost everywhere and has Jacobian determinant that equals to
    \begin{align*}
        \left| \det\left(\frac{\partial \phi}{\partial (\theta,t,b_1,\dots,b_d)}\right)\right| = C_d \cdot (d-1)! \cdot \vol_{d-1}(\conv(b_1,\ldots,b_d))
    \end{align*}
    for some constant $C_d$ depending only on the dimension.
    As a consequence, if $a_1,\dots,a_d \in \R^d$ are points with probability density $\mu(a_1,\dots,a_d)$
    and if $\theta \in \sfe, t \in \R, b_1,\dots,b_d \in \R^{d-1}$ have probability density proportional to
    \[
        \vol_{d-1}(\conv(b_1,\dots,b_d)) \cdot \mu(t\theta + R_\theta(b_1),\dots,t\theta + R_\theta(b_d))
    \]
    then $\E[f(a_1,\dots,a_d)] = \E[f(\phi(\theta,t,b_1,\dots,b_d))]$ for any measurable function $f$.
\end{lemma}
The interested reader might observe that if $(a_1,\dots,a_d)=\phi(\theta,t,b_1,\dots,b_d)$ then the Jacobian
\[(d-1)! \cdot \vol_{d-1}(\conv(a_1,\ldots,a_d)) = (d-1)! \cdot \vol_{d-1}(\conv(b_1,\ldots,b_d))\]
is equal to the determinant of the $d \times d$ matrix $(\theta^\top, (a_1-a_d)^\top,(a_2-a_d)^\top,\dots,(a_{d-1}-a_d)^\top)$,
which is equal to the determinant of the $(d-1) \times (d-1)$ matrix with rows or columns given by
$b_1-b_d,b_2-b_d,\dots,b_{d-1}-b_d$.

In particular, we will use this transformation to condition on the value of $\theta$
and consider events in the variables $t,b_1,\dots,b_d$. For this purpose, we have the following fact.
\begin{lemma}[Log-Lipschitzness of the Position of Affine Hull]\label{fact:joint-t-log-lipschitz}
    Let $a_1, \ldots, a_d \in \R^d$ be $d$ independent $L$-log-Lipschitz random variables,
    and let $(\theta, t, b_1,\dots,b_d) = \phi^{-1}(a_1,\dots,a_d)$, where $\phi: \sfe \times \R \times \R^{(d-1) \times d} \to \R^ {d \times d}$ is defined in Definition \ref{def:change-of-variables}.
    Then conditional on the values of $\theta, b_1,\dots,b_d$, the random variable $t$ is $(dL)$-log-Lipchitz.
\end{lemma}
\begin{proof}
    By \Cref{lem:smooth-jacobian}, the joint probability density of $(a_1, \ldots, a_d)$ is proportional to
    \begin{align*}
        \vol_{d-1}(\conv(b_1, \ldots, b_d)) \cdot \prod_{i=1}^d \mu_i(R_{\theta}(b_i) + t\theta)
    \end{align*}
    where $\mu_i$ is the probability density of $a_i$.
    Conditioning on $b_1,\dots,b_d \in \R^{d-1}$ and $\theta \in \mathbb{S}^{d-1}$, the volume $\vol_{d-1}(\conv(b_1, \ldots, b_d))$ is fixed. 
    The statement then follows from the fact that for each $i \in [d]$, $\mu_i(R_{\theta}(b_i) + t\theta)$
    is $L$-log-Lipschitz in $t$ for any $b_i$.
\end{proof}

Conditional on $\theta$ and $t$, the points $b_1,\dots,b_d$ are not $L$-log-Lipschitz due to the volume term.
When it becomes relevant, we will show that this factor does not affect the argument in any negative way.

\subsection{Non-Degenerate Conditions}

\begin{definition}[Non-degenerate polytope]\label{def:non-degeneracy}
    A polytope $Q = \conv(a_1,\dots,a_n) \subset \R^d$ is called non-degenerate, if it is simplicial (every facet is a simplex) and if, for $i \in [n]$,
    $a_i \in \partial Q$ implies that $a_i$ is a vertex of $Q$.
\end{definition}

\begin{definition}[Non-degenerate intersection with a 2D-plane]\label{def:non-degenerate-intersection}
    Let $Q \subset \R^d$ be a non-degenerate polytope and let $W \subset \R^d$ be a two-dimensional linear subspace.
    We say $Q$ has non-degenerate intersection with $W$, if 
    \begin{enumerate}
        \item the edges of the two-dimensional polygon $Q \cap W$ have one-to-one correspondence to the facets of $Q$ that have non-empty intersection with $W$; and
        \item the vertices of $Q \cap W$ have one-to-one correspondence to the $(d-2)$-dimensional faces (ridges) of $Q$ that have non-empty intersection with $W$
    \end{enumerate}
\end{definition}

\begin{fact}[Non-degenerate conditions of random polytope]\label{lem:non-degeneracy}
    Given any $n \geq d \geq 2$ and any fixed two-dimensional plane $W \subset \R^d$.
    For $a_1, \ldots, a_n \in \R^d$,
    the polytope $Q = \conv(a_1, \ldots, a_n)$ satisfies the following properties
    everywhere except for a set of measure $0$:
    \begin{enumerate}
        \item $Q$ is non-degenerate;
        \item $Q$ has non-degenerate intersection with $W$;
        \item For every normal vector $v$ to any facet of $Q$, $e_1^\T v \neq 0$.
    \end{enumerate}
\end{fact}

Assume the polytope $Q = \conv(a_1,\dots,a_n)$ and the two-dimensional linear subspace $W \subset \R^d$ satisfy the non-degenerate conditions in \Cref{lem:non-degeneracy}.
Each edge of the two-dimensional polygon formed by the intersection $W \cap Q$ can be described by a set of $d$ vertices, where the edge is equivalent to the intersection of $W$ with the convex hull of these $d$ vertices. Furthermore, each vertex of $Q \cap W$ corresponds to a set of $(d-1)$ vertices. The following lemma characterizes the relation of these sets for adjacent vertices and edges:

\begin{fact}[Properties of neighboring vertices on non-degenerate intersection polygon]\label{fact:diff-adjacent-vertex}
    Let $W \subset \R^d$ be a two-dimensional linear subspace, $Q = \conv(a_1,\dots,a_n) \subset \R^d$ is simplicial and has non-degeneracy intersection with $W$. 
    Given $J_1, J_2 \in \binom{[n]}{d-1}$, $I \in \binom{[n]}{d}$, suppose
    (1) $V_{J_1} = \conv(a_j : j\in J_1) \cap W$ and $V_{J_2} = \conv(a_j : j\in J_2) \cap W$ are two adjacent vertices of $Q \cap W$, and (2) $\conv(a_i : i\in I) \cap W$ is an edge of $Q \cap W$ that contains $V_{J_1}$ but not contains $V_{J_2}$.
    Then we have $|J_1 \backslash J_2| = |J_2 \backslash J_1| = 1$ and $|I \backslash J_2| = 2$.
\end{fact}
\begin{proof}
    Let $I' = J_1 \cup J_2$. Then $\conv(V_{J_1}, V_{J_2}) = \conv(a_i : i\in I) \cap W$ is an edge of the polygon $Q \cap W$. Since $Q$ has non-degenerate intersection with $W$, we have that $|I'| = d$. Combining with $|J_1| = |J_2| = d-1$ gives us that $|J_1 \backslash J_2| = |J_2 \backslash J_1| = 1$.
    
    Next we consider $|I \backslash J_2|$. Since $J_1 \subset I$ and $|J_1 \backslash J_2| = 1$, it could only be the case that $|I \backslash J_2| \in \{1, 2\}$. If  $|I \backslash J_2| = 1$, then since $|I| = |J_2|+1$ we must have $J_2 \subset I$, but this contradicts to the fact that $\conv(J_2) \not\subset \conv(I)$. Therefore we could only have $|I \backslash J_2| = 2$.
\end{proof}

\section{Smoothed Complexity Upper Bound}\label{sec:ub-general}

In this section, we establish our key theorem for upper bounding the number of edges
of a random polygon $\conv(a_1,\dots,a_n) \cap W$ for $W$ a fixed $2$-dimensional
linear subspace and $a_1,\dots,a_n \in \R^d$. 
We demonstrate that if 
for any edge on the shadow polygon $\conv(a_1, \ldots, a_n) \cap W$, 
the expected distance between the affine hull of the edge and the next vertex
on the shadow is sufficiently large in expectation, then the expected number
of edges of $\conv(a_1, \ldots, a_n) \cap W$ can be upper-bounded.

\begin{definition}[Facet and edge event]\label{def:edge-event}
    For $I \subset [n]$, we write $F_I = \conv(a_i : i \in I)$.
    Define $E_I$ to be the event that both
    $F_I$ is a facet of $\conv(a_1,\dots,a_n)$
    and $F_I \cap W \neq \emptyset$.
    % this next part is not definition, and is covered by the next remark
    %Assuming non-degeneracy (see \Cref{lem:non-degeneracy}),
    %$E_I$ is equivalent to that $F_I \cap W$ is an edge of $\conv(a_1,\dots,a_n) \cap W$.
\end{definition}
% Note that when $d=2$ then $W=\R^2$ and the condition $F_I \cap W \neq \emptyset$
% is guaranteed to hold.

\begin{remark}\label{fact:edgecounting}
    Any edge $e$ of $\conv(a_1,\dots,a_n) \cap W$ can be written as $e = F_I \cap W$
    for some $I \subset [n]$ for which $E_I$ holds.
    Assuming non-degeneracy, this relation between edges and index sets
    is a one-to-one correspondence, and moreover every $I \subset [n]$
    for which $E_I$ holds satisfies $|I|=d$.
\end{remark}

To state the key theorem's assumption, we need a concept of 'clockwise' to characterize the order of edges and vertices on the shadow polygon.

\begin{definition}[Clockwise order of edges and vertices]
    For any given two-dimensional linear subspace $W \subset \R^d$
    we denote an arbitrary but fixed rotation as ``clockwise''.
    For the polygon $\conv(a_1, \ldots, a_n) \cap W$ of our interest, let $p_1,\dots,p_k$
    denote its vertices in clockwise order and write $p_{k+1}=p_1, p_{k+2}=p_2$.
    Then for any edge $e = [p_{i-1}, p_i]$,
    we call $p_i$ its second vertex in clockwise order and
    we call $p_{i+1}$ the next vertex after $e$ in clockwise order.
    The edge $[p_i, p_{i+1}]$ is the next edge after $e$ in clockwise order.
\end{definition}
Note that the above terms are well-defined in the sense that they depend only on the polygon and
the orientation of the subspace, not on the vertex labels.
With this definition in place, we can now state the theorem itself:
\begin{theorem}[Smoothed complexity upper bound for continuous perturbations]\label{thm:ub-general}
    Fix any $n,d \geq 2$, $\sigma \geq 0$, and any two-dimensional linear subspace $W \subseteq \R^d$. 
    Let $a_1, \ldots, a_n \in \R^d$ be independently distributed each according to a
    continuous probability distribution.
    
    For any $I \in \binom{[n]}{d}$, conditional on $E_I$, 
    define $y_I \in W$ as the outer unit normal of the edge $F_I \cap W$.
    For $\gamma > 0$,
    suppose that for each $I \in \binom{[n]}{d}$ such that $\Pr[E_I] \geq 10 \binom{n}{d}^{-1}$ we have
    \begin{align*}
        \Pr[y_I^\top p_2 - y_I^\top p_3 \geq \gamma \mid E_I] \geq 0.1 ,
    \end{align*}
    where we write $[p_1, p_2] = F_I \cap W$ and $p_3 \in \conv(a_1, \ldots, a_n)\cap W$ as
    the next vertex after $F_I \cap W$ in clockwise order.
    Then we have
    \begin{align*}
        \E\left[ \edges\left(\conv(a_1, \ldots, a_n) \right) \right] &\leq 10 + 80\pi \sqrt{\frac{\E[\max_{i\in [n]} \|\pi_W(a_i)\|]}{\gamma}} \\
        &= O\left(\sqrt{\frac{\E[\max_{i \in [n]} \|\pi_W(a_i)\|]}{\gamma}} \right).
    \end{align*}
\end{theorem}
Note that, assuming non-degeneracy, $y_I$ is well-defined if and only if $E_I$ happens.
In this case, we are guaranteed that $y_I^\top p_2 - y_I^\top p_3 > 0$.

To prove the above theorem, we show that any $I \in \binom{[n]}{d}$
with $\Pr[E_I] \geq \binom{n}{d}^{-1}$ can be charged to either
a portion of the perimeter of the polygon $\conv(a_1,\dots,a_n)\cap W$
or to a portion of its sum $2\pi$ of exterior angles at its vertices.
\begin{definition}[Exterior angle and length of the next edge]
    Given any $I \in \binom{[n]}{d}$, we define two random variables $\theta_I, \ell_{I^+} \geq 0$.
    If $E_I$ happens, write $p_1, p_2 \in F_I \cap W$ for the first and the second endpoint of $F_I \cap W$ in clockwise order.
    Let $\theta_I$ to be the (two-dimensional) exterior angle of $\conv(a_1, \ldots, a_n) \cap W$ at $p_2$; 
    If $E_I$ doesn't happen then let $\theta_I = 0$.
    
    Let $\ell_{I^+}$ denote the following random variable: 
    If $E_I$ happens, then $\ell_{I^+}$ equals to the length of the next edge
    after $F_I \cap W$ in clockwise order, i.e., the other edge
    of $\conv(a_1, \ldots, a_n) \cap W$ containing $p_2$.
    If $E_I$ doesn't happen then let $\ell_{I^+} = 0$.
\end{definition}

\begin{proof}[Proof of Theorem~\ref{thm:ub-general}]
    Since we have non-degeneracy with probability $1$, by \Cref{lem:non-degeneracy} and linearity of expectation we find
    \begin{align*}
        \E\left[ \edges\left(\conv(a_1, \ldots, a_n) \cap W \right) \right] = \sum_{I \in \binom{[n]}{d}} \Pr[E_I].
    \end{align*}
    % Since we have non-degeneracy with probability $1$, the above inequality is in fact an equality.
    We can give an upper bound on the expected number of edges of $\conv(a_1, \ldots, a_n) \cap W$
    by upper-bounding each $\Pr[E_I]$.
    Fix any $I \in \binom{[n]}{d}$ and let $t > 0$ be a parameter to be determined later. 
    We consider three different possible upper bounds on $\Pr[E_I]$, at least one of which will always hold:
    
    {\bf Case 1: $\Pr[E_I] \leq 10\binom{n}{d}^{-1}$.} 
    
    Since $\sum_{I \in \binom{[n]}{d}}10 \binom{n}{d}^{-1} = 10$, one can immediately see that the total contribution of edges counted in this case is at most $10$.
    
    {\bf Case 2: $\Pr[E_I] > 10\binom{n}{d}^{-1}$ and $\Pr[\ell_{I^+} \geq t \mid E_I] \geq \frac{1}{20}$.} 
    
    In this case, by Markov's inequality $\E[\ell_{I^+} \mid E_I] \geq \frac{t}{20}$, therefore we obtain from $\mathbb{E}[\ell_{I^+}] = \mathbb{E}[\ell_{I^+}\mid E_I] \Pr[E_I]$ and $E[\ell_{I^{+}}| E_{I}^{C}] = 0$ that
    \begin{align*}
        \Pr[E_I] = \frac{\E[\ell_{I^+}]}{\E[\ell_{I^+} \mid E_I]} \leq \frac{20}{t} \cdot \E[\ell_{I^+}].
    \end{align*}
    
    {\bf Case 3: $\Pr[E_I] > 10 \binom{n}{d}^{-1}$ and $\Pr[\ell_{I^+} \leq t \mid E_I] \geq \frac{19}{20}$.}
    Readers are referred to \Cref{fig:edgecounts} for more illustration of the proof.
    
    Conditional on $E_I$, without loss of generality we write $[p_1, p_2] = F_I \cap W$ and let $p_3$ denote the next vertex after $F_I \cap W$ in clockwise direction.
    % we let $J \in \binom{[n]}{d-1}$ and $V_J = \conv(a_j : j \in J) \cap W$
    % be such that $V_J \notin F_I$ is the next vertex after $F_I \cap W$ in clockwise direction.
    % Note that, contrary to $I$ which is fixed, $p_3$ is a random variable.
    From the theorem's assumption we have
    $\Pr[\dist(\aff(p_1, p_2), p_3) \geq \gamma \mid  E_I] \geq \frac{1}{10}$.
    Then from the union bound,
    \begin{align*}
    \Pr[(\ell_{I^+} \leq t) &\wedge (\dist(\aff(p_1, p_2), p_3) \geq \gamma) \mid E_I]\\
    \geq 1 - \Pr[\ell_{I^+} > t \mid E_I] &- \Pr[\dist(\aff(p_1, p_2), p_3) < \gamma \mid E_I]
    \geq \frac{1}{20}.
    \end{align*}
    Referring back to \Cref{fig:edgecounts}, we know that $\theta_I \geq 0$ and thus
    \begin{align*}
        \theta_I \geq \sin(\theta_I) =  \frac{\dist(\aff(p_1, p_2), p_3)}{\ell_{I^+}}
    \end{align*}
    we have $\E[\theta_I \mid E_I] \geq \frac{1}{20} \cdot \frac{\gamma}{t}$.
    Combining this with the fact that  $\theta_{I}=0$ if $E_{I}$ does not hold, we know that
    $\E[\theta_I] = \E[\theta_I \mid E_I] \Pr[E_I]$ and we can upper bound
    \begin{align*}
        \Pr[E_I] = \frac{\E[\theta_I ]}{\E[\theta_I \mid E_I]} 
        \leq \frac{20 t}{\gamma} \cdot \E[\theta_I].
    \end{align*}

    Combining the upper bounds for each $\Pr[E_I]$ for the above three cases, we get that
    \begin{align}\label{eq:multi-num-vertices}
    \E&\left[ \edges\left(\conv(a_1, \ldots, a_n) \cap W \right) \right] = \sum_{I \in \binom{[n]}{d}} \Pr[E_I] \nonumber \\
    &\leq  \sum_{I \in \binom{[n]}{d}} \left( 10 \binom n d ^{-1} + \frac{20}{t} \cdot \E[\ell_{I^+}]
    + \frac{20t}{\gamma} \cdot \E[\theta_I]  \right) \nonumber \\
    &=  10 + \frac{20}{t} \cdot \E[\sum_{I \in \binom{[n]}{d}} \ell_{I^+}]
        + \frac{20 t}{\gamma}
        \cdot \E[\sum_{I \in \binom{[n]}{d}}\theta_I]
    \end{align}
    
    To upper bound the second term of \eqref{eq:multi-num-vertices},
    we notice that $\sum_{I \in \binom{[n]}{d}} \ell_{I^+}$ exactly equals
        the perimeter of $\conv(a_1,\dots,a_n) \cap W$. 
    Since the shadow polygon
        $\conv(a_1, \ldots, a_n)\cap W$ is contained in the two-dimensional disk of radius $\max_{i\in [n]} \|\pi_{W}(a_i)\|$,
    by the monotonicity of surface area for convex sets we have
    \[
        \E[\sum_{I \in \binom{[n]}{d}} \ell_{I^+}]
        \leq 2\pi \cdot \E[\max_{i\in [n]} \|\pi_W(a_i)\|].
    \]
        
    To upper bound the third term of \eqref{eq:multi-num-vertices}, we notice that the sum of exterior angles for any polygon always equals $2\pi$. Thus
    \begin{align}\label{eq:ub-main-third}
        \E[\sum_{I \in \binom{[n]}{d}} \theta_I] = 2\pi
    \end{align}
    
    Finally, we combine (\ref{eq:multi-num-vertices} - \ref{eq:ub-main-third}) and
    minimize over all $t > 0$:
    \begin{align*}
        \E\left[ \edges\left(\conv(a_1, \ldots, a_n) \cap W \right)  \right]
        &\leq \min_{t > 0} \left( 10 + \frac{40\pi \E[\max_{i\in [n]} \|\pi_W(a_i)\|]}{t}
        + \frac{40\pi t}{\gamma}  \right)\\
        &= 10 + 80\pi \sqrt{\frac{\E[\max_{i\in [n]} \|\pi_W(a_i)\|]}{\gamma}}.
    \end{align*}
    where in the final step, we set $t = \sqrt{\gamma \E[\max_{i\in [n]} \|\pi_W(a_i)\|]}$.
\end{proof}

% We will not apply \Cref{thm:ub-general} directly to Gaussian distributed points $a_1,\dots,a_n$.
% Instead, we will follow an approach introduced by \cite{DH18}. First, we relate
% the shadow size for Gaussian distributed vectors to the shadow size for \emph{Laplace-Gaussian}
% distributed vectors. We will then show how to use \Cref{thm:ub-general} to any log-Lipschitz probability distribution.

In the subsequent sections, we will show a lower bound for the edge-to-vertex distance $\gamma$ assuming the independently distributed vectors $a_1, \cdots, a_n$ follow Laplace-Gaussian distributions. 
This allows us to directly apply \Cref{thm:ub-general} to derive an upper bound on the expected number of edges of $\conv(a_1, \cdots, a_n) \cap W$. 
Furthermore, by using \Cref{lem:from-LG-to-Gaussian}, we can further reduce our upper bound to the case when $a_1, \dots, a_n$ are Gaussian distributed vectors.

\section{Upper Bound in Two Dimension}\label{sec:ub-2d}

In this section, we establish the smoothed complexity upper bound for $d = 2$.
For this scenario, the shadow plane $W$ encompasses the entire two-dimensional Euclidean space, and $P \cap W$ is identical to $P = \conv(a_1, \cdots, a_n)$.
From \Cref{thm:ub-general} and \Cref{lem:from-LG-to-Gaussian}, it remains to lower bound the distance from the
affine hull of an edge to its neighboring vertex in clockwise order (denoted by 
$\gamma$ in \Cref{thm:ub-general}), where the vertices of the polygon $a_1, \cdots, a_n$ is sampled from a Laplace-Gaussian distribution with the center of $\bar{a}_1, \cdots, \bar{a}_n$.
% $\conv(a_1, \ldots, a_n)$ is  Laplace-Gaussian perturbation.
We will demonstrate a slightly stronger result: a lower bound for the distance between the
affine hull of an edge to all of the remaining $(n-2)$ vertices.

\begin{lemma}[Edge-to-vertex distance in Two Dimension]\label{lem:dist-2d}
    Let $a_1, \ldots, a_n \in \R^2$ be $n$ independent $L$-log-Lipschitz random variables.
    Then, for any $I \in \binom{[n]}{2}$, the outer unit normal
    $y \in W$ of the edge $\conv(a_i : i\in I)$ satisfies
    $$\Pr[y^\top a_i - \max_{j \notin I} y^\top a_j \geq \frac{1}{L} \mid E_I] \geq 0.1,$$
    for any $i \in I$.
\end{lemma}

By applying \Cref{thm:ub-general}, \Cref{lem:from-LG-to-Gaussian}, and the Laplace-Gaussian tail bound of \Cref{lem:properties-laplace-gaussian}, 
we find the following upper bound for two-dimensional polygons under Gaussian perturbation:
\begin{theorem}[Two-Dimensional Upper Bound]\label{thm:ub-2d}
    Let $\Bar{a}_1, \ldots, \Bar{a}_n \in \R^2$ be $n > 2$ vectors with norm at
    most $1$. For each $i \in [n]$, let $a_i$ be independently distributed as
    $\mathcal{N}_2(\Bar{a}_i, \sigma^2 I)$. Then \begin{align*}
        \E\left[ \edges\left(\conv(a_1, \ldots, a_n) \right) \right] \leq O\left(\frac{\sqrt[4]{\log n}}{\sqrt{\sigma}} + \sqrt{\log n}\right).
    \end{align*}
\end{theorem}

\begin{proof}
    For each $i \in [n]$, let $\hat{a}_i$ be independently sampled form the 2-dimensional Laplace-Gaussian distribution $LG_2(\bar{a}_i, \sigma, 4\sigma \sqrt{2\log n})$. It follows from
    \Cref{lem:properties-laplace-gaussian} that
    $\hat{a}_i$ is $(4\sigma^{-1} \sqrt{2\log n})$-log-Lipschitz
    and $\E[\max_{i \in [n]}\|\hat a_i\|] \leq 1 + 4\sigma\sqrt{2\log n}$.
    We use \Cref{lem:dist-2d}, setting $L = 4\sigma^{-1} \sqrt{2\log n}$, and
    \Cref{thm:ub-general}, setting $\gamma = \frac{1}{L} = \frac{\sigma}{4\sqrt{2\log n}}$, to find
    \begin{align*}
        \E\left[ \edges\left(\conv(\hat{a}_1, \ldots, \hat{a}_n) \right) \right] 
         &\leq O\left(\sqrt{\frac{\sqrt{\log n}}{\sigma} + \log n}\right) \leq  O\left(\frac{\sqrt[4]{\log n}}{\sqrt{\sigma}} + \sqrt{\log n}\right).
    \end{align*}
    Finally, from \Cref{lem:from-LG-to-Gaussian}, we conclude that $
        \E\left[ \edges\left(\conv(a_1, \ldots, a_n) \right) \right] \leq 1 + O(\frac{\sqrt[4]{\log n}}{\sqrt{\sigma}} + \sqrt{\log n}).
    $
\end{proof}

\begin{proof}[Proof of Lemma~\ref{lem:dist-2d}]
    Fix any set $I = \{i, i'\} \subset [n]$.
    Define $z \in \mathbb S^1$ and $t$ to satisfy $z^\T a_i = z^\T a_{i'} = t$
    and $z^\T e_1 > 0$. Both are well-defined with probability $1$.

    Note that $E_I$ is now equivalent to either having $z^\T a_j < t$ for all $j \notin I$
    or having $z^\T a_j > t$ for all $j \notin I$. Write $E_I^+$ for the former case and $E_I^-$ for the latter.
    The vector $z$ is always defined, assuming non-degeneracy, and is equal to
    the outer normal unit vector $y$ conditional on $E_I^+$ and equal to $-y$ conditional on $E_I^-$.

    We want to apply the principle of deferred decisions to fix the values of $a_j$ for each $j \notin I$
    and $z$, and now only allow $a_{i}$ and $a_{i'}$ to vary while fixing the slope of the line between $a_{i}$ and $a_{i'}$.
    Let $\mu: \R \to \R_{\geq 0}$ denote the induced probability density function of $t = y^\top a_i =  y^\top a_{i'}$.
    \Cref{fact:joint-t-log-lipschitz} tells us that $\mu$ is $(2L)$-log-Lipschitz. 

    In the first case, for $E_I^+$, we have, still only considering the randomness over $t$,
    \begin{align*}
         & \Pr[ (t - \max_{j \notin I} z^\top a_j \geq \frac{1}{L}) \wedge E_I^+] \\
        = & \int_{\max_{j \notin I} z^\top a_j + 1/L}^\infty \mu(t) \mathrm dt \\
        = & \int_{\max_{j \notin I} z^\top a_j}^\infty \mu(s+1/L) \mathrm ds \\
         \geq & \int_{\max_{j \notin I} z^\top a_j}^\infty e^{-2} \mu(s) \mathrm ds \tag{By $(2L)$-log-Lipschitzness of $\mu$} \\
        = & e^{-2} \Pr[E_I^+].
    \end{align*}
    Similarly for the other case, $E_I^-$, we find
    \[
    \Pr[ (\min_{j \notin I} z^\top a_j - t \geq \frac{1}{L}) \wedge E_I^-]
    \geq e^{-2}\Pr[E_I^-].
    \]
    Now observe that, for any $i \in I$, we have
    \begin{align*}
        &\Pr[y^\top a_i - \max_{j \notin I} y^\top a_j \geq \frac{1}{L} \wedge E_I] \\
        =& \Pr[ (t - \max_{j \notin I} z^\top a_j \geq \frac{1}{L}) \wedge E_I^+]
        + \Pr[ (\min_{j \notin I} z^\top a_j  - t\geq \frac{1}{L}) \wedge E_I^-]\\
        \geq & e^{-2} \Pr[E_I^+] + e^{-2}\Pr[E_I^-] = e^{-2}\Pr[E_I].
    \end{align*}
    This finishes the proof since
    \[
    \Pr[y^\top a_i - \max_{j \notin I} y^\top a_j \geq \frac{1}{L} \mid E_I]
    =\Pr[y^\top a_i - \max_{j \notin I} y^\top a_j \geq \frac{1}{L} \wedge E_I]/\Pr[E_I]
    \geq e^{-2} \geq 0.1.
    \]
\end{proof}

\section{Multi-Dimensional Upper Bound}\label{sec:ub-multi-new}

In this section, we establish the upper bound for the higher-dimensional case (i.e., $d \geq 3$):
\begin{theorem}[Multi-dimensional Upper Bound]\label{thm:ub-multi}
    Let $d > 2, n \geq d$, and $\sigma \leq \frac{1}{16d\sqrt{\log n}}$.
    Let $\bar{a}_1, \ldots, \bar{a}_n$ be $n$ vectors with $\max_{i \in [n]} \|\bar{a}_i\| \leq 1$. For each $i \in [n]$, let $a_i$ be independently distributed as
    $\mathcal{N}_d(\Bar{a}_i, \sigma^2 I)$. Then
    \begin{equation}\label{eq:ub-multi}
        \E[\edges( \conv(a_1,\dots,a_n) \cap W)] = O\left( \sigma^{-3/2} d^{13/4} \log^{5/4} n \right).
    \end{equation}
\end{theorem}

Similar to the two-dimensional case (see \Cref{sec:ub-2d}), the main technical ingredient of \Cref{thm:ub-multi} is a lower-bound of the edge-to-vertex distance (the quantity $\gamma$ in \Cref{thm:ub-general}) on the shadow polygon:

\begin{lemma}[Edge-to-vertex distance of shadow polygon in multi-dimension]\label{lem:normalvector}
    For any $d \geq 3$,
    let $a_1, \ldots, a_n \in \R^d$ be independent $L$-log-Lipschitz random variables.
    For any $I \in \binom{[n]}{d}$ that satisfies
    \(
    \Pr[E_I] \geq 10\binom{n}{d}^{-1},
    \)
    (where $E_I$ is defined in \Cref{def:edge-event}),
    we have
    \begin{align*}
        % \Pr[y^\top p \geq y^\top p' +  \frac{1}{1.07 \cdot 10^6 \cdot L^3d^5 \log n}\mid E_I ] &\geq 0.1, \\
        \Pr[y^\top p - y^\top p' \geq  \Omega(\frac{1}{L^3d^5 \log n})\mid E_I ] \geq 0.1,
    \end{align*}
    where $p$ is any point in $F_I \cap W$, and $p' \in \conv(a_1,\dots,a_n)\cap W$ is the next vertex after $F_I \cap W$ in clockwise direction.
    Here $y \in W$ is the outer unit normal to the edge $F_I \cap W$ on $\conv(a_1, \ldots, a_n) \cap W$.
\end{lemma}

\Cref{thm:ub-multi} then immediately follows from \Cref{lem:normalvector}, \Cref{thm:ub-general},  and \Cref{lem:from-LG-to-Gaussian}:

\begin{proof}[Proof of Theorem~\ref{thm:ub-multi}]
    For each $i \in [n]$, let $\hat{a}_i$ be independently sampled from the Laplace-Gaussian distribution $LG_d(\bar{a}_i, \sigma, 4\sigma \sqrt{d\log n})$.
    % In the following analysis, we will abuse the notation $F_I$, $E_I$ (see \Cref{def:edge-eve='[;''nt}) on vertices $\hat{a}_1, \ldots, \hat{a}_n$.
    From \Cref{lem:properties-laplace-gaussian}, we know that
    \begin{enumerate}
        \item Each $\hat{a}_i$ is $L = (4\sigma^{-1} \sqrt{d\log n})$-log-Lipschitz;
        \item $\E[\max_{i\in [n]}\|\pi_W(\hat{a}_i)\|] \leq 1 + 8\sigma \sqrt{2 \log n} \leq 1.5$.
    \end{enumerate}

    From \Cref{lem:normalvector}, we get that for any $p \in F_I \cap W$, if $p'$ is the next vertex after the edge $F_I \cap W$ in clockwise order, then
    \begin{align*}
        & \Pr[y_I^\top p \geq y_I^\top p' + \Omega(\frac{1}{L^3d^5\log n}) \mid E_I] \geq  0.1.
    \end{align*}
    Here $y_I \in W$ is the outer unit normal vector of the polygon $\conv(\hat{a}_1, \ldots, \hat{a}_n) \cap W$ on the edge $F_I \cap W$.
    Then we can use
    \Cref{thm:ub-general} by setting $L = 4\sigma^{-1} \sqrt{d\log n}$, $\gamma = \Omega(\frac{1}{L^3d^5 \log n})$
    % $\gamma = \frac{1}{1.07 \cdot 10^6 \cdot L^3d^5 \log^2 n}$
    and $\E[\max_{i \in [n]}\|\pi_W(a_i)\|]  = 1.5$, to find
    \begin{align*}
        \E\left[ |\edges\left(\conv(\hat{a}_1, \ldots, \hat{a}_n) \right) |\right]
        &\leq 10  + O(\sqrt{\frac{E[\max_{i \in [n]}\|\pi_W(a_i)\|]}{\gamma}}) \\
        &\leq 10 + O(\sqrt{\frac{1.5}{\frac{1}{L^3d^5 \log n}}}) \\
        &\leq 10 + O(\sqrt{L^3d^5 \log n}) \\
         &\leq  10 + O(\sqrt{\sigma^{-3} d^{13/2}  \log^{5/2} n}).
    \end{align*}
    Finally, from \Cref{lem:from-LG-to-Gaussian}, we conclude that
    \begin{align*}
        \E\left[ |\edges\left(\conv(a_1, \ldots, a_n) \cap W \right) |\right] &\leq  11 + O(\sqrt{\sigma^{-3} d^{13/2}  \log^{5/2} n}) \\
        &= O\left(\sigma^{-3/2} d^{13/4} \log^{5/4} n  \right) .\qedhere
    \end{align*}
\end{proof}

The rest of this section is dedicated to the proof of \Cref{lem:normalvector} and will be structured as follows.
In \Cref{sub:multi-notation} we define some basic notation that will be used in the proof.
In \Cref{sub:det} we establish two sufficient criteria for the conclusion of \Cref{lem:normalvector} to hold.
In \Cref{sub:rand-delta-lb} and \Cref{sub:rand-r-lb}, we prove that these conditions hold with good probability conditional on $E_I$.
\Cref{sec:pf-rand-r-b} to \Cref{sec:pf-facet-dist} include the proof of the auxiliary lemmas.
Finally, we finish the proof of \Cref{lem:normalvector} in \Cref{sub:combine}.

\subsection{Notations}\label{sub:multi-notation}

Since we assume that the constraint matrix rows $a_1,\dots,a_n$ each have a continuous probability density function,
$\conv(a_1,\dots,a_n)$ and $W$ satisfy the non-degeneracy conditions (see \Cref{lem:non-degeneracy}) almost surely.
In this case, each edge of the polygon $\conv(a_1, \ldots, a_n) \cap W$ is
given by $F_I \cap W = \conv(a_i: i\in I)$ for which $I \in \binom{[n]}{d}$ and $E_I$ holds
(where $F_I$ and $E_I$ are defined in \Cref{def:edge-event}).  In addition, each vertex
of the polygon $\conv(a_1, \ldots, a_n) \cap W$ is given by the intersection
between $W$ and $(d-2)$-dimensional ridges of $\conv(a_1, \ldots, a_n)$, which
are convex hulls of $(d-1)$ vertices of $\conv(a_1, \ldots, a_n)$.
We define the following notations for a ridge and its corresponding vertex:

\begin{definition}[Ridge and vertex event]\label{def:vertex-event}
    For any $J \subset [n]$, write $R_J = \conv(a_j : j\in J)$.
    Define $A_J$ to be the event that $R_J$ is a ridge of $\conv(a_1, \ldots, a_n)$
    and $R_J \cap W \neq \emptyset$.
\end{definition}

\begin{remark}
    Any vertex $v$ of $\conv(a_1,\dots,a_n)$ can be written as $v = R_J \cap W$
    for some $J \subset [n]$ for which $A_J$ holds.
    Assuming non-degeneracy, each $J$ for which $A_J$ holds satisfies $|J| = d-1$
    and the relation between vertices and index sets $J\in\binom{[n]}{d-1}$ with $A_J$
    is a one-to-one correspondence.
\end{remark}

\subsection{Deterministic Conditions for a Good Edge-to-Vertex Separator} \label{sub:det}

In this subsection, we present a series of sufficient conditions such that an edge on the polygon $\conv(a_1, \ldots, a_n) \cap W$ maintains a significant separation from its next vertex in clockwise order. When the vertex set $\{a_i\}_{i=1}^n$ is fixed, this edge-to-vertex distance can be decomposed into two geometric components:
\begin{itemize}
    \item the distance from all other vertices of $Q$ to the supporting hyperplane of the facet containing the edge, and
    \item the depth at which the intersection point $p = R \cap W$ lies in the interior of the ridge $R$.
\end{itemize}
Lemma~\ref{lem:det} shows that if these two quantities are bounded below by $\delta$ and $r$ respectively, then their product, $r\delta/3$, guarantees a significant edge-to-vertex gap in the projected polygon $Q \cap W$.

\begin{lemma}\label{lem:det}
Let $W \subset \R^d$ be a two-dimensional linear subspace,
$Q = \conv(a_1, \ldots, a_n) \subset \R^d$ be a non-degenerate polytope
with a non-degenerate intersection with $W$
such that $\max_{i, j\in [n]} \|a_i - a_j\| \leq 3$ and $W \cap Q \neq \emptyset$.
Fix any facet $F$ of $Q$ such that $F \cap W \neq \emptyset$
and any ridge $R \subset F$ of $F$ such that $W \cap R$ is a singleton set $\{p\}$.
Let $\delta, r \geq 0$ be such that
\begin{enumerate}
    \item (distance between $F$ and other vertices) $\forall a_k \notin F, \dist(\aff(F), a_k) \geq \delta$;
    \item (Inner radius of $R$) $\dist(\aff(F \cap W), \partial R) \geq r$.
\end{enumerate}
Then for any $p \in F \cap W$, the outer unit normal vector $\bar{\theta} \in W$
to the edge $F \cap W$ satisfies
\[
\bar{\theta}^\T p - \bar{\theta}^\top p' \geq \delta r/3,
\]
where $p' \in Q \cap W$ is the next vertex after $F \cap W$ in clockwise order.
\end{lemma}

The reader who desires a more intuitive illustration of the geometry involved
may take a look at \Cref{fig:dist-cone} from the next lemma, where relevant concepts are depicted
as they happen for $d=4$. The simplex is the ``next'' facet after $F$, and its bottom face $B$
is the ridge of $Q$ that is shared with $F$.
The distance from $b_1$ to $B$ is large due to the first assumption of \Cref{lem:det},
and the distance from $q$ to the boundary of $B$ is large due to the second assumption.
To map the depicted facet to a three-dimensional space for the purpose of the illustration,
it has been projected orthogonally to the subspace perpendicular to $F \cap W$.

We remark that \Cref{lem:det} gives a sufficient condition assuming that $\{a_{i}: i \in [n]\}$ is fixed. In later subsections, our goal is to show this sufficient condition actually occurs with high probability even when $\{a_{i}: i \in [n]\}$ are Laplace-Gaussian distributed random variables.

To start, we show a lemma about the distance from a point in the simplex to the affine hull of its neighboring facet.
% lower-dimensional lemma about distances in a simplex.
\begin{lemma}\label{lem:rand-r-lb-cone}
    Given $b_1, \ldots, b_d \in \R^{d-1}$ such that $\conv(b_1, \ldots, b_d)$ is non-degenerate. For any $D > 0$, suppose
    \begin{enumerate}
        \item $\forall i, j \in [d]$, $\|b_i - b_j\| \leq D$;
        \item $\dist(b_1, \aff(b_2, \ldots, b_d)) \geq \delta$;
        \item There exists $q \in \conv(b_2, \ldots, b_d)$ such that $\dist(q, \partial(\conv(b_2, \ldots, b_d))) \geq r$.
    \end{enumerate}
    Then we have $\dist(q, \aff(b_1, \ldots, b_{d-1})) \geq r\delta / D$.
\end{lemma}
\begin{proof}
    For simplicity, write $B = \conv(b_2, \ldots, b_{d})$ and $B' = \conv(b_1, \ldots, b_{d-1})$.
    Let $q' = \pi_{B'}(q)$ be the point closest to $q$ on $\aff(B')$, i.e., $\|q-q'\|=\dist(q, \aff(b_1, \ldots, b_{d-1}))$.

    Let $x = (B \cap B') \cap \aff(b_1, q, q')$ be its intersection between the two-dimensional plane $\aff(b_1, q, q')$ and the $(d-3)$-dimensional ridge $B \cap B'$ (which gives a unique point). (See \Cref{fig:dist-cone} for an illustration).
    \begin{figure}[h]
    \centering
       \includegraphics[width=0.5\linewidth]{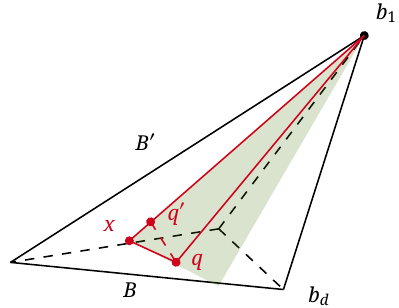}
    \caption{Illustration of \Cref{lem:rand-r-lb-cone} when $d-1 = 3$. In light green is the intersection between the two-dimensional plane $\aff(b_1, q, q')$ and $\conv(b_1, \ldots, b_d)$. The red triangle is $\conv(b_1, x, q)$. The bottom face is $B$ and the left-facing back face is $B'$.}
    \label{fig:dist-cone}
    \end{figure}
    Consider the triangle $\conv(b_1, q, x)$ and calculate its area
    in two different ways.
    On one hand, it has base $\conv(b_1, x)$ of length $\|b_1 - x\| \leq D$ with height
    $\dist(q, \aff(b_1, x)) = \|q - q'\|$, which gives that the area of the triangle is
    at most $D \|q - q'\|/2$.
    On the other hand, this triangle has base $\conv(x, q)$ of length $\|x - q\| \geq \dist(q, \partial(B)) \geq r$ with height $\dist(b_1, \aff(x, q)) \geq \dist(b_1, \aff(B)) \geq \delta$, which gives that the area of the triangle is at least $r\delta/2$.

    Combining the above two ways of determining the area of triangle $\conv(b_1, q, x)$, we have $D \|q - q'\|/2 \geq r\delta/2$.
    Therefore we have $\dist(q, \aff(B')) = \|q - q'\| \geq \frac{r\delta}{R}$ as desired.
\end{proof}

What \Cref{lem:rand-r-lb-cone} tells us is that, if we have sufficiently strong information about the geometry in base of a simplex (conditions 1 and 2)
then we can relate the height of that simplex to the length of a given chord. In our setting, we let a ridge of the polar polytope play the role of the base, and take the apex to be one of the vertices $a_i$ not contained in the chosen facet.

To prove our main deterministic geometric result (Lemma~\ref{lem:det}), we \emph{squash} the configuration by projecting it orthogonally onto the $(d-1)$-dimensional subspace $s^\perp$ orthogonal to the edge direction. Under this projection, the two neighboring ridges $R$ and $R'$ become adjacent facets of a simplex, with the apex corresponding to the first vertex outside the facet~$F$. The bounded diameter, height, and inner radius conditions in Lemma~\ref{lem:det} then align precisely with the assumptions of Lemma~\ref{lem:rand-r-lb-cone}, which allow us to apply it and lift the resulting bound back to the original two-dimensional slice~$W$.

\begin{proof}[Proof of Lemma~\ref{lem:det}]

    \begin{figure}[h]
      \centering
             \includegraphics[width=0.6\linewidth]{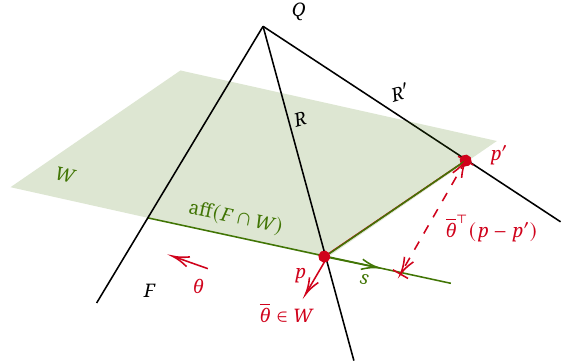}
    \caption{Illustration of the variables in \Cref{lem:det} for $d = 3$. The light-green parallelogram is the two-dimensional plane $W$.
    The red arrows are $\theta$ (outward unit normal of $F$) and $\bar{\theta} = \pi_W(\theta) / \|\pi_W(\theta)\|$.
    The red points $p = W \cap R$ and $p' = W \cap R'$ are two consecutive points on the polygon $Q \cap W$. The red dashed line demonstrates the edge-to-vertex distance $\bar{\theta}^\top (p - p')$.}
    \label{fig:det-dist}
    \end{figure}

Write $R'$ for the ridge of $Q$ such that
$\{p'\} = R' \cap W$.
Since $p' \in Q \cap W$ is
adjacent to vertex $p$ and the edge $F \cap W$,
by \Cref{fact:diff-adjacent-vertex} we may relabel the $a_i$ such that
$R' = \conv(a_1,\dots,a_{d-1})$,
$R = \conv(a_2,\dots,a_d),$ and $F = \conv(a_2,\dots,a_{d+1})$ without loss of generality.
Let $\theta \in \sfe$ denote the outward unit normal to $F$.
This normal vector satisfies
\[
    \label{eq:ddim-dist}
    \delta \leq \min_{\substack{i \in [n] \\ a_i \notin F}} \theta^T( p - a_i)
    \leq \theta^\T( p - a_1).
\]
Let $s \in \mathbb{S}^{d-1}$ be a unit vector indicating the direction of the (one-dimensional) line $F \cap W$. This vector is unique up to sign.
Also, let $\bar \theta = \pi_W(\theta)/\|\pi_W(\theta)\|$
be the outward unit normal to $F \cap W$ in the two-dimensional plane $W$.
See \Cref{fig:det-dist} for an illustration of the variables for $d = 3$.
Notice that $\bar{\theta}$ and $s$ form an orthonormal basis of $W$.
Therefore we get
    \begin{align}
        \label{eq:2d-dist}
        \bar \theta^\T (p - p')
        = \bar{\theta}^\T \pi_{s^\perp} (p - p')
        = \|\pi_{s^\perp} (p - p')\|
    \end{align}
Here the last equality comes from $(p - p') \in W = \mathrm{span}(\bar{\theta}, s)$, so $\pi_{s^\perp} (p - p') = \pi_{\bar \theta} (p - p')$.

Now we focus on the $(d-1)$-dimensional space $s^\perp$, and consider the orthogonal projections $\pi_{s^\perp}(a_1), \ldots, \pi_{s^\perp}(a_d)$.
Since the diameter of $\conv(a_1, \ldots, a_d)$ is at most $3$, we have $\max_{i, j \in [d]} \|\pi_{s^\perp}(a_i) - \pi_{s^\perp}(a_j)\| \leq 3$.
By definition, $\theta$ is a unit normal of $R$, so since $s \in R$, we have $\theta \in s^{\perp}$. It follows that $\theta$ is also a unit normal of $\pi_{s^\perp}(R) = \pi_{s^\perp}(\conv(a_2, \ldots, a_d))$. This gives
$$\dist(\pi_{s^\perp}(a_1), \aff(\pi_{s^\perp}(R))) = \theta^\top (p - a_1) \geq \delta.$$
Also, since $\dist(F \cap W, \partial R) \geq r$ and $F \cap W$ is one-dimensional, after the projection to $s^\perp$ we have
$$\dist(\pi_{s^\perp}(p), \partial \pi_{s^\perp}(R)) = \dist(\aff(F \cap W), \partial R) \geq r.$$
Therefore we can use \Cref{lem:rand-r-lb-cone} to get
\begin{align*}
    \|\pi_{s^\perp}(p) - \pi_{s^\perp}(p')\| \geq \dist(\pi_{s^\perp}(p), \aff(\pi_{s^\perp}(R'))) \geq r \delta /3,
\end{align*}
where the first step comes from $\pi_{s^\perp}(p') \in \aff(\pi_{s^\perp}(R'))$.
The lemma then follows from \eqref{eq:2d-dist}.
\end{proof}

\subsection{Randomized Lower-Bound for \texorpdfstring{$\delta$}{delta}: Distance between vertices and facets}\label{sub:rand-delta-lb}

In this section, we show that the affine hull of a given facet $F$ of the polytope
$\conv(a_1, \ldots, a_n)$ is $\Omega(\frac{1}{Ld\log n})$-far away to other
vertices with good probability, or in other words, the distance $\delta$ in
\Cref{lem:det} is at least $\Omega(\frac{1}{Ld\log n})$ with good probability.
Our main result of this section is as follows:

\begin{lemma}[Randomized lower-bound for $\delta$]\label{lem:rand-delta-lb}
    Let $a_1,\dots,a_n \in \R^d$ be independent $L$-log-Lipschitz random vectors.
    For any $I \in \binom{[n]}{d}$ such that $\Pr[E_I] \geq 10\binom{n}{d}^{-1}$, we have
    \[
        \Pr[\min_{k \in [n] \backslash I}\dist(\aff(F_I), a_k) \geq \frac{1}{10e^3 dL \log n}) \mid E_I] \geq 0.72.
    \]
\end{lemma}
\paragraph{Intuition}
The proof of this lemma will span this entire subsection.
Let us start with some words on the intuition behind it.
Assume $\aff(F_I)$ is fixed arbitrarily. Then the quantities $\dist(\aff(F_I),a_k)$
are determined solely by the points $a_k$, $k \in [n]\setminus I$.
The points are $L$-log-Lipschitz, which makes each signed distance $\abs{\theta^\top a_k} = \dist(\aff(F_I), a_k) \in \R$
into an $L$-log-Lipschitz random variable as well.
Any $L$-log-Lipschitz random variable has its probability density function pointwise upper bounded by $L$,
and hence the probability that for a given $k \in [n]\setminus I$ we have
$\Pr[\theta^\top a_k \in [-\eps, \eps]] \leq 2\eps L$.
A union bound would then give
\[\Pr[\min_{k \in [n] \backslash I}\dist(\aff(F_I), a_k)\leq\eps] \leq 2(n-d)\eps L.\]
Taking $\eps = \frac{0.28}{2(n-d)L}$ would give us a bound on the probability.

It is clear that this weaker version of \Cref{lem:rand-delta-lb} would be relatively easy to prove.
However, it has a linear dependence on $n$ and thus it would add a factor of $\sqrt{n}$ to our shadow bound.
This is undesirable.
To obtain the stronger conclusion, we consider the expected number of points
at close distance to $\aff(F_I)$.
\begin{equation}
\Pr\Big[\min_{k \in [n] \backslash I}\dist(\aff(F_I), a_k) \leq \eps ~\Big|~ E_I\Big]
\leq
\E\Big[\abs{\{i \in [k]\setminus I : \dist(\aff(F_I),a_k) \leq \eps \}} ~\Big|~ E_I\Big].\label{eq:expectedcloseby}
\end{equation}
If $\eps \leq 1/L$ then we can use $L$-log-Lipschitzness to derive a lower bound on the number of points $a_k$
lying above (or below) the affine subspace $\aff(F_I)$.
What we prove is that if \eqref{eq:expectedcloseby} is large then (without conditioning on $E_I$)
both the expected number of points above $\aff(F_I)$ and the expected number of points below $\aff(F_I)$
are at least $\Omega(\frac{1}{\eps L})$ times as large as \eqref{eq:expectedcloseby}.

However, recall that the event $E_I$ requires all points to lie on the same side of $\aff(F_I)$.
If there is simultaneously a point above $\aff(F_I)$ and a point below $\aff(F_I)$ then $E_I$ does not hold.
Using the Chernoff bound we can show, conditional on some $\aff(F_I) = H$, that if we have
$\Pr[E_I \mid \aff(F_I) = H] \geq n^{-d}$
then at least one of the expected number of points above $\aff(F_I)$ or the expected number of points below
$\aff(F_I)$ must be at most $2d \log n$.
If one of these is bounded from above, then \eqref{eq:expectedcloseby} must be bounded from above.
Taking proper care to observe that those affine subspaces $H$ for which $\Pr[E_I \mid \aff(F_I) = H] \geq n^{-d}$
together account for most of the probability mass, this will yield the desired result.

\paragraph{Proving the lemma}
To show \Cref{lem:rand-delta-lb}, we fix any $I \in \binom{[n]}{d}$ of consideration.
Without loss of generality, assume $I = [d]$ and write $E = E_{[d]}$.
We define the following event $B_{\epsilon}$ indicating that the distance from
$\aff(F_{[d]})$ to other vertices is at least $\epsilon$.

\begin{definition}[Separation by the margin of a facet]\label{def:Beps}
    Let $\theta \in \mathbb S^{d-1}, t \in \R$ be as in \Cref{def:change-of-variables}.
For any $\eps > 0$, let $B_\eps^+$ denote the event that
$\theta^\T a_i < t - \eps$ for all $i \in [n]\setminus[d]$
and $B_\eps^-$ denote the event that
$\theta^\T a_i > t + \eps$ for all $i \in [n]\setminus[d]$.
We write $B_\eps = B_\eps^+ \vee B_\eps^-$.
\end{definition}

In the following lemma, we show that for sufficiently small $\epsilon$, $\Pr[E \wedge B_{\epsilon}]$ is still a constant fraction of $\Pr[E]$.

\begin{lemma}\label{lem:delta-prob-compare}
For any $0 < \eps \leq \frac{1}{10e^3 L d\log n}$ it holds that
\[
\Pr[E] \leq \binom{n}{d}^{-1} + \frac{5}{4} \cdot \Pr[E \wedge B_\eps].
\]
\end{lemma}
\begin{proof}
    Writing random variables as subscripts to denote which expectation is over which variables,
    we start by using Fubini's theorem to write
    $$\Pr_{a_1,\dots,a_n}[E] = \E_{a_1,\dots,a_d}[ \Pr_{a_{d+1},\dots,a_n}[E]].$$
    Fix any $a_1,\dots,a_d \in \R^n$ subject to \mbox{$\conv(a_1,\dots,a_d) \cap W \neq \emptyset$} and the non-degeneracy assumptions in \Cref{lem:non-degeneracy}.
    Define $\theta \in \sfe, t > 0$ as described in \Cref{def:change-of-variables}, i.e., $\theta^\top a_i = t$ for each $i \in [d]$.
    Write $s_i = \theta^\T a_i$ for each $i \in [n]\setminus [d]$.
    We note that $s_i$ is an $L$-log-Lipschitz random variable for all $i \in [n] \backslash [d]$.
    Moreover, over the remaining randomness in $a_{d+1},\dots,a_n$, we have
    $\Pr[E] = \Pr[B_0^+] + \Pr[B_0^-]$ and $\Pr[B_\eps] = \Pr[B_\eps^+] + \Pr[B_\eps^-]$.
    We will show that
    \begin{equation}
        \Pr[B_0^+] \leq \frac{1}{2\binom{n}{d}} + \frac 5 4 \Pr[B_\eps^+] \label{eq:delta-prob-compare-eq}
    \end{equation}
    and the appropriate statement will follow for $B_\eps^-$ analogously.
    Putting together this will prove the lemma.
    If $\Pr[B_0^+] \leq \frac{1}{2}\binom n d^{-1}$ then the desired inequality holds directly.
    
    In order to prove \eqref{eq:delta-prob-compare-eq} we require the following claim:
    \begin{claim}\label{claim:highprobabilityimpliesfewnearbypoints}
        Conditional on $\theta, t$, if $\Pr[B_0^+ \mid \theta, t] \geq n^{-d}$
        then $\E[\#\{i \in [n] \setminus [d] : s_i \geq t\}] \leq 2 d\log n$.
        If $\Pr[B_0^- \mid \theta, t] \geq n^{-d}$
        then $\E[\#\{i \in [n] \setminus [d] : s_i \leq t\}] \leq 2 d\log n$.
    \end{claim}
    \begin{proof}
        We prove the first implication, and the second follows analogously.
        For each $i \in [n]\setminus[d]$, let $X_i \in \{0,1\}$ have value $1$
        if and only if $s_i \geq t$. Since $\theta, t$ are fixed
        and depend only on $a_1,\dots,a_d$, the random variables
        $X_{d+1},\dots,X_n$ are independent.
        Write $X = \sum_{i=d+1}^n X_i$. The Chernoff bound gives
        \[\Pr[X = 0]
        \leq \exp\left(-\frac{\E[X]}{2}\right).\]
        As such, \(\E[X] > 2d\log n\) would imply \(\Pr[X = 0] < n^{-d}\),
        contradicting the original assumption that \(\Pr[X = 0] \geq n^{-d}.\)
        It follows that \(\E[X] \leq 2d\log n\).
    \end{proof}

    Thus, in what remains, we may suppose that $Pr[B_{0}^{+}] > 1/2 \binom{n}{d}^{-1}$.
    Fix any $i \in [n] \backslash [d]$ and let $\mu_i$ denote the induced probability density function of $s_i$. We then have a sequence of inequalities as found below.    
    The first two inequalities above follow from $L$-log-Lipschitzness of $\mu_i$. 
    For the first inequality in particular, note that for any $s \in [-\frac{1}{L}, 0]$,
    we have $\abs{s-\eps L s} \leq 1/L$ from $\eps \leq 1/L$.
    This then gives $\frac{\mu_i (t + \epsilon L s)}{\mu_i (t + s)} \leq \exp(L \cdot (1 - \epsilon L) s) \leq 1$.
    \begin{align}
        \Pr[ s_i \geq t - \eps \mid s_i \leq t]
        &= \frac{\int_{-\eps}^0 \mu_i(t + s) \mathrm d s}{\int_{-\infty}^0 \mu_i(t + s) \mathrm d s} \nonumber\\
        &= \frac{\eps L \int_{-1/L}^0 \mu_i(t + \eps L s) \mathrm d s}{\int_{-\infty}^0 \mu_i(t + s) \mathrm d s} \nonumber\\
        &\leq e \frac{\eps L \int_{-1/L}^0 \mu_i(t + s) \mathrm d s}{\int_{-\infty}^0 \mu_i(t + s) \mathrm d s}\nonumber \\
        &= e \frac{\eps L \int_{0}^{1/L} \mu_i(t + s - 1/L) \mathrm d s}{\int_{-\infty}^{1/L} \mu_i(t + s - 1/L) \mathrm d s}\nonumber \\
        &\leq e^3 \frac{\eps L \int_{0}^{1/L} \mu_i(t + s) \mathrm d s}{\int_{-\infty}^{1/L} \mu_i(t + s) \mathrm d s} \nonumber\\
        &= e^3 \eps L \Pr[s_i \geq t \mid s_i \leq t + 1/L] \nonumber\\
        &\leq e^3 \eps L \Pr[s_i \geq t]. \label{eq:prob-compare-1}
    \end{align}
    The third inequality, on the final line, follows from the fact that $s_i \geq t + 1/L$ implies $s_i \geq t$,
    and hence $\Pr[s_i \geq t \mid s_i > t+1/L] = 1$.
    As such we can, for fixed $t,\theta,$ upper-bound the probability over $s_1,\dots,s_d$ that,
    conditional on $B_0^+$, there exists a vertex being $\epsilon$-close to $\aff(F_I)$:
    \begin{align}
        \Pr[\neg B_\eps^+\mid B_0^+]
    &= \Pr[\exists i \in [n]\setminus [d] : s_i \geq t - \eps \mid B_0^+] \nonumber \tag{By union bound} \\
    &\leq \sum_{i \in [n] \setminus [d]} \Pr[s_i \geq t - \eps \mid B_0^+] \nonumber \\
    &\leq \sum_{i \in [n] \setminus [d]} e^3 \eps L \Pr[s_i \geq t ] \tag{By \eqref{eq:prob-compare-1}} \nonumber \\
    % &= \E[\#\{i \in [n]\setminus [d] : s_i \geq t - \eps\}\mid B_0^+, \theta, t] \\
      &= e^3 \eps L \; \E[\#\{i \in [n] \setminus [d] : s_i \geq t \}]. \label{eq:prob-compare-2}
    \end{align}
    To interpret the last equality above, we observe that
    $\#\{i \in [n] \setminus [d] : s_i \geq t \} = 0$ if and only if $B_0^+$ happens.
    Then by applying \eqref{eq:prob-compare-2} to \Cref{claim:highprobabilityimpliesfewnearbypoints} (note that we are using the assumption $Pr[B_{0}^{+}] > 1/2 \binom{n}{d}^{-1}$) with our choice of $\eps$
    % The upper bound $\Pr[B_0^+] \leq \frac 5 4 \Pr[B_\eps^+]$ then follows
    % from \eqref{eq:prob-compare-2} By applying \Cref{claim:highprobabilityimpliesfewnearbypoints}
    % and our choice of $\eps$.
    we conclude that 
    \[
    \Pr[B_{0}^{+}] \leq \frac{5}{4} \Pr[B_{\varepsilon}^{+}] \leq \frac{1}{2} \binom{n}{d}^{-1} + \frac{5}{4} \Pr[B_{\varepsilon}^{+}].\qedhere
    \]
\end{proof}

Now we can prove \Cref{lem:rand-delta-lb} using \Cref{lem:delta-prob-compare}.

\begin{proof}[Proof of Lemma~\ref{lem:rand-delta-lb}]
    Fix any $I \in \binom{[n]}{d}$.
    By \Cref{lem:delta-prob-compare}, we have that $\Pr[E_I ] \leq \binom{n}{d}^{-1} + \frac{5}{4} \cdot \Pr[E_I  \wedge (\delta \geq \eps)]$ for $\eps = \frac{1}{10e^3 Ld\log n}$.
    This gives that
    \begin{align*}
         \frac{\Pr[E_I  \wedge (\delta \geq \eps)]}{\Pr[E_I ]}  \geq \frac{4}{5} - \binom{n}{d}^{-1} \cdot \frac{4}{5\Pr[E_I ]}
    \end{align*}
    Moreover, since $\Pr[E_I ] \geq 10 \binom{n}{d}^{-1}$, we have
    \begin{align*}
        \Pr[(\delta \geq \eps) \mid E_I  ] &= \frac{\Pr[E_I  \wedge (\delta \geq \eps)]}{\Pr[E_I  ]} \\
        &\geq \frac{4}{5} - \binom{n}{d}^{-1} \cdot \frac{4}{5\Pr[E_I ]} \geq 0.72,
    \end{align*}
    as desired.
\end{proof}

\subsection{Randomized Lower-Bound for \texorpdfstring{$r$}{r}: Inner Radius of a Ridge Projected onto \texorpdfstring{$(d-1)$}{d-1}-Dimensional Subspace}\label{sub:rand-r-lb}

In the next sections, we demonstrate that for any ridge $R$ of the polytope $P$, wherein $R \cap W$ is a vertex of $P \cap W$, its inner radius—after projection onto the subspace orthogonal to the adjacent edge of $P \cap W$—is at least $\Omega(d^{-4}L^{-2})$ with high probability.
Essentially, this establishes that the parameter $r$, as referred to in \text{\Cref{lem:det}}, is at least $\Omega(d^{-4}L^{-2})$ with good probability.
We remark that \Cref{lem:rand-r-lb} does not have an analogue when $d=2$.
Moreover, it will require substantially more technical effort.
Its proof is similar to Lemma 4.1.1 (Distance bound) in \cite{ST04}, their main technical result.
In an effort to help the ease of understanding the larger structure, some lemmas will be stated
while their proofs will be given in later sections.

\begin{lemma}[Randomized Lower-bound for $r$]\label{lem:rand-r-lb}
    Let $a_1, \ldots, a_n \in \R^d$ be independent $L$-log-Lipschitz random vectors.
    Let $D$ denote the event that $\forall i, j \in [n], \|a_i - a_j\| \leq 3$.
    % such that
    % \begin{align*}
    %     \Pr[\max_{i, j \in [n]}\|a_i - a_j\| \leq 3] \geq 1 - \binom{n}{d}^{-1}.
    % \end{align*}
    Fix any $I \in \binom{[n]}{d}$
    % such that $\Pr[E_I] \geq 10 \binom{n}{d}^{-1}$,
    and any $J \in \binom{I}{d-1}$, we have
    \[
        \Pr[\dist(W \cap \aff(a_i : i\in I),
        \partial \conv(a_j : j\in J )) \leq \frac{1}{19200 d^4L^2} \mid E_I \wedge A_J ] \leq 0.1 + \Pr[\neg D \mid E_I \wedge A_J].
    \]
    % and we let $G$ denote the event that $W \cap \conv(a_1,\dots,a_{d-1}) \neq \emptyset$,
    % then we have
    % \[
    %     \Pr[\dist(W \cap \aff(a_1,\dots,a_d),
    %     \partial \conv(a_1,\dots,a_{d-1})) \leq ? \mid E_I \wedge G] \leq 0.3.
    % \]
\end{lemma}

Then \Cref{lem:rand-r-lb} will  quickly from its lower-dimensional equivalent:
\begin{lemma}[Randomized lower bound for $r$ after change of variables]\label{lem:rand-r-b}
    Let $b_1,\dots,b_d \in \R^{d-1}$ be random vectors with joint probability density proportional to
    \begin{align*}
        \vol_{d-1}(\conv(b_1, \ldots, b_d)) \cdot \prod_{i=1}^d \Bar{\mu}_i(b_i)
    \end{align*}
    where $\bar{\mu}_i$ is $L$-log-Lipschitz for each $i \in [d]$.
    Let $D'$ denote the event that the set $\{b_1,\dots,b_d\}$ has Euclidean diameter of at most $3$.
    Fox any fixed one-dimensional line $\ell \subset \R^{d-1}$, we have that
    \begin{align*}
        \Pr&\left[ \left(\dist \big(\ell, \partial \conv(b_1,\dots,b_{d-1}) \big)
        < \frac{1}{19200 d^4 L^2} \right) \mid \ell \cap \conv(b_1,\dots,b_{d-1}) \neq \emptyset \right] \\
           &\leq 0.1 + \Pr[\neg D'\mid \ell \cap \conv(b_1,\dots,b_{d-1}) \neq \emptyset].
    \end{align*}
\end{lemma}

The proof of \Cref{lem:rand-r-b} will be presented in \Cref{sec:pf-rand-r-b}.

\begin{proof}[Proof of Lemma~\ref{lem:rand-r-lb}]
We may assume without loss of generality that $I = [d]$ and $J = [d-1]$.
Apply the change of variables $\phi$ as in \Cref{def:change-of-variables} to $\{a_i : i\in [d]\}$ and obtain
\begin{align*}
    \phi(\theta, t, b_1, \ldots, b_d) = (a_1, \ldots, a_d).
\end{align*}
where $\theta \in \mathbb{S}^{d-1}, t \in \R, b_1, \ldots, b_d \in \R^{d-1}$.
For any $i\in [n]$, let $\mu_i$ denote the probability density function of
$a_i$.
Writing the conditioning to $(E_{[d]} \wedge A_{[d-1]})$ as part of the pdf, we find that the joint probability density of $t,
\theta, b_1, \ldots, b_d, a_{d+1},\dots,a_n$ is proportional to
\begin{align}
    \vol_{d-1}(\conv(b_1, \ldots, b_d)) \cdot \prod_{i=1}^d \Bar{\mu}_i(t,\theta,b_i) \cdot \prod_{i=d+1}^n \mu_i(a_i) \cdot 1[E_{[d]} \wedge A_{[d-1]}], \label{eq:prob-bi}
\end{align}
where $\vol_{d-1}(\cdot )$ is the volume function of $(d-1)$-dimensional simplex in its spanning hyperplane,
$\Bar{\mu}_i(t,\theta,b_i) = \mu_i(t\theta + R_\theta(b_i))$ is the induced
probability density of $b_i$, which is $L$-log-Lipschitz, and $1[\cdot]$ denotes the indicator
function.
Write $S$ for the event that
$$\dist(W \cap \aff(a_i : i\in I), \partial \conv(a_j : j\in J )) \leq \frac{1}{19200 d^4L^2}.$$
In this language, our goal is to prove that
$\Pr[S] \leq 0.1 + \Pr[\neg D]$.

Let $D'$ denote the event that $\|b_i - b_j\| \leq 3$ for all $i,j \in [d]$.
Each of the events $E_I, A_J, S, D', D$ are functions of
the random variables $\theta, t, b_1,\dots,b_d,a_{d+1},\dots,a_n$.
We then use Fubini's theorem to write
\begin{align*}
    \Pr_{\theta,t,b_1,\dots,b_d,a_{d+1},\dots,a_n}[S]
    &= \E_{\theta,t,a_{d+1},\dots,a_n}[\Pr_{b_1,\dots,b_d}[S]]
\end{align*}
With probability $1$ over the choice of $\theta,t,a_{d+1},\dots,a_n$,
the inner term satisfies all the conditions of \Cref{lem:rand-r-b}.
Specifically, since
the value of $1[E_{[d]}]$ is already fixed, the intersection
$(t\theta + \theta^\perp) \cap W$ is a line.
Let $\ell \subset \R^{d-1}$ be the image of such line under the inverse change of variables $\phi^{-1}$, i.e., $(t\theta + \theta^\perp) \cap W = t\theta + R_{\theta}(\ell)$.
Then the event $A_{[d-1]}$ is equivalent to $\ell \cap \conv(b_1,\dots,b_{d-1}) \neq \emptyset$.
From \Cref{lem:smooth-jacobian},
the joint probability distribution of $b_1,\dots,b_d$ is thus proportional to
    \begin{align*}
        \vol_{d-1}(\conv(b_1, \ldots, b_d)) \cdot \prod_{i=1}^d \Bar{\mu}_i(b_i) \cdot 1[\ell \cap \conv(b_1,\dots,b_{d-1} \neq \emptyset]
    \end{align*}
Applying \Cref{lem:rand-r-b} to the term $\Pr_{b_1,\dots,b_d}[S]$ we find
\begin{align*}
    \E_{\theta,t,a_{d+1},\dots,a_n}\big[\Pr_{b_1,\dots,b_d}[S] \big]
    &\leq \E_{\theta,t,a_{d+1},\dots,a_n}\big[0.1 + \Pr_{b_1,\dots,b_d}[\neg D'] \big]\\
    &= 0.1 + \Pr[\neg D'] \leq 0.1 + \Pr[\neg D],
\end{align*}
using Fubini's theorem for the equality and the fact that $\neg D'$ implies $\neg D$
for the final inequality.
\end{proof}

\subsection{Proof of Lemma~\ref{lem:rand-r-b}: Randomized lower bound for $r$ after change of variables}\label{sec:pf-rand-r-b} 

In this section, we deliver the proof of \Cref{lem:rand-r-b}. We use the following two technical lemmas.
As in from the proof of \Cref{lem:rand-r-b} we define $\lambda \in \R^{d-1}_{\geq 0}$ to be the unique solution to $\sum_{i=1}^{d-1} \lambda_i b_i = \ell \cap \conv(b_1, \ldots, b_{d-1})$ and $\sum_{i=1}^{d-1} \lambda_i = 1$.
First, we describe \Cref{lem:ridge-lambda-lb} which we will use to show that every convex parameter $\lambda_i$ is at least $\Omega(1/d^2L)$ with constant probability.
% Here we use the similar technique from \cite{ST04} and $L$-log-Lipschitzness of each $b_i$.

\begin{lemma}[Lower-bound for Convex Parameters of Vertices on the Ridge]\label{lem:ridge-lambda-lb}
    % \Sophie{the $b_i$ are conditioned on $a_1,\dots,a_d \in \aff(b_1, \ldots, b_{d-1})$
    % and $\conv(a_1,\dots,a_{d-1}) \cap W \neq \emptyset$. also $s$ and $w$ are already defined.
    % assume that $L \geq 1$.}
    Let $b_1,\ldots,b_d \in \R^{d-1}$ be random vectors with joint probability density proportional
    to
    \[
    \vol_{d-1}(\conv(b_1, \ldots, b_d)) \cdot \prod_{i=1}^d \Bar{\mu}_i(b_i)
    \]
    where each $\Bar\mu_i : \R^{d-1} \to \R_+$ is $L$-log-Lipschitz.
    Let one-dimensional line $\ell \subset \R^{d-1}$ and conditional on $\ell \cap \conv(b_i : i\in [d-1]) \neq \emptyset$.
    Let $\lambda \in \R^{d-1}_+$ be the unique solution to $\sum_{i=1}^{d-1} \lambda_i b_i \in \ell \cap \conv(b_i : i\in [d-1])$.
    % Write $\pi = \pi_{\ell^\perp}$ for the orthogonal projection onto $\ell^\perp$.
    Let $D'$ denote the event that $\forall i, j \in [d]$, $\|b_i - b_j\| \leq 3$.
    Then we have
    \[
    \Pr\left[ \forall i \in [d-1] : \lambda_i \geq \frac{1}{120d^2 L} \middle| D' \wedge \ell \cap \conv(b_i : i\in [d-1]) \neq \emptyset \right] \geq 0.95.
    \]
\end{lemma}
We will prove \Cref{lem:ridge-lambda-lb} in \Cref{sec:proof-ridge-lambda-lb}.

Secondly, we present a lemma to lower bound the distance between each vertex $b_j$ (where $j \in [d-1]$) and the $(d-3)$-dimensional hyperplane spanned by the other vertices $\aff(b_j : j \in [d-1], j\neq i)$. Specifically, we show the following lemma:
% todo: state that the following two lemmas are in two independent probablisitic space?

\begin{lemma}\label{lem:facet-distance-lb}
   Let $b_1,\ldots,b_d \in \R^{d-1}$ be random vectors with joint probability density proportional to
    \[
    \vol_{d-1}(\conv(b_1, \ldots, b_d)) \cdot \prod_{i=1}^d \Bar{\mu}_i(b_i)
    \]
    where each $\Bar\mu_i : \R^{d-1} \to [0, 1]$ is $L$-log-Lipschitz.
    Given a one-dimensional line $\ell \subset \R^{d-1}$, let $w \in \mathbb{S}^{d-2}$ be any unit direction of $\ell$.
    For any $i \in [d-1]$ we have
    \begin{align*}
        \Pr\Big[ \dist \left( \pi_{w^\perp}(b_i), \aff(\pi_{w^\perp}(b_j) : j \in [d-1], j \neq i) \right) \geq \frac{1}{160 d^2L} ~\Big|~
        & \ell \cap \conv(b_i : i\in [d-1]) \neq \emptyset \Big] \\
        &\geq 1-\frac{1}{20d}.
    \end{align*}
\end{lemma}

We defer the proof of \Cref{lem:facet-distance-lb} to \Cref{sec:proof-facet-distance-lb}.
In addition to the above two lemmas, we need the following two basic linear algebraic statements.

\begin{fact}\label{lem:proj-dist}
    Let $\ell \subset \R^k$ be an affine line
    and $x \in \R^k$ be a point.
    Suppose $w \in \mathbb{S}^{k-1} \setminus \{0\}$
    points in the direction of $\ell$, i.e.,
    that $\ell + w = \ell$.
    Then $\dist(x, \ell) = \dist(\pi_{w^\perp}(x), \pi_{w^\perp}(\ell))$.
\end{fact}
\begin{proof}
    Let $\{z\} = \ell \cap (x + w^{\perp})$. This intersection is non-empty because
    $w$ points in the direction along $\ell$ and so must intersect any affine subspace orthogonal to $w$.
    The intersection is a singleton because
    if $z, z'$ were two distinct points in this set
    then $z-z' \in w\R \cap w^\perp = \{0\}$.

    Since $x-z \in w^\perp$, we have $\dist(x, z) = \dist(\pi_{w^\perp}(x), \pi_{w^\perp}(\ell))$. Also notice that for any $z' \in \ell, z' \neq z$,
    $$\|x - z'\|^2 = \|\pi_{w}(x - z')\|^2 + \|\pi_{w^\perp}(x - z')\|^2 \geq \|\pi_{w^\perp}(x - z')\|^2 = \|x - z\|^2.$$
    Therefore, $\dist(x, \ell) = \|x - z\| = \dist(\pi_{w^\perp}(x), \pi_{w^\perp}(\ell))$.
%
    % Let $z \in \ell$ be such that
    % $\|x-z\| = length([x,z]) = \dist(x, \ell)$. \sophie{why does it exist?}
    % now we'd project, noting that
    % $\pi(z) \in \pi(\ell)$
    % and $length([x,z]) = length([\pi(x), \pi(z)])$.
    % \sophie{this doesn't make anything clearer so far lmao}
\end{proof}

\begin{fact}\label{lem:conv-comb-dist}
    Let $H \subset \R^k$ be a hyperplane, and let $p_1,\dots,p_k \in H$
    and $p_{k+1} \in \R^k$,
    and assume $\lambda_{k+1} \geq 0$.

    Then
    \[
    \dist(\sum_{i=1}^{k+1} \lambda_i p_i, H)
    =
    \lambda_{k+1} \dist(p_{k+1}, H).
    \]
\end{fact}
\begin{proof}
    Let $y \in \R^k, t \in \R$ be such that
    $H = \{x \in \R^k : y^\T x = t\}$
    and $\norm{y} = 1$.

    Now we have
    \begin{align*}
        \dist(\sum_{i=1}^{k+1} \lambda_i p_i, H)
        &= |t - y^\T (\sum_{i=1}^{k+1} \lambda_i p_i)| \\
        &= |\sum_{i=1}^{k+1} \lambda_i ( t -  y^\T p_i)| \\
        &= |\lambda_{k+1} ( t -  y^\T p_{k+1})| \\
        &= \lambda_{k+1} \dist( p_{k+1}, H),
    \end{align*}
    using that $y^\T p_i = t$ for all $i = 1,\dots,k$.
\end{proof}

Assuming \Cref{lem:ridge-lambda-lb} and \Cref{lem:facet-distance-lb}, we can prove \Cref{lem:rand-r-b}.

\begin{proof}[Proof of Lemma~\ref{lem:rand-r-b}]
    We can bound the distance from $\ell$ to $\partial \conv(b_1,\dots,b_{d-1})$.
    Note that generically $w$ is not parallel to any direction in $\aff(b_1,\dots,b_{d-1})$
    which makes it so that $\pi_{w^\perp}(\conv(b_1,\dots,b_{d-1}))$
    is a simplex of the same dimension as $\conv(b_1,\dots,b_{d-1})$.
    \begin{align*}
        & \dist \big(\ell, \partial \conv(b_1,\dots,b_{d-1}) \big) \\
        = & \dist \big(\pi_{w^\perp}(\ell), \pi_{w^\perp}(\partial \conv(b_1,\dots,b_{d-1})) \big) \tag{By \Cref{lem:proj-dist}} \\
        = & \dist \big(\pi_{w^\perp}(\ell), \partial \conv(\pi_{w^\perp}(b_1),\dots,\pi_{w^\perp}(b_{d-1})) \big) \\
        = & \min_{i \in [d-1]} \dist \big(\pi_{w^\perp}(\ell), \conv(\pi_{w^\perp}(b_1),\dots,\pi_{w^\perp}(b_{i-1}),\pi_{w^\perp}(b_{i+1}),\dots,\pi_{w^\perp}(b_{d-1})) \big) \\
        \geq & \min_{i \in [d-1]} \dist \big(\pi_{w^\perp}(\ell), \aff(\pi_{w^\perp}(b_1),\dots,\pi_{w^\perp}(b_{i-1}),\pi_{w^\perp}(b_{i+1}),\dots,\pi_{w^\perp}(b_{d-1})) \big) \\
        = & \min_{i \in [d-1]} \lambda_i \cdot \dist(\pi_{w^\perp}(b_i), \aff(\pi_{w^\perp}(b_j) : j \in [d-1], j \neq i)) \tag{By \Cref{lem:conv-comb-dist}}\\
        \geq & \min_{i \in [d-1]} \lambda_i \cdot  \min_{k \in [d-1]} \dist(\pi_{w^\perp}(b_k), \aff(\pi_{w^\perp}(b_j) : j \in [d-1], j \neq k))
    \end{align*}
    Where in the second step, we use the fact that $\pi_{w^{\perp}}$ is an affine isomorphism restricting to $b_{1}, \dots, b_{d-1}$, thus taking the boundary of $\conv(\pi_{w^\perp}(b_1),\dots,\pi_{w^\perp}(b_{d-1}))$ commutes with the projection $\pi_w^{\perp}$.
    In the fifth step,  $\lambda \in \R^{d-1}_{\geq 0}$ is the unique solution to $\sum_{i=1}^{d-1} \lambda_i b_i = \ell \cap \conv(b_1, \ldots, b_{d-1})$ and $\sum_{i=1}^{d-1} \lambda_i = 1$.
    Additionally, assume $\ell = w\R$ for a non-zero $w \in \mathbb{S}^{d-2}$.
    Abbreviate, for $k \in [d-1]$,
    $$r_k = \dist(\pi_{w^\perp}(b_k), \aff(\pi_{w^\perp}(b_j) : j \in [d-1], j \neq k)).$$
    Let $T$ denote the event that $\ell \cap \conv(b_1,\dots,b_{d-1}) \neq \emptyset$.
    We now find using a union bound, for any $\alpha,\beta > 0$,
    \begin{align*}
        & \Pr[\dist(\ell,\partial\conv(b_1,\dots,b_{d-1})) < \alpha\beta \mid T] \\
        \leq & \Pr[\min_{i \in [d-1]} \lambda_i < \alpha \mid T] + \Pr[\min_{k \in [d-1]} r_k < \beta \mid T] \\
        \leq & \Pr[\min_{i \in [d-1]} \lambda_i < \alpha \mid D' \wedge T] + \Pr[\min_{k \in [d-1]} r_k < \beta \mid T] + \Pr[\neg D' \mid T]
    \end{align*}
    By \Cref{lem:ridge-lambda-lb} we know that
    $\Pr[\min_{i \in [d-1]} \lambda_i < \alpha \mid D' \wedge T] \leq 0.05$ for $\alpha = \frac{1}{120d^2L}$.
    By \Cref{lem:facet-distance-lb}, we know that
        $\Pr[\min_{k \in [d-1]} r_k < \beta \mid T]\leq
        \sum_{k\in[d-1]} \Pr[ r_k < \beta \mid T]
        \leq 0.05$ for $\beta = \frac{1}{160d^2L}$.
    This proves the lemma.
\end{proof}

\subsection{Proof of Lemma~\ref{lem:ridge-lambda-lb}}\label{sec:proof-ridge-lambda-lb}
Now we show that, with good probability, the convex multipliers $\lambda$ are not too small.

\begin{proof}[Proof of Lemma~\ref{lem:ridge-lambda-lb}]
    By using a union bound, it suffices to prove for each $i \in [d-1]$ that
    \[
    \Pr[\lambda_i < \frac{1}{120d^2 L} \mid D' \wedge \ell \cap \conv(b_j : j\in [d-1]) \neq \emptyset ] \leq \frac{1}{20(d-1)}.
    \]
    Fix any $i \in [d - 1]$, without loss of generality $i=1$.
    Recall that $w \in \mathbb{S}^{d-1}$ is such that $\ell = \ell+w$, and hence that $\pi_{w^\perp}(\ell)$ is a singleton point.
    For ease of exposition we assume that the plane $w^\perp$ is coordinatized such that $\pi_{w^\perp}(\ell) = 0$
    and hence $\ell = w\R$.
    Thus $\lambda$ is defined to satisfy $\sum_{j=1}^{d-1} \lambda_j \pi_{w^\perp}(b_j) = 0$.

    Using the principle of deferred decision, we fix the values of $b_1 - b_j$, $j \in [d]$.
    This determines the shape of the simplex $\conv(b_j : j \in [d])$, including its volume.
    The remainder of this proof will use the randomness in the position of the simplex
    in the subspace orthogonal to the line, which we represent using $\pi_{w^\perp}(b_1)$.
    For the remainder of this proof, we can consider all $b_j, j \in [d]$
    to be functions of $b_1$.
    The position $\pi_{w^\perp}(b_1)$ has
    probability density $\mu'(\pi_{w^\perp}(b_1)) \propto \prod_{j=1}^d \mu_j(b_j)$, which is
    $dL$-log-Lipschitz in $\pi_{w^\perp}(b_1)$ with respect to the $(d-2)$-dimensional Lebesgue
    measure on $w^\perp$.

    Define $M = \conv(\pi_{w^\perp}(b_1 - b_j) : j \in [d-1]) \subset w^\perp$ and note that,
    due to our fixing the values of $b_1-b_j$ in the previous paragraph,
    the shape $M$ is fixed and we can see that
    that $\pi_{w^\perp}(b_1) \in M$ if and only if $\lambda \geq 0$.
    It remains to show that
    \begin{align}
        \Pr[\lambda_1 < \frac{1}{ 120d^2 L} \mid D' \wedge \pi_{w^\perp}(b_1) \in M] < \frac{1}{20d}. \label{eq:lambda-lb-goal}
    \end{align}

    For any $j \in [d-1]$, let $l_j : M \to [0, 1]$ be the function sending any point
    to its $j$'th convex coefficient, i.e., the functions satisfy $\sum_{j=1}^{d-1} l_j(x) = 1$
    and $\sum_{j=1}^{d-1} l_j(x) \cdot \pi_{w^\perp}(b_1 - b_j)  = x$ for every $x \in M$.
    For any $1 \geq \alpha \geq 0$, observe that $l_1$
    takes values in the interval $[\alpha, 1]$
    on the set $(1-\alpha) M$. Hence we get
    \begin{align*}
        \Pr[\lambda_1 \geq \alpha \mid \pi_{w^\perp}(b_1) \in M]
        &= \frac{
            \int_{M} \mu'(x) \mathbf{1}[l_1(x) \geq \alpha] \mathrm d x
            }{
                \int_{M} \mu'(x) \mathrm d x
            } \\
        &\geq \frac{
            \int_{(1-\alpha) M} \mu'(x) \mathrm d x
            }{
                \int_{M} \mu'(x) \mathrm d x
            }
            \tag{$\forall x \in (1-\alpha)M$, $l_1(x) \geq \alpha$}\\
        &= \frac{
            (1-\alpha)^{d-2} \int_{M} \mu'((1-\alpha)x) \mathrm d x
            }{
                \int_{M} \mu'(x) \mathrm d x
            }
            \\
        &\geq
            (1-\alpha)^{d-2} \max_{x \in M} e^{-dL\|\alpha x\|},
    \end{align*}
    where in the last inequality, we use $dL$-log-Lipschitzness of $\mu'$ to see that for any $s \in M$
    we have $\frac{\mu'((1-\alpha)s)}{\mu'(s)} \leq \max_{x \in M}e^{-dL \|\alpha x\|}$.
    
    By definition of $D'$, we know that $M$ has Euclidean diameter at most $3$.
    Thus we can bound $\|\alpha x\| \leq 3\alpha$ for any $x\in M$.
    Now take $\alpha = \frac{1}{120d^2L}$, we find
    \begin{align*}
        \Pr[(\lambda_i < \frac{1}{120d^2 L}) \mid D' \wedge \pi_{w^\perp}(b_i) \in M] &\leq 1 -
        \Pr[\lambda_i \geq \frac{1}{120d^2 L} \mid D' \wedge \pi_{w^\perp}(b_i) \in M] \\
        &\leq 1 - (1-\frac{1}{120 d^2 L})^{d-2} e^{-1/40d}
        \leq \frac{1}{20(d-1)},
    \end{align*}
    where the last inequality comes from $d \geq 3$ and $L \geq 1$. Thus \eqref{eq:lambda-lb-goal} holds as desired.
\end{proof}

\subsection{Proof of Lemma~\ref{lem:facet-distance-lb}}\label{sec:pf-facet-dist}

To show \Cref{lem:facet-distance-lb} on the width of the facet, we need the following upper bound about the mass around zero for a random variable whose density is formed by a log-Lipschitz function multiplied by a convex function.
Its function is to deal with the volume term that we receive from the Jacobian in \Cref{lem:smooth-jacobian}.

\begin{lemma}\label{lem:integrals}
    Assume that $h : \R \to \R_{\geq 0}$ is a $K$-log-Lipschitz function
    and $g : \R \to \R_{\geq 0}$ is a convex function such that
    $\int_{-\infty}^\infty g(x) \cdot h(x) \mathrm d x = 1$.
    Suppose that $X \in \R$ is distributed with probability density $g(X)\cdot h(X)$.
    For any $\eps > 0$ we have
    $\Pr[X \in [-\eps, \eps]] \leq 8 \eps K$.
\end{lemma}
\begin{proof}
    We can assume that $\eps < 1/(8K)$, for otherwise the bound is trivial.
    First, we use the rudimentary upper bound
    \[
        \Pr\big[X \in [-\eps,\eps] \big]  \leq \Pr\big[X \in [-\eps, \eps] \mid X \in [-1/K,1/K] \big]
            = \frac
                {\int_{-\eps}^{\eps} g(x) \cdot h(x) \mathrm d x}
                {\int_{-1/K}^{1/K} g(x) \cdot h(x) \mathrm d x}.
    \]
    Log-Lipschitzness implies that for any $\gamma > 0$ we have
    \[
        e^{-\gamma K} h(0) \int_{-\gamma}^\gamma g(x) \mathrm d x
        \leq
        \int_{-\gamma}^\gamma g(x) \cdot h(x) \mathrm d x
        \leq
        e^{\gamma K} h(0) \int_{-\gamma}^\gamma g(x) \mathrm d x,
    \]
    and hence we get
    \[
        \Pr[X \in [-\eps, \eps]]
        \leq e^{(1/K + \eps) K} \frac
        {\int_{-\eps}^{\eps} g(x) \mathrm d x}
        {\int_{-1/K}^{1/K} g(x) \mathrm d x}
        \leq e^{1 + \eps K} \frac
        {2\eps \cdot \max_{x \in [-\eps, \eps]} g(x)}
        {\int_{-1/K}^{1/K} g(x) \mathrm d x}
    \]
    Since $g(x)$ is convex, at least one of
    \begin{align*}
        \max_{x \in [-\eps, \eps]} g(x) \leq \min_{x \in [-1/K, -\eps]} g(x)
        \quad \text{or} \quad
        \max_{x \in [-\eps, \eps]} g(x) \leq \min_{x \in [\eps, 1/K]} g(x)
    \end{align*}
    holds. Without loss of generality, assume the second case holds. Then we bound
    \[
        \frac
            {\max_{x \in [-\eps, \eps]} g(x)}
            {\int_{-1/K}^{1/K} g(x) \mathrm d x}
        \leq
        \frac
            {\max_{x \in [-\eps, \eps]} g(x)}
            {\int_{\eps}^{1/K} g(x) \mathrm d x}
        \leq \frac{1}{1/K - \eps}.
    \]
    To summarize, we find
        $\Pr[ X \in [-\eps, \eps]]
        \leq
        e^{1+\eps K} \cdot \frac{2\eps}{1/K - \eps}$.
    Since $\eps < 1/(8K)$ this implies
    \[
        \Pr[X \in [-\eps,\eps]] \leq 2 e^{9/8} \cdot \frac 8 7 \cdot \eps K \leq 8\eps K .\qedhere
    \]
\end{proof}

Now we show \Cref{lem:facet-distance-lb}.
\begin{proof}[Proof of Lemma~\ref{lem:facet-distance-lb}]\label{sec:proof-facet-distance-lb}
In the following arguments, we condition on $\ell \cap \conv(b_i : i\in [d-1]) \neq \emptyset$.
Without loss of generality, set $i = d-1$ and assume that $\ell$
is a linear subspace, i.e., $\pi_{w^\perp}(\ell)=0$.

We start with a coordinate transformation.
Let $\phi \in w^\perp \cap \mathbb S ^{d-2}$ denote the unit vector
satisfying $\phi^\T b_1 = \phi^\T b_j > 0$ for all $j = 1,\dots,d-1$.
Note that $\phi$ is uniquely defined almost surely: $w^\perp$ is a $(d-2)$-dimensional linear space
and we impose $(d-3)$ linear constraints $\{\phi^\top b_1 = \phi^\top b_j, \forall j \in [d-2] \}$.
Almost surely, these give a one-dimensional linear subspace which,
after adding the unit norm and $b_1^\T \phi > 0$ constraint, leaves a unique choice of $\phi$.

Now define $h \in \R$ by $h = \phi^\T b_1$ and define $\alpha \in \R$
by $\alpha h = - \phi^\T b_{d-1}$.
Since $0 \in \conv(\pi_{w^\perp}(b_i) : i \in [d-1])$ but $\phi^\T b_i > 0$ for all $i \in [d-2]$,
we must have $\alpha \geq 0$ for otherwise $\phi$ would separate $\conv(\pi_{w^\perp}(b_i) : i \in [d-1])$ from $0$.
Again from almost-sure non-degeneracy we get $\alpha > 0$ and $h \neq 0$.
We define the following coordinate transformation:
\begin{align*}
    b_j &= h\phi + c_j, \quad \forall j \in [d-2] \\
    b_{d-1} &= - \alpha h \phi + c_{d-1}
\end{align*}
where for each $j \in [d-1]$, $c_j \in \phi^\perp \cap \aff(b_1, \ldots, b_{d-1})$ has $(d-3)$ degrees of freedom.
From here on out, we consider the vertices $(b_1, \ldots,  b_{d-1})$
to be a function of $(h,\alpha, \phi, c_1,\dots,c_{d-1})$.
Again by \Cref{lem:smooth-jacobian},
the induced joint probability density on $(h,\alpha, \phi, c_1,\dots,c_{d-1}, b_d)$,
is proportional to
\begin{align*}
\vol_{d-1}(\conv(b_1,\dots,b_d)) \cdot
\vol_{d-3}(\conv(c_1, \ldots, c_{d-2})) \cdot
\prod_{j=1}^d \bar{\mu}_j(b_j)
\end{align*}
Using the principle of deferred decision,
fix the exact values of $(\alpha, \phi, c_1,\dots,c_{d-1}, b_d)$.
When this is the case we find that
\begin{align*}
& \vol_{d-1}(\conv(b_1,\dots,b_d)) \cdot
\vol_{d-3}(\conv(c_1, \ldots, c_{d-2})) \cdot
\prod_{j=1}^d \bar{\mu}_j(b_j) \\
\propto ~ & \vol_{d-1}(\conv(b_1,\dots,b_d)) \cdot
\vol_{d-3}(\conv(\pi_{w^\perp}(c_1),\dots,\pi_{w^\perp}(c_{d-2}))) \cdot
\prod_{j=1}^d \bar{\mu}_j(b_j).
\end{align*}
To see why $\vol_{d-3}(\conv(c_{1}, \dots, c_{d-2}))$ and $\vol_{d-3}(\conv(\pi_{w^\perp}(c_{1}),\dots, \pi_{w^\perp}(c_{d-2}))$ are proportional, consider the following.
Note that $\pi_{w^\perp}$ is a fixed projection.
Generically we know that $w \notin \operatorname{span}(b_2-b_1,\dots,b_{d-1}-b_1)$,
this is a non-degeneracy condition that holds true with probability $1$.
This implies that $\pi_{w^\perp}$ is a bijection between the affine hyperplane
$\aff(b_1,\dots,b_{d-1})$ and its image $\pi_{w^\perp}(\aff(b_1,\dots,b_{d-1}))$,
and also a bijection between the subsets
\[
\phi^\perp \cap \aff(b_1,\dots,b_{d-1}) \text{~and~} \pi_{w^\perp}(\phi^\perp \cap \aff(b_1,\dots,b_{d-1})).
\]
This is the space where $c_1,\dots,c_{d-2}$ live. It follows that
$\conv(c_1, \ldots, c_{d-2})$ has the same dimension as
its projection 
$\conv(\pi_{w^\perp}(c_1),\dots,\pi_{w^\perp}(c_{d-2}))$.
The ratio between their volumes depends only on $w$ and
$\phi^\perp \cap \operatorname{span}(b_2-b_1,\dots,b_{d-1}-b_1)$.
Note that $h = \phi^\perp b_1$ and so, after fixing $(\alpha, \phi, c_1,\dots,c_{d-1}, b_d)$
the ratio between the volumes is constant.

Note that the event $0 \in \conv(\pi_{w^\perp}(b_i) : i \in [d-1])$ depends only on these variables
and not on $h$, and the same is true for
$\vol_{d-3}(\conv(\pi_{w^\perp}(c_1),\dots,\pi_{w^\perp}(c_{d-2})))$.
We are looking only at the randomness in $h$,
and so we can ignore any constant factors in the probability density function.
The induced probability density on $h$ is now proportional to
\[
    \vol_{d-1}(\conv(b_1,\dots,b_d)) \cdot
    \prod_{j=1}^{d-1} \tilde{\mu}_j(h),
\]
where $\tilde{\mu}_j(h) := \bar{\mu}_j(h\phi + c_j), j \in [d-2]$ and
$\tilde{\mu}_{d-1}(h) := \bar{\mu}_{d-1}(\alpha h \phi + c_{d-1})$.
Since each $\bar{\mu}_j$ is $L$-log-Lipschitz,
it follows that the product $\prod_{j=1}^{d-1} \tilde{\mu}_j(h)$
is $(d-2 + \alpha)L \leq d(1 + \alpha)L$-log-Lipschitz
in $h$.

Next, consider the volume term.
We can write
    $\vol_{d-1}(\conv(b_1,\dots,b_d))$ as a constant depending on $d$ times
    the absolute value of the determinant of the following $(d-1)\times(d-1)$ matrix
    % (note that $b_1, \cdots, b_d \in \R^{d-1}$)
    \[
        \begin{bmatrix}
            (b_1 - b_d)^\top \\
            \vdots \\
            (b_{d-1} - b_d)^\top
        \end{bmatrix}
        =
        \begin{bmatrix}
            (h\phi + c_1 - b_d)^\top \\
            \vdots \\
            (h\phi + c_{d-2} - b_d)^\top \\
            (- \alpha h\phi +  c_{d-1} - b_d)^\top 
        \end{bmatrix}
        =
        \begin{bmatrix}
            (c_1 - b_d)^\top \\
            \vdots \\
            (c_{d-2} - b_d)^\top \\
            (c_{d-1} - b_d)^\top
        \end{bmatrix} +
        h\cdot \begin{bmatrix} 1 \\ \vdots \\ 1 \\ -\alpha \end{bmatrix} \phi^\top,
    \]
    % Let $\theta \in \sfe$ denotes a normal vector to $\aff(b_1, \ldots, b_{d-1})$.
    % Since this matrix is an affine function $A + tB$
    % of $t$, the matrix determinant lemma states that
    Define
    \begin{align*}
        B := \begin{bmatrix}
            (c_1 - b_d)^\top \\
            \vdots \\
            (c_{d-2} - b_d)^\top \\
            (c_{d-1} - b_d)^\top
        \end{bmatrix} ,
        v := \begin{bmatrix} 1 \\ \vdots \\ 1 \\ -\alpha
        \end{bmatrix}.
    \end{align*}
    Recall that both $B$ and $v$ are fixed and we are only interested in the distribution of $h$.
    Then by the matrix determinant lemma, we can write the volume as the absolute value of an affine function of $h$ (which is a convex function):
    \begin{align*}
        k(h):=\vol_{d-1}(\conv(b_1, \ldots, b_d)) &\propto \left|\det(B + h v\phi^\top) \right| \\
        &= \left| \det(B)(1 + j \phi^\top B^{-1} v) \right|
    \end{align*}
    Hence, we have found a convex function $k : \R \to \R_{\geq 0}$
    and a $d(1+\alpha)L$-log-Lipschitz function $\nu : \R \to \R_{\geq 0}$
    such that $h$ has probability density proportional to
    $k(h) \cdot \nu(h)$.

    To finalize the argument, we write
    $\dist(\pi_{w^\perp}(b_i), \aff(\pi_{w^\perp}(b_j) : j \in [d-1], j \neq i)) = |(1+\alpha)h|$.
    It follows that the signed distance $(1+\alpha)h$ has a probability density function
    proportional to the product of a $dL$-log-Lipschitz function
    and a convex function. The result follows from \Cref{lem:integrals}
    by plugging in the signed distance $(1+\alpha)h$, $K = dL$, and $\eps = \frac{1}{160 d^2 L}$.
\end{proof}

\subsection{Combining Together and Proof of Lemma~\ref{lem:normalvector}}\label{sub:combine}

In this section, we combine the deterministic argument in \Cref{lem:det} with the probabilistic arguments in \Cref{lem:rand-delta-lb} and \Cref{lem:rand-r-lb}. We can finally show the main technical lemma (\Cref{lem:normalvector}).

\begin{proof}[Proof of Lemma~\ref{lem:normalvector}]
    Without loss of generality, let $I = [d]$ and write $E = E_I$.
    Suppose $p' = A_J = \conv(a_j : j \in J) \cap W$ is the next vertex after the edge $F_I \cap W$.
    Here $J \in \binom{[n]}{d-1}$ and $R_J$ is the $(d-2)$-dimensional ridge.
    With probability $1$, the polytope $\conv(a_1,\dots,a_n)$
    is non-degenerate and $W \cap R'$ is a single point for any ridge $R'$ of $\conv(a_1,\dots,a_n)$ that intersects with $W$.
    We will show that conditional on $E$, each of the following conditions in the deterministic argument (\Cref{lem:det}) is satisfied with good probability:
    \begin{enumerate}
        \item (Bounded diameter) $\forall i,j\in [n]$, $\|a_i - a_j\| \leq 3$;
        \item (Lower bound of $\delta$) $\min_{k \in [n] \backslash I}\dist(\aff(F_I), a_k) \geq \Omega(\frac{1}{Ld\log n})$;
        \item (Lower bound of $r$) $\forall J \in \binom{I}{d-1}$ for which the ridge
        $R_J = \conv(a_j : j\in J)$ has nonempty intersection with $W$,
        we have $\dist(F_I \cap W, \partial R_J) \geq  \Omega(\frac{1}{d^4L^2})$.
    \end{enumerate}
    Note for the last point that \Cref{lem:det} only requires this for the set $J$
    which indexes the second vertex of $F_I \cap W$ in clockwise direction,
    but we prove it for both of the sets $J$ for which $R_J \cap W \neq \emptyset$.

    First, we write $D$ as the event that $\forall i, j \in [n]$ for which $\|a_i - a_j\| \leq 3$.
    From \Cref{lem:properties-laplace-gaussian}, for any $\sigma \leq \frac{1}{8\sqrt{d \log n}}$, with probability at least $1 - \binom{n}{d}^{-1}$, we have $\max_{i \in [n]}\|\hat{a}_i\| \leq 1 + 4\sigma \sqrt{d \log n} \leq \frac{3}{2}$, i.e., $\Pr[D] \geq 1 - \binom{n}{d}^{-1}$.
    Using the assumption that $\Pr[E_I] \geq 10\binom{n}{d}^{-1}$, we have
    $$\Pr[\neg D \mid E] = \Pr[\neg D \wedge E]/\Pr[E] \leq \frac{\Pr[\neg D]}{\Pr[E]} \leq 0.1, $$
    This immediately implies $\Pr[D \mid E] \geq 0.9$.

    Next, we consider $\delta := \dist(\aff(a_1,\dots,a_d), \{a_{d+1},\dots,a_n\})$.
    Using \Cref{lem:rand-delta-lb}, we have $\Pr[\delta \geq \frac{1}{10e^3 Ld\log n} \mid E] \geq 0.72$.

    Finally, we consider $r := \max_{J} \dist(\aff(a_1, \ldots, a_d) \cap W, \partial R_J)$ subject to all $J \in \binom{I}{d-1}$ such that $A_J$ happens (in other words, $R_J = \conv(a_j : j \in J)$ is a ridge of $F_I$ such that $R_J \cap W \neq \emptyset$). By a union bound,
    \begin{align}\label{eq:rand-1}
        & \Pr\left[ \exists J \in \binom{I}{d-1}, A_J \wedge \dist(\aff(F \cap W), \partial R_J \geq \frac{1}{19200d^4L^2} \mid E \right] \nonumber \\
        \geq & 1 - \sum_{J \in \binom{I}{d-1}} \Pr[ A_J \wedge \dist(\aff(F \cap W), \partial R_{J} < \frac{1}{19200d^4L^2} \mid E ] \nonumber \\
        = & 1 - \sum_{J \in \binom{I}{d-1}} \Pr[ \dist(\aff(F \cap W), \partial R_{J} < \frac{1}{19200d^4L^2} \mid E \wedge A_J ] \Pr[A_J \mid E ] .
    \end{align}

    From \Cref{lem:rand-r-lb}, for each $J \in \binom{I}{d-1}$, we know that
    \[
        \Pr[\dist(\aff(F \cap W), \partial R_J < \frac{1}{19200d^4L^2} \mid E \wedge A_J ]
        \leq 0.1 + \Pr[\neg D \mid E \wedge A_J],
    \]
    Notice that when $E$ happens, there are exactly two distinct ridges $R_{J}, R_{J'}$ that has nonempty intersection with $W$ (or $A_J$ happens), thus $\sum_{J \in \binom{I}{d-1}} \Pr[A_J \mid E] = 2$.
    Therefore
    \begin{align*}
        &  \sum_{J \in \binom{I}{d-1}} \Pr[ \dist(\aff(F \cap W), \partial R_{J} < \frac{1}{19200d^4L^2} \mid E \wedge A_J ] \Pr[A_J \mid E ] \\
        \leq &  \sum_{J \in \binom{I}{d-1}} (0.1 +  \Pr[\neg D \mid E \wedge A_J])\Pr[A_J \mid E ] \\
        \leq &  0.1 \cdot 2 - 2\cdot\Pr[\neg D \mid E] ,
    \end{align*}
    and \eqref{eq:rand-1} becomes
    \begin{align*}
        & \Pr\left[ \exists J \in \binom{I}{d-1}, A_J \wedge \dist(\aff(F \cap W), \partial R_J \geq \frac{1}{19200d^4L^2} \mid E  \right]  \\
        \geq & 1 - 0.1 \cdot 2 - 2\cdot\Pr[\neg D \mid E] \geq 0.6 .
    \end{align*}
    Therefore, by a union bound, the three conditions hold with probability at least $1 - (1 - 0.9) - (1 - 0.72) - (1 - 0.6) \geq 0.1$, and the lemma directly follows from \Cref{lem:det}.
\end{proof}

\section{Smoothed Complexity Lower Bound}\label{sec:lb}

In this section, we present the lower bound of the smoothed complexity of the shadow vertex simplex method by studying the intersection of the smoothed polar polytope \mbox{$\conv(a_1, \ldots, a_n) \subset \R^d$} (where each $a_i$ is a Gaussian perturbation from a $\bar a_i$), and the two-dimensional shadow plane $W \subset \R^d$. Our main result is as follows:

\begin{theorem}\label{thm:lb-main}
    For any $d > 5, n = 4d - 13$, there exists a two-dimensional linear subspace $W \subset \R^{d}$ and vectors $\bar a_1,\ldots,\bar a_n \in \R^d$, $\max_{i \in n} \|\bar{a}_i\|_1 \leq 1$ such that the following holds. 
    Let $a_1, \ldots, a_n$ be independent Gaussian random variables where each $a_i \sim \mathcal{N}_d(\bar{a}_i, \sigma^2 I), \sigma \leq \frac{1}{7200d\sqrt{\log n}}$.
    Then with probability at least $1 - \binom{n}{d}^{-1}$ we have
    \begin{align*} 
        \edges(\conv(a_1,\ldots,a_n) \cap W) \geq \Omega \left(\min \Big(\frac{1}{\sqrt{d\sigma\sqrt{\log n}}},2^d \Big) \right).
    \end{align*}
\end{theorem}

\Cref{thm:lb-main} is a direct consequence of the next theorem, which is a lower bound for adversarial perturbations of bounded magnitude:
\begin{theorem}\label{thm:lb-generic}
    For any $d > 5, n = 4d-13$, there exists a two-dimensional linear subspace $W \subset \R^{d}$ and vectors $\bar a_1,\ldots,\bar a_n \in \R^d$, $\max_{i \in n} \|\bar{a}_i\|_1 \leq 1$ such that the following holds. 
    For any $\eps < \frac{1}{180}$, if $a_1, \ldots, a_n \in \R^d$ satisfy $\|a_i - \bar a_i\|_1 \leq \eps$ for all $i \in [n]$
    then we have
    \begin{align*} 
        \edges(\conv(a_1,\ldots,a_n) \cap W) \geq \Omega \left(\min \Big(\frac{1}{\sqrt{\eps}},2^d \Big) \right).
    \end{align*}
\end{theorem}

The rest of this section is organized as follows.
In \Cref{sec:lb-construction}, we construct an auxiliary polytope $P \subset \R^d$ and a two-dimensional shadow plane $W$.
In \Cref{sec:lb-radius-proj}, we show that the projection $\pi_W(P)$ approximates the unit disk $\ball^2$. 
In \Cref{sec:lb-radius-P}, we analyse the largest $\ell_\infty$-ball contained in $P$ and the smallest $\ell_\infty$-ball containing $P$.
\Cref{sec:lb-dual} investigates the polar polytope $Q = (P-x)^\circ$ of a shift of $P$, such that we may choose $\bar{a}_1,\dots,\bar{a}_n$
to satisfy $Q = \conv(\bar a_1,\dots,\bar a_n)$.
The largest contained $\ell_1$-ball in $Q$ is derived from the smallest $\ell_\infty$-ball containing $P$,
and the smallest $\ell_1$-ball containing $Q$ is derived from the largest $\ell_\infty$-ball contained in $P$.
Similarly, the section $Q \cap W$ approximates a circular disk because the projection $\pi_W(P)$ approximates a circular disk.
Finally, \Cref{sec:lb-perturbation} shows that the bounded ratio between the inner and outer radii implies that any
sufficiently small perturbation $\tilde Q$ still has $\tilde Q \cap W$ approximate the unit disk $\ball^2$ well
and uses this to prove \Cref{thm:lb-main}.

\subsection{Construction of the Auxiliary Polytope}\label{sec:lb-construction}

In this subsection, we construct the auxiliary polytope $P$ and the two-dimensional plane $W$.
For $k \in \N$, we construct a $(k + 5)$-dimensional polytope.
We will use the following vectors in the definition:
\begin{itemize}
    \item Define $e_1 = \begin{bmatrix}1\\ 0 \end{bmatrix} \in \R^2$ and $e_2 = \begin{bmatrix} 0\\ 1 \end{bmatrix}$.
    \item For every $i \in \{0, 1, \ldots, k\}$, define the pair of orthogonal unit vectors $w_i  = \begin{bmatrix}\cos(\pi/2^{i+2}) \\ \sin(\pi/2^{i+2}) \end{bmatrix} \in \R^2$ and $v_i =  \begin{bmatrix}\sin(\pi/2^{i+2}) \\ -\cos(\pi/2^{i+2})\end{bmatrix} \in \R^2$.
\end{itemize}

With these definitions in mind, construct an auxiliary $P' \subset \R^{3k+5}$ as the set of
points $(x,y, p_0,\dots,p_k, t,s)$ where $x,y \in \R, p_0,p_1,\ldots,p_k \in \R^2, t \in \R^k$ and $s \in \R$
satisfy the following system of linear inequalities:

\pagebreak
\begin{align}
    % 4 inequalities
    e_1^\top p_0 &\geq |x|, e_2^\top p_0 \geq |y| \label{cons:absvalue}  \\
    % not counting this one since we substitute the variables out
    w_i^\top p_i &= w_i^\top p_{i-1} ,~ \forall i \in [k] \label{cons:w}\\
    % 2k inequalities
    t_i  + is = v_i^\top p_i &\geq |v_i^\top p_{i-1}| ,~ \forall  i \in [k] \label{cons:v}\\
    % 1 inequality
    e_1^\top p_k &\leq 1 \label{cons:deathbarrier}\\
    % 2k inequalities
    \mathbf{0}_k \leq t &\leq \mathbf{1}_k \label{cons:t}\\
    % 2 inequalities
    0 \leq s &\leq 1 \label{cons:s}.
\end{align}

We remark that $p_0, t, s$ uniquely define the values of $p_1,p_2,\ldots,p_k$ via \eqref{cons:w} and \eqref{cons:v}. As such, define the polytope $P \subset \R^{k + 5}$ as the projection of $P'$ onto the subspace spanned by the variables $(x,y,p_0,t,s)$:
\begin{equation}\label{eq:extension}
P = \{(x,y,p_0,t,s) : \exists ~ p_1,\ldots,p_k \mbox{, s.t. } (x,y,p_0,\ldots,p_k,t,s) \in P'\}.
\end{equation}
We choose the plane $W$ to be the one that is spanned by the unit vectors in the $x$ and $y$ directions.

An illustration of the vertices of the projected polytope $\pi_W(P)$ can be found in \Cref{fig:unperturbed} for $k=4$.
Note that the figure appears to depict a regular polygon with $2^{k+1}$ vertices.
Our construction of $P'$ is similar to those of \cite{BN01,glineur}.
The primary difference lies in the addition of the variables
$t$ and $s$ in \eqref{cons:v}.
This change is made to ensure that the projected polytope $P$
has its largest contained and smallest containing
$\ell_\infty$-ball be of similar sizes.

\begin{sidefigure}
    \centering
    \includegraphics[width=0.5\textwidth]{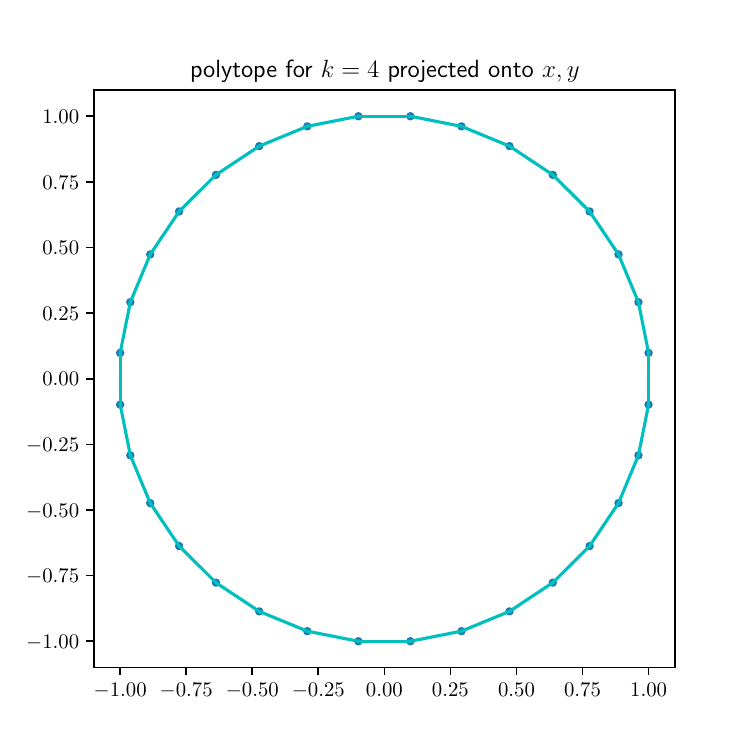}
    \caption{Vertices of the projected auxiliary polytope $\pi_W(P)$ (see \eqref{eq:extension}) without perturbation for $k=4$.}\label{fig:unperturbed}
\end{sidefigure}

\subsection{Projected Auxiliary Polytope Approximates Two-Dimensional Unit Disk}\label{sec:lb-radius-proj}

In this subsection, we will show that the polytope $P$ we constructed in
\eqref{eq:extension} has a projection $\pi_W(P)$ which approximates
the two-dimensional unit disk $\ball^2 = \{x, y \in \R: x^2 + y^2 \leq 1\}$
within exponentially small error:

\begin{lemma}[Projected auxiliary polytope approximates the two-dimensional disk]\label{lem:lb-radius-piwP}
For any $k \in \N$,
let $P \subset \R^{k+5}$ be the polytope defined by the linear system \eqref{eq:extension} with variables $x, y, s \in \R, p_0 \in \R^2, t \in \R^k$.
Let $W$ be the two-dimensional subspace spanned by the directions of $x$ and $y$. Then we have
\[
\ball^2 \subset \pi_W(P) \subset \cos(\pi/2^{k+2})^{-1}\ball^{2}.
\]
\end{lemma}

\Cref{lem:lb-radius-piwP} directly follows from the next two lemmas.
First, we show that the two-dimensional unit disk is contained in $\pi_W(P)$:

\begin{lemma}[Inner $\ell_2$-radius of the projected auxiliary polytope]\label{lem:innerdisk}
For every $x,y \in \R$ with $x^2+y^2 \leq 1$ there exist $p_0 \in \R^2$,
$t \in \R^k$ and $s \in \R$ such that $(x,y,p_0,t,s) \in P$.
\end{lemma}
\begin{proof}
    Suppose $x,y \in \mathbb{R}$ with $x^2 + y^2 \le 1$. We want to exhibit $p_0 \in \mathbb{R}^2$, $t \in \mathbb{R}^k$, and $s \in \mathbb{R}$ such that $(x, y, p_0, t, s)\in P$. We proceed as follows:

    First, set $s = 0$ and define $p_0 = (x,y)$. Then define $p_{i}$ inductively by the recurrence 
    \[
    p_{i} = \begin{bmatrix} v_{i}^\top \\  w_{i}^\top \end{bmatrix}^{-1} \begin{bmatrix}\abs{v_{i-1}^{\top}p_{i-1}}\\ w_{i-1}^{\top}p_{i-1} \end{bmatrix},
    \]
    with the base case being $p_{0} = (|x|, |y|)$. Finally, for each $i \in [k]$ define $t_{i} = v_{i}^{\top} p_{i} = |v_{i}^{\top}p_{i-1}|$.

    Then \eqref{cons:absvalue}, \eqref{cons:w}, \eqref{cons:v} and \eqref{cons:s} follow by definition.  It remains to argue that \eqref{cons:deathbarrier} and \eqref{cons:t} are satisfied. Since $v_{i}$ and $w_{i}$ are orthogonal and defined to have norm $1$, the matrix $\begin{bmatrix} v_{i}^\T \\ w_{i}^\T \end{bmatrix}$ is an isometry. Thus, it preserves norms, meaning that each of the $p_{i}$ have $\ell_2$-norm $\sqrt{x^2 + y^2}$.
    Since $t_i = |v_i^\top p_{i-1}| \leq \|v_i\|\cdot \|p_{i-1}\| \leq 1$,
    we know that \eqref{cons:t} is satisfied.
    Furthermore, we have $e_1^\top p_k \leq \|e_1\|\cdot\|p_k\| = \|e_1\| \cdot \sqrt{x^2 + y^2} \leq 1$,
    which ensures that \eqref{cons:deathbarrier} is satisfied.
\end{proof}

In the next lemma, we show that the projection $\pi_W(P)$ is contained in the two-dimensional disk $\cos(\pi/2^{k+2})^{-1}\ball^2$:

\begin{lemma}[Outer $\ell_2$-radius of the projected auxiliary polytope]\label{lem:outerdisk}
For every $x,y \in \R$ such that $\sqrt{x^2+y^2} \geq \cos(\pi/2^{k+2})^{-1}$ , there exist no
$p_0 \in \R^2, t \in \R^k$ and $s \in \R$ such that $(x,y,p_0,t,s) \in P$.
\end{lemma}
\begin{proof}
Fix any $(x, y) \in \R^2$ and $p_0 \in \R^2$, such that $\sqrt{x^2 + y^2} > \cos(\pi/2^{k+2})^{-1}$ and $p_0 \geq \begin{bmatrix}|x|, |y| \end{bmatrix}^\top$. 
Also fix any $p_1, \ldots, p_k \in \R^2$ satisfying \eqref{cons:w} and \eqref{cons:v}. We will show that such $p_1, \ldots, p_k$ would violate \eqref{cons:deathbarrier}, i.e., $e_1^\top p_k > 1$.  To simplify our notation, for all $i \in \{0, 1, \ldots, k\}$, let $(p_i)_v = v_i^\top p_i \in \R$ and $(p_i)_w = w_i^\top p_i \in \R$. Then $p_i = (p_i)_v v_i + (p_i)_w w_i$.

Notice that for all $i \in [k]$, the increment of the first coordinate from $p_{i-1}$ to $p_i$ is
\begin{align}\label{eq:increment-e1}
    e_1^\top p_i - e_1^\top p_{i-1} &= e_1^\top \left( (p_i)_w w_i + (p_i)_v v_i - (w_i^\top p_{i-1})w_i - (v_i^\top p_{i-1})v_i  \right) \nonumber \\
    &= e_1^\top \left( (p_i)_v v_i - (v_i^\top p_{i-1})v_i  \right) & (\text{By \eqref{cons:w}}) \nonumber \\
    &\geq e_1^\top v_i \left( |v_i^\top p_{i-1}| - v_i^\top p_{i-1} \right) & (\text{By $e_1^\top v_i > 0$ and \eqref{cons:v}}) \\
    &\geq 0 \nonumber
\end{align}
where the inequality in \eqref{eq:increment-e1} is tight when $v_i^\top p_i = |v_i^\top p_{i-1}|$.
Let $p_0^*, p_1^*, \ldots, p_k^*\in\R^2$ be the (unique) sequence defined by
\begin{align*}
    p_0^* &= p_0 \\
    w_i^\top p_i^* &= w_i^\top p_{i-1}^*, & \forall i\in [k] \tag{Tight for \eqref{cons:w}} \\
    v_i^\top p_i^* &= |v_i^\top p_{i-1}^*|, & \forall i\in [k] \tag{Tight for \eqref{cons:v}}
\end{align*}
Then $e_1^\top p_i - e_1^\top p_{i-1} \geq e_1^\top p_i^* - e_1^\top p_{i-1}^*$ for each $i \in [k]$. Also, notice that $e_1^\top p_0 = e_1^\top p_0^*$, therefore for each $i \in [k]$, $e_1^\top p_i \geq e_1^\top p_i^*$.

It remains to show that $e_1^\top p_k^* > 1$.
For all $i \in \{0, 1, \ldots, k\}$, let $\theta_i \in [-\pi, \pi]$ denote the angle between $p_i^* \in \R^2$ and $e_1$. Then since $e_1^\top p_0 \geq 0$ and $e_2^\top p_0 \geq 0$, we have $0 \leq \theta_0 \leq \frac{\pi}{2}$.
For any $i \in [k]$, notice that $p_i^*$ equals to $p_{i-1}^*$ (if $\theta_{i-1} \leq \frac{\pi}{2^{i+2}}$), or equals to the mirror of $p_{i-1}^*$ with respect to the line spanned by $w_i$ (if $\theta_{i-1} \geq \frac{\pi}{2^{i+2}}$). By induction, this gives 
\begin{align*}
    \|p_i^*\|_2 = \|p_{i-1}^*\|_2 = \ldots = \|p_0^*\|_2 = \|p_0\|_2,
\end{align*}
and
\begin{align*}
    \theta_{i} = \frac{\pi}{2^{i+2}} - |\theta_{i-1} - \frac{\pi}{2^{i+2}}| \leq \frac{\pi}{2^{i+2}}.
\end{align*}
Therefore, we get
\begin{align*}
    e_1^\top p_k^* = \|p_k^*\| \cdot \cos(\theta_k) \geq \|p_0\| \cdot \cos\left(\frac{\pi}{2^{k + 2}}\right) > 1.
\end{align*}
Thus we have shown $x^2 + y^2 > \cos(\pi/2^{k+2})^{-1}$ implies that $e_1^\T p_k \geq e_1^\T p_k^* > 1$
as desired.
\end{proof}

\subsection{Inner and Outer \texorpdfstring{$\ell_{\infty}$}{l\_infty}-Radius of the Auxiliary Polytope}\label{sec:lb-radius-P}

In this subsection, we will show that the auxiliary polytope $P$ has large inner $\ell_\infty$-radius and small outer $\ell_\infty$-radius.

\begin{lemma}[Inner and Outer $\ell_\infty$-Radius of the Auxiliary Polytope]\label{lem:lb-radius-P}
For $k \in \N$, let $P \subset \R^{k+5}$ be the polytope defined by the linear system \eqref{eq:extension}.
% Then there exists a point $(\bar x, \bar y, \bar{p}_0, \bar t, \bar{s})$ such that
Then for or $\bar x = \bar y = 0, \bar{p}_0 = (1/6,1/6)^\top, \bar t  = \mathbf{1}_{k}/30, \bar{s} = \frac{1}{3}$, it holds that
\[
    \frac{1}{30}\cdot \mathbb{B}_\infty^{k+5} \subset P - (\bar x, \bar y, \bar{p}_0, \bar t, \bar{s}) \subset \frac 3 2 \mathbb{B}_\infty^{k+5}.
\]
\end{lemma}

The following lemma is the key to show Lemma~\ref{lem:lb-radius-P}, where we construct a point in $B$ such that the $\ell_\infty$ ball with radius $\frac{1}{30}$ centered at that point is contained in $P$.
Obtaining a large inner ball is relatively easy using linear programming.
In this exposition we demonstrate the large inner ball by hand, which will require a somewhat tedious calculation to 3 significant digits.

\begin{lemma}[Inner $\ell_\infty$-Radius of the Polytope]\label{lem:innerball}
For $\bar x = \bar y = 0, \bar{p}_0 = (1/6,1/6)^\top, \bar t  = \mathbf{1}_{k}/30, \bar{s} = \frac{1}{3}$,
we have $(\bar x, \bar y, \bar{p}_0, \bar t, \bar{s}) + r\cdot \mathbb B_\infty^{k+5} \subseteq P$ for $r = \frac{1}{30}$.
\end{lemma}

\begin{proof}
    Fix any $(x, y, p_0, t, s) \in \R^{k + 5}$ such that $\|(x - \bar{x}, y - \bar{y}, p_0 - \bar{p}_0, t - \bar{t}, s - \bar{s})\|_\infty \leq r$. 
    We set $p_1, \ldots, p_k \in \R^2$ inductively by \eqref{cons:w} and the equations in \eqref{cons:v}, i.e.,
    $w_i^\top p_i = w_i^\top p_{i-1} $ and $t_i + is = v_i^\top p_i $ start from the base case of $p_0, s$ and $t$. 
    
    We will show that $(x, y, p_0, t, s) \in P$ by verifying the conditions in (\ref{cons:absvalue} - \ref{cons:s}).
    To simplify our notation, we define $(p_i)_v = v_i^\top p_i \in \R$ and $(p_i)_w = w_i^\top p_i \in \R$ for all $i \in \{0, 1, \ldots, k\}$. Note that $v_i$ and $w_i$ are orthogonal, $p_i = (p_i)_v v_i + (p_i)_w w_i$.
    
    First, observe that $e_1^\top p_0 \geq \frac{1}{6} - r \geq r \geq |x|$ and $e_2^\top p_0 \geq \frac{1}{6} - r \geq r \geq |y|$, confirming that \eqref{cons:absvalue} holds.
    Also, notice that $t_i \in [\bar{t}_i - r, \bar{t_i} + r] \subset [0, 1]$
    % $t_i \leq \bar{t}_i + r = \frac{i}{3} + r \leq k$
     and $s \in [\bar{s} - r, \bar{s} + r] \subset [0, 1]$. Thus \eqref{cons:t} and \eqref{cons:s} hold. 
    The equality constraint \eqref{cons:w} holds directly by definition of $p_1,\ldots, p_k$.
    It remains to show \eqref{cons:v} and \eqref{cons:deathbarrier}, i.e., $ (p_i)_v \geq |v_i^\top p_{i-1}|$ for all $i \in [k]$ and $e_1^\top p_k \leq 1$.
    
    To aid in the remaining steps of the proof, we show the claim that $(p_i)_w \geq 0$
    for all $i \in \{0,1,\ldots,k\}$ by induction. Observe that
    $w_0^\top p_0 \geq w_0^\top \bar{p}_0 - \|w_0\|\cdot\|p_0 - \bar p_0\| \geq \frac{\sqrt{2}}{6} - \sqrt 2 r \geq 0$.
    Also, for all $i \in [k]$,
    \begin{align}
        (p_i)_w &= w_i^\top p_{i-1} \tag{By \eqref{cons:w}} \nonumber \\
        &= w_i^\top \left( (p_{i-1})_w w_{i-1} + (p_{i-1})_v v_{i-1} \right) \nonumber \\
        &= (p_{i-1})_w w_i^\top w_{i-1} + (p_{i-1})_v w_i^\top v_{i-1} \nonumber \\
        &= (p_{i-1})_w \cdot \cos(\frac{\pi}{2^{i+2}}) + (p_{i-1})_v \cdot \sin(\frac{\pi}{2^{i+2}}) \label{eq:innerradius-cons-w} \\
        &\geq (p_{i-1})_w \cdot \cos(\frac{\pi}{2^{i+2}}), \tag{By $(p_{i-1})_v = t_{i-1} + (i-1)s \geq 0$} \nonumber
    \end{align}
    where \eqref{eq:innerradius-cons-w} follows from $w_{i}^\top w_{i-1} = \|w_{i}\|\|w_{i-1}\| \cos(\theta) = \cos(\theta)$ where $\theta = \pi / 2^{i+2}$ is the angle between $w_i$ and $w_{i-1}$; similarly,  we have $w_{i}^\top v_{i-1} = \|w_{i}\|\|v_{i-1}\| \cos(\pi/2 - \theta) = \sin(\theta) = \sin (\pi / 2^{i+2})$. 
    It then follows by induction that $(p_i)_w \geq 0$ for all $i \in \{0,1,\ldots,k\}$.
    
    \paragraph{Verify \eqref{cons:v}:}
    To verify the inequality listed in \eqref{cons:v}, i.e.,
    $(p_i)_v \geq |v_i^\top p_{i-1}|$ for all $i \in [k]$, we upper-bound the right-hand-side by expanding it into the inner product with $v_{i-1}$ and the inner product with $w_{i-1}$.
    Notice that for all $i \in [k]$,
    \begin{align}\label{eq:innerradius-cond-v}
        |v_i^\top p_{i-1}| &= \left|v_i^\top ((p_{i-1})_v v_{i-1} + (p_{i-1})_w w_{i-1}) \right| \nonumber \\
        &\leq |(p_{i-1})_v| v_i^\top v_{i-1} + |(p_{i-1})_w| \cdot |v_i^\top w_{i-1}| & \text{(Triangle inequality)} \nonumber \\
        &= (t_{i-1} + (i-1)s) \cdot \cos(\frac{\pi}{2^{i + 2}}) + (p_{i-1})_w \cdot \sin(\frac{\pi}{2^{i+2}}) & \text{(By \eqref{cons:v} and $(p_{i-1})_w \geq 0$)} 
        % &\leq {\color{red}(\bar{t}_{i-1} + (i-1)\bar{s} +  ir)} \cdot \cos(\frac{\pi}{2^{i + 2}}) + (p_{i-1})_w \cdot \sin(\frac{\pi}{2^{i+2}}).
    \end{align}
    Next, we require an upper bound on $(p_{i-1})_w$. For all $i\in [k]$, from \eqref{eq:innerradius-cons-w}
    \begin{align*}
        (p_i)_w &= (p_{i-1})_w \cdot \cos(\frac{\pi}{2^{i+2}}) + (p_{i-1})_v \cdot \sin(\frac{\pi}{2^{i+2}}) \\
        &\leq (p_{i-1})_w + (t_{i-1} + (i-1)s) \cdot \sin(\frac{\pi}{2^{i+2}}) \tag{By $(p_{i-1})_w \geq 0$ and \eqref{cons:v}} 
    \end{align*}
    Let $t_0 = v_0^\top p_0$ and $\bar{t}_0 = v_0^\top \bar{p}_0 = 0$.
    By applying the above inequality to $1, 2, \cdots, i-1$,
    we have
    \begin{align}
        (p_{i})_w &\leq w_0^\top p_0 + \sum_{j=0}^{i-1} (t_j + js) \cdot \sin(\frac{\pi}{2^{j+3}}) \nonumber \\
        &\leq w_0^\top p_0 + \sum_{j=0}^{i-1} (t_j + js) \cdot \sin(\frac{\pi}{8}) \cdot 1.9^{-j}  \nonumber \\
        &\leq w_0^\top p_0 + \sum_{j=0}^{i-1} (\bar{t}_j + r + js) \cdot \sin(\frac{\pi}{8}) \cdot 1.9^{-j} \nonumber \\
        &\leq (\frac{\sqrt{2}}{6} + \sqrt{2} r) + \sin(\frac{\pi}{8}) \sum_{j=0}^{i-1}  (\frac{1}{30} + r + js) \cdot 1.9^{-j} \tag{By $w_0^\top \bar{p}_0 = \frac{\sqrt{2}}{6}$, $\bar{t}_j = \frac{1}{30}$} \nonumber \\
        % &\leq \frac{\sqrt{2}}{6} + \sin(\frac{\pi}{8}) \cdot 0.7819 + r \cdot (1 + \sin(\frac{\pi}{8}) \cdot 2.1111)
        &\leq \frac{\sqrt{2}}{6} + \sin(\frac{\pi}{8}) \cdot 
    \frac{1}{30} \cdot \sum_{j=0}^{\infty} 1.9^{-j} + r\left(\sqrt{2} + \sin(\frac{\pi}{8})\sum_{j=0}^{\infty} 1.9^{-j}\right) +  s\left(\sin(\frac{\pi}{8}) \sum_{j=0}^{\infty} j(1.9)^{-j} \right) \nonumber \\
        &\leq 0.263 + 2.226r + 0.898s. \label{eq:innerradius-piw-ub}
    \end{align}
    Here the second inequality uses the fact that for every $x \in [0, \frac{\pi}{8}]$ one has $\sin(x)/1.9 \geq \sin(x/2)$.
    This is because by the half angle formula, $\sin(x/2) = \pm\sqrt{\frac{1-\cos(x)}{2}}$, thus
    \[
        \frac{\sin(x)^2}{\sin(x/2)^2} = \frac{1-\cos(x)^2}{(1-\cos(x))/2} = 2 + 2\cos(x) \geq 2 + 2\cos(\pi/8) \geq 1.9^2.
    \]

    Plugging \eqref{eq:innerradius-piw-ub} back into \eqref{eq:innerradius-cond-v}, we have for all $i \in [k]$,
    \begin{align*}
        |v_i^\top p_{i-1}| &\leq (t_{i-1} + (i-1)s) \cdot \cos(\frac{\pi}{2^{i + 2}}) +  (0.263 + 2.226r + 0.898s) \cdot \sin(\frac{\pi}{2^{i+2}}) \\
        &\leq ( \frac{1}{30} + r + (i-1)s) + ( 0.263 + 2.226r + 0.898s) \cdot \sin(\frac{\pi}{8}) \tag{By $\bar{t}_{i-1} = \frac{1}{30}$ and $i \geq 1$} \\
        &\leq 0.134 + 1.852r + (i - 0.656)s \\
        &\leq 0.134 + 1.852r + is - 0.656(\bar{s} - r)  \\
        &\leq is \tag{By $\bar{s} = \frac{1}{3}, r\leq \frac{1}{30}$} \\
        &\leq (p_i)_v.
    \end{align*}
    Therefore, $v_i^\top p_i \geq |v_i^\top p_{i-1}|$ for all $i \in [k]$ and \eqref{cons:v} holds.
    
    \paragraph{Verify \eqref{cons:deathbarrier}:}
    To verify \eqref{cons:deathbarrier}, notice that increment of the first coordinate from $p_{i-1}$ to $p_i$ is
    \begin{align}
        e_1^\top p_i - e_1^\top p_{i-1} &= e_1^\top \left( (p_i)_v v_i - (v_i^\top p_{i-1})v_i  \right) \tag{By $w_i^\top p_i = w_i^\top p_{i-1}$} \nonumber \\
        &= \sin(\frac{\pi}{2^{i + 2}}) \cdot (t_i + is - v_i^\top p_{i-1}) \tag{By $e_1^\top v_i = \sin(\frac{\pi}{2^{i + 2}})$ and \eqref{cons:v}} \nonumber \\
        &= \sin(\frac{\pi}{2^{i + 2}}) \cdot \left(t_i + is - v_i^\top ((p_{i-1})_w w_{i-1} + (p_{i-1})_v v_{i-1}) \right) \nonumber \\
        &= \sin(\frac{\pi}{2^{i + 2}}) \cdot \left(t_i + is + (p_{i-1})_w \cdot \sin(\frac{\pi}{2^{i + 2}}) - (t_{i-1} + (i-1)s) \cdot \cos(\frac{\pi}{2^{i + 2}})  \right)  \label{eq:innerball-1}
    \end{align}
    where the last step comes from $v_i^\top v_{i-1} = \cos(\frac{\pi}{2^{i + 2}})$ and $v_i^\top w_{i-1} = -\sin(\frac{\pi}{2^{i + 2}})$.
    For all $i \geq 2$, we can show that in \eqref{eq:innerball-1}, the third term in the brackets is at most the fourth term, thus the right-hand-side of \eqref{eq:innerball-1} is at most $\sin(\pi/2^{i+2}) \cdot (t_i + is)$:
    \begin{align*}
        (p_{i-1})_w \cdot \sin(\frac{\pi}{2^{i + 2}}) &\leq  (0.263 + 2.226r + 0.898s) \cdot \sin(\frac{\pi}{2^{i+2}}) \tag{By \eqref{eq:innerradius-piw-ub}}\\
        &\leq (0.263 + 2.226r + 0.898s)\cdot \tan(\frac{\pi}{2^{i+2}}) \cdot \cos(\frac{\pi}{2^{i+2}})  \\
        &\leq (0.263 + 2.226r + 0.898s)\cdot \tan(\frac{\pi}{8}) \cdot \cos(\frac{\pi}{2^{i+2}}) \tag{By $\frac{\pi}{2^{i+2}} \leq \frac{\pi}{8}$}  \\
        &\leq 0.277 \cdot \cos(\frac{\pi}{2^{i+2}})  \tag{By $r \leq \frac{1}{30}$ and $s \leq \bar{s} + r = \frac{11}{30}$} \\
        &\leq (\bar{t}_{i-1} - r + (i-1)s)  \cdot \cos(\frac{\pi}{2^{i+2}}) \tag{By $\bar{t}_{i-1}-r = 0$ and $(i-1)s \geq s \geq \bar{s} - r =0.3$}  \\
        &\leq (t_{i-1} + (i-1)s) \cdot \cos(\frac{\pi}{2^{i+2}}).
    \end{align*}
    Plugging back into \eqref{eq:innerball-1}, we have for all $i\geq 2$ that
    \begin{align*}
        e_1^\top p_i - e_1^\top p_{i-1} &\leq \sin(\frac{\pi}{2^{i + 2}}) \cdot (t_i + is) \\
        &\leq \sin(\frac{\pi}{2^{i + 2}}) \cdot (\bar{t}_i + i\bar{s} + (i+1)r) \\
        % &\leq \sin(\frac{\pi}{2^{i + 2}}) \cdot t_i \tag{By $w_{i-1}^\top p_{i-1} \leq t_{i-1} \cdot \cos(\frac{\pi}{2^{i+2}})$}\\
        &\leq \sin(\frac{\pi}{8}) \cdot 1.9^{-(i-1)} \cdot (\frac{1}{30} + \frac{i}{3} + (i+1)r) \tag{By $\sin(\frac{\pi}{2^{i + 2}}) \leq \sin(\frac{\pi}{8}) \cdot 1.9^{-(i-1)}$}
    \end{align*}
    For $i=1$ we recall from above that $(p_{i-1})_w \cdot \sin(\frac{\pi}{2^{i + 2}}) \leq 0.277 \cdot \cos(\frac{\pi}{2^{i+2}})$.
    Now take \eqref{eq:innerball-1} and observe
    $e_1^\top p_1 - e_1^\top p_0 = \sin(\pi/8)(t_1 + s + (p_0)_w \cdot\sin(\pi/8) - t_0\cdot\cos(\pi/8))
    \leq \sin(\pi/8)(11/30 + 2r + 0.277 \cdot \cos(\pi/8))$.
    Hence we find that $e_1^\top p_1 - e_1^\top p_0 \leq 0.239 + 0.766r$    
    and therefore,
    \begin{align*}
        e_1^\top p_k &\leq e_1^\top p_0 + \sum_{i=1}^k (e_1^\top p_i - e_1^\top p_{i-1}) \\
        &\leq (\frac{1}{6} + r) + 0.239 + 0.766r + \sin(\frac{\pi}{8}) \cdot \sum_{i=2}^k 1.9^{-(i-1)} \cdot (\frac{1}{30} + \frac{i}{3} + (i+1)r) \\
        &\leq (\frac{1}{6} + r) + 0.239 + 0.766r + 0.456 + 1.749r \\
        &\leq 0.862 + 3.515 r \leq 1.
    \end{align*}
    where the last inequality holds for any $r \leq \frac{1}{30}$.
    Therefore \eqref{cons:deathbarrier} holds and $(x, y, p_0, t, s) \in P$.
\end{proof}

\begin{proof}[Proof of Lemma~\ref{lem:lb-radius-P}]
In \Cref{lem:innerball} we construct a point $(\bar x, \bar y, \bar p_0, \bar t, \bar s)$
such that $\frac{1}{30}\mathbb{B}_\infty^{k+5} \subset P - (\bar x, \bar y, \bar{p}_0, \bar t, \bar{s})$.

For the second inclusion, we show that $P \subset (\cos(\pi/2^{k+2})^{-1}\cdot \mathbb{B}_\infty^{k+5}$.
Suppose $(x,y,p_0,t,s) \in P$ is arbitrary. From \Cref{lem:outerdisk} we know
that $\|(x,y)\|_\infty \leq \sqrt{x^2 + y^2} \leq \|p_0\|_2 \leq \cos(\pi/2^{k+2})^{-1}$.
Since $\mathbf{0}_k \leq t \leq \mathbf{1}_k$ we get $\|t\|_\infty \leq 1$, and
lastly we have $0 \leq s \leq 1$.
Put together, we find that $\|(x,y,p_0,t,s)\|_\infty \leq \cos(\pi/2^{k+2})^{-1}$.
Since $(x,y,p_0,t,s) \in P$ was arbitrary, we find that $P \subset \cos(\pi/2^{k+2})^{-1}\cdot \mathbb{B}_\infty^{k+5}$.
By the triangle inequality we find that
$P - (\bar x, \bar y, \bar{p}_0, \bar t, \bar{s}) \subset (\cos(\pi/2^{k+2})^{-1} + \|(\bar x, \bar y, \bar p_0, \bar t, \bar s)\|_\infty)\cdot \mathbb{B}_\infty^{k+5}$
and we see in \Cref{lem:innerball} that $\|(\bar x, \bar y, \bar p_0, \bar t, \bar s)\|_\infty = 1/3$.
Finally, note that $\cos(\pi/2^{k+2})^{-1} \leq \cos(\pi/8)^{-1} \leq 1.1$, we have 
\[
    \cos(\frac{\pi}{2^{k+2}})^{-1} \leq \cos(\frac{\pi}{8})^{-1} + \frac{1}{3} \leq \frac{3}{2}.\qedhere
\]
\end{proof}

\subsection{Properties of the Polar Polytope}\label{sec:lb-dual}

In this section, we will analyse the scaled polar polytope
$Q = \frac{1}{30}(P - (\bar x, \bar y,\bar p_0, \bar t, \bar s))^\circ$.
From well-known duality properties, we will find that $Q$ satisfies the following desirable properties:
\begin{enumerate}
    \item $Q \cap W$ approximates a two-dimensional disk;
    \item The inner $\ell_1$-radius of $Q$ is at least $\frac{1}{45}$ when centered at $\mathbf{0}$;
    \item The outer $\ell_1$-radius of $Q$ is at most $1$ when centered at $\mathbf{0}$.
    % $\max_{i \in [n]} \|a_i\| \leq 1$.
\end{enumerate}

\begin{lemma}\label{lem:lb-prim-to-dual}
For any $k \in \N$, there exists a two-dimensional linear subspace $W \subset \R^{k+5}$ and $n = 4k+7$ points $\bar a_1,\ldots,\bar a_n \in \mathbb{B}_1^{k+5}$ such that $Q := \conv(\bar a_1,\ldots,\bar a_n)$ satisfies
\[
\frac{\cos(\pi/2^{k+2})}{30} \cdot \ball^{k+5} \cap W \subset Q \cap W \subset \frac{1}{30}\cdot \ball^{k+5}\cap W
\]
and
\[
    \frac{1}{45} \cdot \mathbb{B}_1^{k+5} \subset Q \subset \mathbb{B}_1^{k+5}
\]
\end{lemma}
\begin{proof}
Let $P \subset \R^{k+5}$ be the polytope defined by the linear system in \eqref{eq:extension}, and let 
\begin{align*}
    \tilde{P} = P - (\bar x, \bar y, \bar{p}_0, \bar t, \bar{s})
\end{align*}
denote the polytope obtained from shifting its center $(\bar x, \bar y, \bar{p}_0, \bar t, \bar{s})$ to $\mathbf{0}_d$.
Here, as in \Cref{lem:innerball},
\begin{align*}
    \bar x = \bar y = 0, \bar{p}_0 = (1/6,1/6)^\top, \bar t  = \mathbf{1}_{k}/30, \bar{s} = \frac{1}{3}.
\end{align*}
By applying row rescalings we transform the constraints
(\ref{cons:absvalue} - \ref{cons:s})
into a matrix $A \in \R^{(4k+7)\times(k+5)}$
such that
\[
\tilde{P} = \{z \in \R^{k+5} : Az \leq 1\}.
\]
Let $\tilde{Q} = (\tilde{P})^\circ \subset \R^d$ denote the polar body of $\tilde{P}$.
Since $\tilde{P}$ is bounded, $\tilde{Q}$ is the convex hull of the rows of the matrix $A$, i.e.
\[
\tilde{Q} = \{A^\top\lambda : \lambda \in [0,1]^{4k+7}\; \mbox{s.t.} \sum_{i=1}^{4k+7} \lambda_i = 1\}.
\]
Then by \Cref{lem:lb-radius-P} and \Cref{fact:dual-containing}, the inner and outer ball of $\tilde{Q}$ satisfy
\[
    \frac{2}{3}\cdot \mathbb{B}_1^{k+5} \subset \tilde{Q} \subset 30\cdot \mathbb{B}_1^{k+5}.
\]
Also, by \Cref{lem:lb-radius-piwP}, \Cref{fact:dual-intersection} and \Cref{fact:dual-containing}, the inner and outer ball of $\tilde{Q} \cap W$ satisfy
\[
\cos(\pi/2^{k+2}) \cdot \ball^{k+5}\cap W \subset \tilde{Q} \cap W \subset \ball^{k+5} \cap W.
\]
The lemma then follows from taking $Q = \frac{1}{30}\tilde{Q}$.
\end{proof}

\subsection{Perturbation Analysis and Proof of the Lower Bound}\label{sec:lb-perturbation}

In this subsection, we study the number of edges of the intersection polygon $Q \cap W$ after perturbation and prove our main lower bound theorem (\Cref{thm:lb-main}). 
% We restate the theorem as follows:
% \begin{theorem}
%     For any $d > 5, n \geq 4d-13, \sigma \leq \frac{0.19}{\sqrt{\log n}}$, there exists a two-dimensional linear subspace $W \subset \R^{d}$ and vectors $\bar a_1,\ldots,\bar a_n \in \R^d$, $\max_{i \in n} \|\bar{a}_i\| \leq 1$ such that the following holds. 
%     Let $a_1, \ldots, a_n$ be independent Gaussian random variables where each $a_i \sim \mathcal{N}_d(\bar{a}_i, \sigma^2 I)$, then with probability at least $1 - \binom{n}{d}^{-1}$,
%     \begin{align*} 
%         \edges(\conv(a_1,\ldots,a_n) \cap W) \geq \Omega(\min(\frac{1}{\sqrt{d\sigma\sqrt{\log n}}},2^d)).
%     \end{align*}
% \end{theorem}
To show that our construction has many edges even after perturbation,
we require the following two statements:
\begin{lemma}\label{lem:hdist}
Let $\bar a_1,\dots,\bar a_n \in \R^d$ be points with
$r \mathbb B ^d_1 \subseteq \conv(\bar a_1,\dots,\bar a_n)$ for some $r > 0$.
If $\eps \leq r/2$ and points $a_1,\dots, a_n \in \R^d$ satisfy $\|a_i - \bar a_i\|_1 \leq \eps$
for all $i \in [n]$ then it follows that
\[
(1-\frac{2 \eps}{r}) \conv(\bar a_1,\dots,\bar a_n) \subseteq \conv(a_1,\dots,a_n)
\subseteq (1+\frac{\eps}{r}) \conv(\bar a_1,\dots,\bar a_n).
\]
\end{lemma}
\begin{proof}
Write $Q = \conv(a_1,\dots,a_n)$ and $\bar Q = \conv(\bar a_1,\dots,\bar a_n)$.
The second inclusion follows by $Q \subseteq \bar Q + \eps \mathbb B_1^d \subseteq \bar Q + \frac{\eps}{r}\bar Q$.
For the first inclusion, we observe that
$r \mathbb B_1^d \subseteq \bar Q \subseteq Q + \eps \mathbb B_1^d \subseteq Q + \frac{r}{2} \mathbb B_1^d$.
This implies that $\frac r 2 \mathbb B_1^d \subseteq Q$, for
if there were to exist $x \in \frac r 2 \mathbb B^d_1$ such that $x \notin Q$ then,
since $Q$ is closed and convex, then we could find $y \in \R^d$ such that
$y^\T x > y^\T z$ for all $z \in Q$. Writing $f(S) = \max_{z \in S} y^\T z$
for $S \subseteq \R^d$, this would give
\[
f(r\mathbb B_1^d) \geq f(x + \frac r 2 \mathbb B_1^d)
= y^\T x + f(\frac r 2 \mathbb B_1^d)
> f(Q) + f(\frac r 2 \mathbb B_1^d)
\geq f(Q + \frac r 2 \mathbb B_1^d)
\geq f(r \mathbb B_1^d).
\]
By contradiction it follows that $\frac r 2 \mathbb B_1^d \subseteq Q$.

Note that $1-x^{2} \leq 1$, so $1-x \leq 1/(1+x)$ for any $1+x > 0$. Combining with $\bar Q \subseteq Q + \eps\mathbb B_1^d \subseteq Q + \frac{\eps}{r/2} Q$ we conclude the desired result.
\end{proof}

% Further, we need a way to lower bound the number of edges. This is provided by the following lemma:
\begin{lemma}\label{lem:edgecount}
If a polygon $T \subset \R^2$ satisfies $\alpha\cdot \ball^2 \subset T \subset \beta\cdot \ball^2$
for some $\alpha, \beta > 0$ then $T$ has at least $\sqrt{\alpha/(\beta - \alpha)}$ edges.
\end{lemma}
\begin{proof}
If $\beta > 2\alpha$ then the bound is trivially true,
so assume that $\beta \leq 2\alpha$.

Without loss of generality, re-scale $T$ so that $\ball^2 \subset T \subset (1+\eps)\cdot \ball^2$,
where $\eps = \beta/\alpha - 1 > 0$. Note that since $\beta \leq 2 \alpha$ we have $\eps \leq 1$.

Consider any edge $[q_1, q_2] \subset T$ and let $p \in [q_1,q_2]$ denote
the minimum-norm point in this edge.
Then we have $\|q_1 - p\|^2 = \|q_1\|^2 + \|p\|^2 - 2\langle q_1,p\rangle$.
Since $p$ is the minimum-norm point, we have $\langle q_1, p \rangle \geq \|p\|^2$,
and hence $\|q_1 - p\|^2 \leq \|q_1\|^2 - \|p\|^2 \leq (1+\eps)^2 - \|p\|^2$.
Since $p$ lies on the boundary of $T$ we have $\|p\| \geq 1$,
which implies that $\|q_1-p\|^2 \leq (1+\eps)^2 - 1 = 2\eps + \eps^2$.
The analogous argument for $\|q_2-p\|$ and the triangle inequality tell us that
$\|q_1-q_2\| \leq 2\sqrt{2\eps + \eps^2} \leq 4\sqrt{\eps}$, where we use $\eps \leq 1$ at the second step. The choice of the edge $[q_1,q_2]$ was
arbitrary, hence every edge of $T$ has length as most $4\sqrt{\eps}$.

But $T$ has perimeter at least $2\pi$. Since the perimeter of a polygon
is equal to the sum of the lengths of its edges, this implies that
$T$ has at least $\frac{2\pi}{4\sqrt{\eps}} > \frac{1}{\sqrt \eps}$ edges.
\end{proof}

Now, we can prove the generic lower bound claimed in \Cref{thm:lb-generic} on the shadow size
under adversarial $\ell_1$-perturbations.

\begin{proof}[Proof of Theorem~\ref{thm:lb-generic}]
Fix any $d > 5$, let $k = d-5$ and observe that $n = 4k+7$. Let $\bar a_1,\ldots,\bar a_{n}$ be as constructed in \Cref{lem:lb-prim-to-dual}. Then we have
\[
\frac{\cos(\pi/2^{k+2})}{30} \cdot \ball^{k+5} \cap W \subset \conv(\bar{a}_1, \ldots, \bar{a}_n) \cap W \subset \frac{1}{30}\cdot \ball^{k+5}\cap W ,
\]
and
\begin{align*}
    \frac{1}{45} \cdot \mathbb B_1^{k+5} \subset \conv(\bar{a}_1, \ldots, \bar{a}_n) \subset \mathbb B_1^{k+5} .
\end{align*}
For any set of points $a_1, \ldots, a_n$ such that $\|a_i - \bar{a}_i\|_1 \leq \eps$ for each $i \in [n]$. By \Cref{lem:hdist} and setting $r = 1/45$, we have
\[
    \conv(\bar{a}_1, \ldots, \bar{a}_n) \subseteq \frac{1}{1 - 2\eps / r} \conv(a_1, \ldots, a_n) = \frac{1}{1 - 90 \eps} \conv(a_1,\ldots,a_n).
\]
Therefore,
\[
    \frac{\cos(\pi/2^{k+2})}{30} \cdot \ball^{k+5} \cap W \subset \conv(\bar a_1,\ldots,\bar a_n) \cap W \subset (1 - 90\eps)^{-1} \conv(a_1,\ldots,a_n)\cap W ,
\]
and 
\[
    \conv(a_1,\ldots,a_n) \cap W \subset (1+45\eps) \conv(\bar a_1,\ldots,\bar a_n) \cap W \subset \frac{1+45\eps }{30}\cdot \ball^{k+5}\cap W.
\]
Therefore, we can bound the inner and outer $\ell_1$-radius of $\conv(a_1, \ldots, a_n) \cap W$ by 
\begin{align*}
    \frac{(1 - 90\eps) \cdot \cos(\pi/2^{k+2})}{30} \cdot \ball^{k+5}\cap W \subseteq \conv(a_1, \ldots, a_n) \cap W \subseteq \cdot \frac{1+45\eps}{30} \cdot \ball^{k+5}\cap W .
\end{align*}
It then follows from \Cref{lem:edgecount} that the polygon $\conv(a_1,\ldots,a_n) \cap W$ has at least
\begin{align*}
    & \left( \frac{(1 - 90\eps) \cdot \cos(\pi/2^{k+2})}{(1+45\eps) - (1 - 90\eps) \cdot \cos(\pi/2^{k+2})} \right)^{1/2}  \\
    \geq ~ & \left( \frac{(1 - 90\eps) \cdot (1 - \frac{\pi^2}{4^{k+2}})}{(1+45\eps) - (1 - 90\eps) \cdot (1 - \frac{\pi^2}{4^{k+2}})} \right)^{1/2} \tag{By $\cos(x) \geq 1 - x^2$ for $x \leq \frac{\pi}{2}$} \\
    \geq ~ & \left( \frac{1 - 90\eps - \frac{\pi^2}{4^{k+2}}}{(1+45\eps) - (1 - 90\eps - \frac{\pi^2}{4^{k+2}})} \right)^{1/2}\\
    \geq ~ & \Omega \left(\frac{1}{\sqrt{\eps + 4^{-k}}} \right)
\end{align*}
edges.
\end{proof}

Finally, we can prove our main result using Gaussian tail bound:
\begin{proof}[Proof of Theorem~\ref{thm:lb-main}]
Using concentration of Gaussian distribution in \Cref{cor:gaussian-globaldiam}, we find that if $\sigma
\leq 1/(360 d\sqrt{\log n})$, then with probability at least
$1 - \binom{n}{d}^{-1}$, we have $\max_{i \in [n]} \|a_i - \bar{a}_i\|_2 \leq
4\sigma \sqrt{d \log n} \leq \frac{1}{90\sqrt{d}}$.
The result follows from \Cref{thm:lb-generic} and the fact that $\|x\|_1 \leq \sqrt d \|x\|_2$
for every $x \in \R^d$.
\end{proof}

\subsection{Experimental Results}\label{sub:experiments}

\begin{figure}[ht]
    \centering
    % this is experiments/32samples/replotted.pdf
    % the code generating this data can be found in the same directory
    \includegraphics[width=0.7\textwidth]{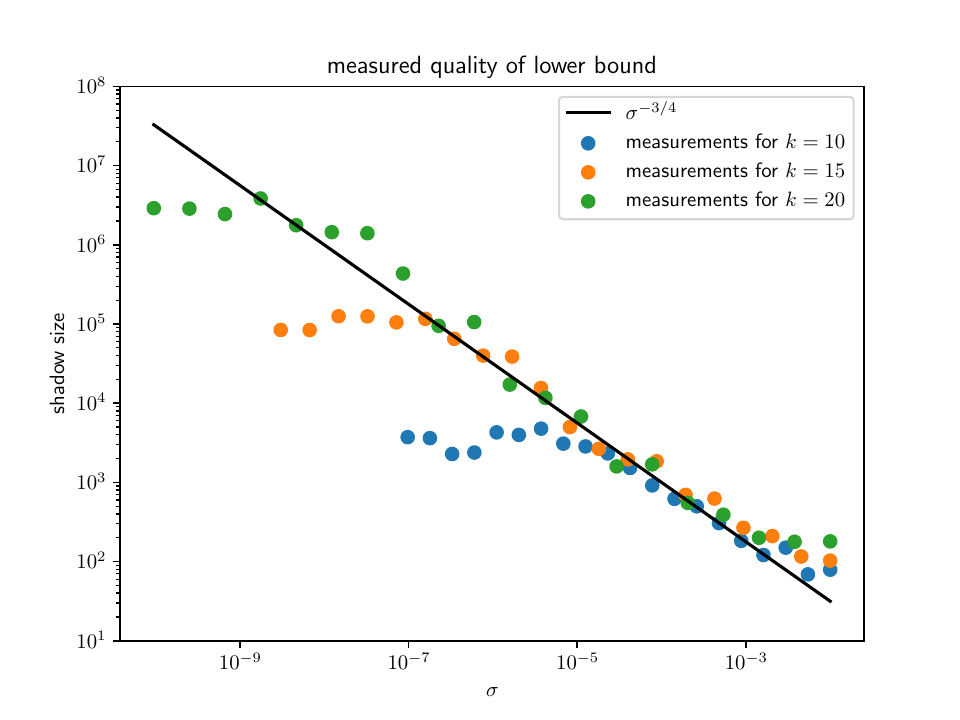}
    \caption{Measured shadow sizes for sampled perturbations of our construction, for different values of $k$ and $\sigma$.}
    \label{fig:measurements}
\end{figure}
To measure whether the analysis in \Cref{thm:lb-main} is tight or not, we ran numerical experiments.
Using \texttt{Python} and \texttt{Gurobi 10.0.3}, we constructed a matrix $A$ such that
\[
    P - (\bar x, \bar y,\bar p_0, \bar t, \bar s) = \{z \in \R^{k+5} : Az \leq \mathbf 1_{4k+7}\},
\]
as described earlier in this section.
Writing $R$ as the maximum Euclidean norm among the row vectors of $A$,
we sampled $\hat A$ with independent Gaussian distributed entries with standard deviation $\sigma R$ and $\E[\hat A] = A$.
To approximate the shadow size, we optimized the objective vectors
$\cos(\frac{(i+0.3)\pi}{2^{k+4}})x + \sin(\frac{(i+0.3)\pi}{2^{k+4}})y$, with $i = 0,\dots,2^{k+5}-1$,
over the polyhedron $\{z \in \R^{k+5} : \hat A z \leq 1\}$ and
counted the number of distinct values $(x, y)$ found among the solutions.
Two consecutive solutions were counted as distinct if their $x, y$ coordinates differed in $\ell_1$ norm
by at least $10^{-11}$.\footnote{The code used for this experiment is available at \url{https://github.com/sophiehuiberts/shadow-size/tree/main/HLZ23}}
Note that the number of vertices of this projection of the perturbed auxiliary polytope
is equal to the number of vertices of the section of the perturbed polar polytope,
which then describes the shadow size.

When $\sigma = 0$, our code found $2^{k+1}$ such distinct points.
For $\sigma > 0$, \Cref{thm:lb-main} shows that we expect to find at least
$\Omega \Big(\min \Big(\frac{1}{\sqrt{d\sigma\sqrt{\log d}}},2^k \Big) \Big)$ distinct pairs $(x,y)$.

For $k = 10,15,20$, we measured the shadow size for $20$ different values of $\sigma$ ranging from $0.01$ to $0.0001/2^k$.
The resulting data is depicted in \Cref{fig:measurements} along with a graph of the function $\sigma \mapsto \sigma^{-3/4}$.
We observe that for each $k$, the measured shadow size appears to follow the graphed function up to a point,
plateauing slightly above $2^{k+1}$ when $\sigma$ is small.
The fact that some measurements come out higher than $2^{k+1}$, the shadow size for $\sigma = 0$, is not unexpected:
the polytope $P$ is highly degenerate, whereas the perturbed polytope is simple and can thus have many more vertices.

The measured shadow sizes appear to grow much faster than $1/\sqrt\sigma$ as $\sigma$ gets small, closer to the $\sigma^{-3/4}$ line that we plotted.
These results suggest that the behaviour of the shadow size is
substantially different in $d = 2$, where
we have an upper bound of
$O\left(\frac{\sqrt[4]{\log(n)}}{\sqrt\sigma} + \sqrt{\log n}\right)$,
and $d > 3$, where one might expect a lower bound with a higher dependence on $\sigma$.

% \bibliographystyle{ACM-Reference-Format}
% \bibliography{ref}

%\addcontentsline{toc}{section}{References}
\printbibliography

\end{document}